\documentclass[11pt]{amsart}

\usepackage{fullpage}
\usepackage{amsmath}
\usepackage{amsthm}
\usepackage{amsfonts}
\usepackage{amssymb}

\newcommand{\mc}{\mathcal}

\newtheorem{lemma}{Lemma}[section]
\newtheorem{proposition}{Proposition}[section]
\newtheorem{theorem}{Theorem}[section]
\newtheorem{corollary}{Corollary}[section]
\theoremstyle{remark}
\newtheorem{remark}{Remark}[section]
\theoremstyle{definition}
\newtheorem{definition}{Definition}[section]

\title{A proof of Price's Law on
Schwarzschild black hole manifolds for all angular momenta}

\author{Roland Donninger}
\address{University of Chicago, Department of Mathematics,
5734 South University Avenue, Chicago, IL 60637, U.S.A.}
\email{donninger@uchicago.edu}
\thanks{The first author 
is an Erwin Schr\"odinger Fellow of the 
FWF (Austrian Science Fund) Project No. J2843 and 
he wants to thank Peter C. Aichelburg for his support. Furthermore, all 
three authors would like to thank Piotr Bizo\'n for a number of helpful
remarks on a first version of this paper.}

\author{Wilhelm Schlag}
\address{University of Chicago, Department of Mathematics,
5734 South University Avenue, Chicago, IL 60637, U.S.A.}
\email{schlag@math.uchicago.edu}
\thanks{The second author was partly supported by the National
Science Foundation DMS-0617854.}

\author{Avy Soffer}
\address{Rutgers University, Department of Mathematics, 110 Freylinghuysen Road, Piscataway, NJ 08854, U.S.A.}
\email{soffer@math.rutgers.edu}
\thanks{The third author wants to thank A. Ori and T. Damour for helpful 
discussions, the
IHES France for the invitation and the NSF DMS-0903651 for partial support.}

\begin{document}

\begin{abstract}
Price's Law states that linear perturbations of a Schwarzschild black hole
fall off as $t^{-2\ell-3}$ for $t \to \infty$ provided the initial data decay sufficiently fast at spatial infinity. Moreover, if the perturbations are initially static (i.e., their time derivative is zero), then the decay is predicted to be $t^{-2\ell-4}$. 
We give a proof of $t^{-2\ell-2}$ decay for general data in the form of weighted $L^1$ to
$L^\infty$ bounds for solutions of the Regge--Wheeler equation.
For initially static perturbations we obtain $t^{-2\ell-3}$.
The proof is based on an integral representation of the solution which follows
from self--adjoint spectral theory.
We apply two different perturbative arguments in order to construct the
corresponding spectral measure and 
the decay bounds are obtained by appropriate oscillatory integral 
estimates.
\end{abstract}

\maketitle

\section{Introduction and main result}

In General Relativity, the dynamics of spacetime is governed by 
Einstein's equation
which, in the absence of matter, takes the form
$$
R_{\mu\nu}(g) = 0
$$
where $R_{\mu \nu}(g)$ is the Ricci tensor of the Lorentz metric
$g$.  Exact solutions (i.e., solutions which are known in closed form)
include the free flat Minkowski spacetime as well as the Schwarzschild
metric and, more generally, the Kerr solution. The
Schwarzschild solution is spherically symmetric and corresponds to
a nonrotating black hole whereas rotating black holes are described by the
axially symmetric Kerr spacetime. 
A fundamental mathematical
problem in General Relativity is the understanding of the stability
of these solutions. The stability of the flat Minkowski spacetime under small
perturbations was shown in the seminal work of Christodoulou and
Klainerman~\cite{K-C} in the late 1980's. A simpler proof was later
developed by Lindblad and Rodnianski~\cite{L-R}.
However, we are very far from understanding the dynamics near a
black hole. Yet, latest experimental setups are crucially dependent
on such an analysis, in order to observe gravitational waves (see
for example \cite{Dam-N1}, \cite{Dam-N2}, \cite{Dam-N3}, \cite{Dam-N4}
and cited ref.).
Most efforts are now
focused on understanding the linear dynamics and stability of such
solutions, see e.g.~\cite{Hod2}, \cite{Bla-Burko} and cited ref., as
well as~\cite{Poi}. 
The mathematical aspects of the problem will be discussed below in more detail.
We also refer the reader to the survey \cite{D-Rod2} which
gives an excellent overview of recent developments in the field from the 
mathematical perspective.

\subsection{Wave evolution on the Schwarzschild manifold}

As a first approximation to the linear stability problem of a nonrotating black
hole one may consider the wave equation on a fixed Schwarzschild background.
One is then typically interested in decay estimates for the evolution. To simplify
things even more, one restricts the analysis to the exterior region of the black
hole which, however, is physically reasonable: 
such a model describes a black hole subject to a small external perturbation 
by a scalar field --- a
situation which, with a more realistic matter model, is certainly relevant
in an astrophysical context.
In order to formulate the problem we choose coordinates such that the
exterior region of
the black hole can be written as $(t,r,(\theta,\phi)) \in \mathbb{R} \times
( 2M, \infty) \times S^2$ with the metric
$$
g = -F(r)dt^2 + F(r)^{-1} dr^2 + r^2(d\theta ^2 +  \sin^2\theta d\phi^2)
$$
where $ F(r) = 1 - \frac{2M}{r}$ and, as usual, $M>0$ denotes the mass. 
We now introduce the well--known \emph{Regge--Wheeler
tortoise coordinate} $r_*$ which (up to an additive constant) 
is defined by the relation
$$
F = \frac{dr}{dr_*}.
$$
In this new coordinate system, the outer region is described by
$(t,r_*,(\theta,\phi)) \in \mathbb{R} \times \mathbb{R} \times S^2$,
\begin{equation}
\label{eq_sstortoise}
g = -F(r)dt^2 +F(r)dr_*^2 +r^2(d\theta^2 + \sin^2 \theta d\phi ^2)
\end{equation}
with $F$ as above and $r$ is now interpreted as a function of $r_*$. 
Explicitly,  $r_*$ is computed as
\begin{equation*}
r_* = r + 2M\log \left(\frac{r}{2M}-1\right).
\end{equation*}
Generally, the Laplace--Beltrami operator on a manifold with metric $g$
is given by
$$ \Box_g=\frac{1}{\sqrt{|\det (g_{\mu \nu})|}}\partial_\mu 
\left (\sqrt{|\det (g_{\mu \nu})|}g^{\mu \nu}\partial_\nu \right ) $$
and thus, for the metric $g$ in (\ref{eq_sstortoise}), we obtain
$$ \Box_g=F^{-1}\left (-\partial_t^2+\frac{1}{r^2}\partial_{r_*} \left (r^2
\partial_{r_*} \right ) \right )+\frac{1}{r^2}\Delta_{S^2}. $$
By setting $\psi(t,r_*,\theta, \phi)=r(r_*)\tilde{\psi}(t,r_*,\theta,\phi)$ 
and writing $x=r_*$, 
the wave equation $\Box_g
\tilde{\psi}=0$ is equivalent to
\begin{equation}
\label{eq_wavess}
-\partial_t^2 \psi+\partial_x^2
\psi-\frac{F}{r}\frac{dF}{dr}\psi+\frac{F}{r^2}\Delta_{S^2}\psi=0. 
\end{equation}

The mathematically rigorous analysis of this equation has been initiated by
Wald \cite{Wald1}, however,
the first complete proof of uniform boundedness of solutions is due to
Kay and Wald \cite{Kay-Wald}.
Recently, Dafermos and Rodnianski have found a more robust method to prove
boundedness of solutions
based on vector
field multipliers that capture the so--called red--shift effect \cite{D-Rod3},
see also \cite{D-Rod2} for a survey and
generalizations of these results.
The goal of our present work is to prove $L^1$ to $L^\infty$ decay
estimates for Eq.~(\ref{eq_wavess}).  Different types of decay
estimates have been proved before. Local decay estimates, based on
multipliers generalizing the Morawetz estimates, were initiated in
\cite{B-Sof1}, \cite{B-Sof2} and \cite{B-Sof3}. Later, a similar approach 
was used in
\cite{D-Rod}, \cite{B-St}, \cite{D-Rod3}, \cite{D-Rod4} and 
\cite{Luk} to prove both local decay
estimates and pointwise decay in time based on conformal type
identities. In \cite{Met-Tat1}, \cite{Met-Tat2}, \cite{Tat-Toh} and 
\cite{MMTT} it is
shown how to apply such estimates to obtain Strichartz type decay
estimates.
We also mention the recent work \cite{Andersson}.
After submission of the present paper, Tataru announced a proof of the sharp pointwise
$t^{-3}$ decay for general data without symmetry assumptions, see \cite{tataru09}.
Moreover, his result also applies to the more complicated case of rotating Kerr black holes.
In fact, in the follow--up paper \cite{DSS2} we also obtain pointwise $t^{-3}$ decay on Schwarzschild for general data. We will discuss this below in more detail.
Our results differ from the above in certain respects: the methods
we use are based on constructing the Green's function and deriving
the needed estimates on it. Previous works in this direction include mainly the series of
papers \cite{FKSY1}, \cite{FKSY2}, \cite{FKSY3} where the first pointwise decay result
for Kerr black holes has been proved, see also \cite{Kron} and \cite{Kron2} for Schwarzschild. 
In our approach, we freeze the
angular momentum $\ell$ or, in other words, we project onto a
spherical harmonic.
More precisely, let $Y_{\ell,m}$ be a spherical harmonic (that is, an eigenfunction of
the Laplacian on $S^2$ with eigenvalue $-\ell(\ell+1)$) and insert the 
ansatz $\psi(t,x,\theta,\phi)
=\psi_{\ell,m}(t,x)Y_{\ell,m}(\theta,\phi)$ in Eq.~(\ref{eq_wavess}).
This yields the \emph{Regge--Wheeler equation}
$$
\partial_t^2 \psi_{\ell,m}-\partial_x^2
\psi_{\ell,m}+V_{\ell,\sigma}(x)\psi_{\ell,m}=0
$$
with $\sigma=1$ 
where
$$ V_{\ell,\sigma}(x)=\left (1-\frac{2M}{r(x)} \right )
\left (\frac{\ell(\ell+1)}{r^2(x)}+\frac{2M\sigma}{r^3(x)} \right ) $$
is known as the \emph{Regge--Wheeler potential}.
In the present work, we obtain decay estimates for solutions of this equation.
However, before we explain our results in more detail, we further motivate 
the study of the Regge--Wheeler equation by considering more general black hole
perturbations. 

\subsection{Other types of black hole perturbations}
The wave equation on the Schwarzschild manifold describes the time evolution of
linearized scalar field perturbations of a black hole.
Of course, not all physically relevant situations are covered by this
simple model since it ignores the underlying tensorial structure altogether.
Eventually, one is interested in perturbing fields of higher spin, 
in particular gravitational perturbations.
However, as a remarkable fact, the Regge--Wheeler equation is also relevant in
this context.
This follows from a reduction procedure that goes back to Regge and Wheeler
\cite{RW} as well as Zerilli \cite{Zerilli}, see also 
\cite{Vish} and \cite{Chandra}. We
will briefly sketch how this comes about.
In order to study gravitational perturbations, one considers a 
perturbed Schwarzschild metric $\tilde{g}$ of 
the form
$$ \tilde{g}=-e^{2(\nu+\delta \nu)}dt^2+e^{2(\psi+\delta \psi)}
(d\phi-\delta \omega dt-\delta q_2 dr - \delta q_3 d\theta)^2
+e^{2(\mu_2+\delta \mu_2)}dr^2+e^{2(\mu_3+\delta \mu_3)}d\theta^2 $$
where the various coefficients are allowed to depend on $t,r,\theta$ and 
$e^{2\nu}=e^{-2\mu_2}=1-\frac{2M}{r}$, $e^{\mu_3}=r$, $e^\psi=r \sin
\theta$ (we follow the notation of \cite{Chandra}).
It can be shown (see \cite{Chandra}) that this ansatz is sufficiently general.
One then requires the metric $\tilde{g}$ to satisfy the linearized 
Einstein vacuum
equations, i.e., one linearizes $R_{\mu \nu}(\tilde{g})=0$ 
with respect to the perturbations $\delta \nu$, $\delta \psi$, etc.
It turns out that one has to distinguish between so--called \emph{axial} 
($\delta \omega, \delta
q_2, \delta q_3$) and \emph{polar} ($\delta \nu, \delta \psi, \delta \mu_1, \delta
\mu_2$) perturbations, depending on the behavior of the metric under the
reflection $\phi \mapsto -\phi$.
After a lengthy calculation and separation of the $\theta$--dependence one
arrives at
$$ \partial_t^2 \psi_\ell-\partial_x^2
\psi_\ell+\left (1-\frac{2M}{r(x)} \right )
\left (\frac{\ell(\ell+1)}{r^2(x)}-\frac{6M}{r^3(x)} \right )
\psi_\ell=0 $$
where $\psi_\ell$ is an auxiliary function which completely determines
the axial perturbations, see
\cite{Chandra} for details.
Thus, $\psi_\ell$ satisfies the Regge--Wheeler equation with $\sigma=-3$.
In the case of polar perturbations, Zerilli \cite{Zerilli} has derived an analogous 
equation with a more
complicated effective potential. However, Chandrasekhar \cite{Chandr2} (see also
\cite{Chandra}) has found a transformation involving differential operations
that relates this equation to the one for axial perturbations.
As a consequence, the Regge--Wheeler equation provides a fairly complete
description of gravitational perturbations with a fixed angular momentum
parameter $\ell$.
Moreover, we mention the fact that the Regge--Wheeler equation with parameter
$\sigma=0$ appears in the study of electromagnetic perturbations of
Schwarzschild black holes, i.e., if one considers
the Einstein--Maxwell system and linearizes around
the Reissner--Nordstr\"om solution with zero charge.
We do not comment on this further but simply refer the reader to the literature,
see \cite{Chandra} and references therein.
As a consequence, the study of the Regge--Wheeler equation can provide valuable
information on the stability of Schwarzschild black holes under various types of
perturbations and it is truly remarkable that such a unified approach is
available.

\subsection{Decay estimates for the Regge--Wheeler equation}

The most salient feature of the Regge--Wheeler potential is
that it decays exponentially as $x \to -\infty$ which corresponds to
approaching the black hole, whereas for $x \to \infty$, it falls off as
$x^{-2}$. Strictly speaking, this is only true for $\ell>0$. The case $\ell=0$ 
is exceptional and we consider it separately in the companion paper
\cite{donninger} where we obtain the sharp $t^{-3}$ decay as predicted by Price's Law.
Consequently, in this paper, we focus on $\ell>0$ which, unless otherwise stated, will be assumed throughout.
For $\ell>0$ the potential has inverse square decay and
it is well-known~\cite{deift} that this fall--off behavior
is in some sense critical for the scattering theory.
In order to explain this we define the Schr\" odinger operator
$\mc{H}_{\ell,\sigma}$ by
$$ \mc{H}_{\ell,\sigma}f=-f''+V_{\ell,\sigma}f $$
and recall that the Jost solutions $f_\pm(x,\lambda)$ are defined by
$\mc{H}_{\ell,\sigma} f_\pm(\cdot,\lambda)=\lambda^2 f_\pm(\cdot,\lambda)$ and $f_\pm(x,\lambda) \sim e^{\pm i\lambda x}$ as
$x \to \pm \infty$.
The property $V_{\ell,\sigma} \in L^1(\mathbb{R})$
is sufficient to guarantee the existence of these solutions, see \cite{deift},
but the inverse square decay of $V_{\ell,\sigma}$ is critical in the sense
that at this power the Jost solutions typically are no longer
continuous as $\lambda \to 0$.
Nevertheless, following \cite{schlag2}, it is possible to perform a
detailed spectral and scattering analysis of the Schr\"odinger
operator $\mc{H}_{\ell,\sigma}$. 
However, we emphasize that the present work differs considerably from
\cite{schlag2} due to the asymmetric decay properties of
the potential $V_{\ell,\sigma}$.
Of
particular importance is the asymptotic behavior of the 
resolvent $((\lambda+i0)-\mc{H}_{\ell,\sigma})^{-1}$ (and thus, of the
Jost solutions and their Wronskian) as $\lambda \to 0$.
This is
a common feature in dispersive estimates, see~\cite{schlag3}. 
In particular, we are faced with the possibility of a {\em zero
energy resonance}. However, it was already observed earlier \cite{RW}, 
\cite{Wald2},
\cite{Price1}, \cite{Vish}
that
in the physically relevant cases such a zero energy resonance does
not occur (see also Lemma \ref{lem_nonresonant} below).
Our approach is detailed enough to show rigorously, for the first
time, the decay estimates depending on the angular momentum of the
initial data. In his seminal work \cite{Price1}, \cite{Price2}, 
see also \cite{GPP}, \cite{Price-Bur}, 
Price heuristically derived
the decay rate in time at a fixed point in space, and concluded that, depending
on initial conditions, 
the decay rate is either $t^{-2\ell-3}$ or $t^{-2\ell-2}$ where $\ell$
is the angular momentum. 
This result is now referred to as \emph{Price's Law}.
There has been some confusion in the literature concerning the precise prediction of Price's Law. 
This has been clarified in Price and Burko \cite{Price-Bur}.
If the initial data decay sufficiently fast at spatial infinity then the pointwise decay in time is predicted to be $t^{-2\ell-3}$.
In the present paper we give the first proof of an $\ell$--dependent decay rate.
More precisely, we obtain a $t^{-2\ell-2}$ estimate which
is one power off the sharp version of Price's Law. 
However, we emphasize that our method yields estimates in terms of the initial data and not just a pointwise decay law as is common in the physics literature.
To be more precise, we show that
\begin{align}
\label{eq_crux1} 
\|w_\alpha \cos(t\sqrt{\mc{H}_{\ell,\sigma}})f\|_{L^\infty(\mathbb{R})}
& \leq C_{\ell,\alpha} \langle t \rangle^{-\alpha}\left (\left \|\frac{f'}{w_\alpha}\right
\|_{L^1(\mathbb{R})}+ \left \|\frac{f}{w_\alpha}\right
\|_{L^1(\mathbb{R})} \right ) \\
\label{eq_crux2}
\left \|w_\alpha
\frac{\sin(t\sqrt{\mc{H}_{\ell,\sigma}})}{\sqrt{\mc{H}_{\ell,\sigma}}}f \right 
\|_{L^\infty(\mathbb{R})}
& \leq C_{\ell,\alpha} \langle t \rangle^{-\alpha+1}\left \|\frac{f}{w_\alpha}\right
\|_{L^1(\mathbb{R})} 
\end{align}
for all $t \geq 0$ where $w_\alpha(x):=\langle x \rangle^{-\alpha}$ is a
polynomial weight and, as usual, $\langle x \rangle:=(1+|x|^2)^{1/2}$.
Here one has to require $1 \leq \alpha \leq 2\ell+3$ and one needs to exclude
\footnote{In fact, the case $(\ell,\sigma)=(0,0)$ is trivial since for these
parameter values the Regge--Wheeler equation reduces to 
the free wave equation
on the line.}
$(\ell,\sigma)\in\{(0,0), (0,-3), (1,-3)\}$ which are exactly those cases
where zero energy resonances {\em do} occur --- however, they are
physically irrelevant due to a gauge freedom, cf.~\cite{Chandra}. 
Observe that for $\alpha=2\ell+3$ we obtain precisely the aforementioned bound.
It is also obvious from our approach that the decay of initially static perturbations is better by one power of $t$ as is reflected by our cosine estimate.
This is a general effect which is also present in Price's prediction, cf.~\cite{Price-Bur}.

The proof of (\ref{eq_crux1}) and (\ref{eq_crux2}) is based on 
representing the solution as an oscillatory
integral in the energy variable $\lambda$, schematically one may write
$$
\psi(t,x) = \int U(t,\lambda) \mathrm{Im} \left [
G_{\ell,\sigma}(x,x',\lambda) \right ] f(x')\,
dx'd\lambda
$$
where  $U(t,\lambda)$ is a combination of $\cos (t\lambda) $ and
$\sin (t\lambda)$ terms and $G_{\ell,\sigma}(x,x',\lambda)$  
is the kernel (Green's function) of the
resolvent of the  operator $\mc{H}_{\ell,\sigma}$.  
$G_{\ell,\sigma}(x,x',\lambda)$ is constructed in terms of the
Jost solutions and we obtain these functions in various domains of
the $(x,\lambda) $ plane by perturbative arguments: for $|x \lambda|$ small
we perturb in $\lambda$ around $\lambda=0$, whereas for $|x \lambda|$ large we
perturb off of Hankel functions.
This is done in such a way that there remains a small window where the two
different
perturbative solutions can be glued together.
One of the main technical difficulties of the proof lies with the fact that we
need good estimates for arbitrary derivatives of the perturbative solutions.
This is necessary in order to control the oscillatory integrals.
The most important contributions come from $\lambda \sim 0$ and we therefore
need to derive the
exact asymptotics of the Green's function and its derivatives in the limit
$\lambda \to 0$. 
For instance, we prove that
$$ \mathrm{Im}\left [G_{\ell,\sigma}(0,0,\lambda) \right
]=\lambda P_\ell(\lambda^2)+O(\lambda^{2\ell+1}) $$
as $\lambda \to 0+$ where $P_\ell$ is a polynomial of degree $\ell-1$ (we set
$P_0 \equiv 0$) and the
$O$--term satisfies $O^{(k)}(\lambda^{2\ell+1})=O(\lambda^{2\ell+1-k})$ for all
$k \in \mathbb{N}_0$. 
Our approach therefore yields further information on the Green's
function and the fundamental solution of the wave equation on
the Schwarzschild manifold.

\subsection{Interpretation of the result and further comments}

For the relevant parameter values, i.e., $(\ell,\sigma)\not\in\{(0,0),(0,-3),(1,-3)\}$, the Regge-Wheeler potential
$V_{\ell,\sigma}$ is positive, decays as described above, and has a unique 
nondegenerate maximum at $r=r_0$ which
is known as the \emph{photon sphere}. As a helpful analogy, consider a Newtonian particle 
placed at $r=r_0$ with vanishing kinetic energy
but potential energy $V_{\ell,\sigma}|_{r=r_0}$. It will remain at rest, 
but every slight perturbation will make it lose
potential energy and gain kinetic energy; the larger $\ell$ is, 
the faster this will occur which reflects itself in the $\ell$--dependent
decay rates. In the context of the Schwarzschild geometry, 
the flow of null geodesics near
$r=r_0$ is unstable and
the dispersion provides a 
mechanism that spreads out the wave away from the photon sphere.
Moreover, the higher the angular momentum $\ell$, the faster the 
geodesics will pull away leading to the accelerated decay
provided by Price's Law. 
Note carefully, however, that this is 
counteracted by what can be viewed as a stabilizing effect of large~$\ell$.
Technically speaking, this reflects itself in the constant $C_{\ell,\alpha}$: 
the larger this constant is, the longer one has to wait
before the decay estimates become effective. 
It is important to note that our approach is essentially blind to the local
geometry, that is, the fine structure of the potential is irrelevant --- only
positivity, the decay properties and the nonexistence of a zero energy
resonance are used.
This is in contrast to the methods based on Morawetz type estimates.
In particular, the phenomenon of trapping does not play a role at this level ---
it simply enlarges the constants.
However, eventually one is interested in the overall decay which is obtained by
summing the individual contributions over all $\ell$ and
at this stage, of course, trapping becomes relevant since the 
$\ell$--dependence 
of the constants is crucial for the summation.
As a matter of fact, our proof produces a constant which grows super--exponentially 
in~$\ell$.
Consequently, in order to be able to sum the estimates, a different approach is necessary for large $\ell$.
This issue is addressed in our paper \cite{DSS2} where a detailed semiclassical asymptotic analysis is performed. 
The role of the semiclassical parameter $\hbar$ is played by $\ell^{-1}$ 
(simply divide ${\mathcal H}_{\ell,\sigma}$ by $\ell^2$). 
In particular, such an analysis requires a careful study of the
spectral measure near the maximum of $V_{\ell,\sigma}$ and 
it is exactly at this point where the instability of null geodesics at 
the photon sphere becomes
crucial.
As a consequence, in \cite{DSS2}, we show that the estimates for individual $\ell$'s can indeed be summed and thereby, we obtain the sharp $t^{-3}$ decay bound for
general data
with a loss of a finite number of angular derivatives.
We also remark that there are various formal approaches in the physics 
literature to
find the $\ell$--dependence of the constants, see in particular \cite{Bar-Ori}. 

Let us finally mention that decay estimates like (\ref{eq_crux1}) and (\ref{eq_crux2}) 
play an important role in the current
theoretical and numerical analysis of black holes. For instance, 
they serve as a
way to verify various numerical schemes for solving the Einstein
equations in the presence of black holes, see e.g.
\cite{aich1}, \cite{aich2}, \cite{aich3}, 
\cite{GPP}, \cite{Bizon}, \cite{Bizon2}, \cite{Bla-Burko}, \cite{Bur-Khanna}, \cite{Bar-Ori}, 
\cite{Dam-N1}, \cite{Dam-N2}, \cite{Dam-N3},
\cite{Dam-N4} and cited
ref. 
For other recent theoretical implications of the angular
behavior see for example \cite{Bar-Ori}, \cite{Bizon} and \cite{Matsas}. 

\subsection{Notations and conventions}
For a given smooth function $f$ we denote by $O(f(x))$ a generic 
\emph{real--valued} function that satisfies
$|O(f(x))|\leq |f(x)|$ in a specified range of $x$ which follows from
the context. We write $O_\mathbb{C}(f(x))$ if the function attains complex 
values.
The symbol $f(x) \sim g(x)$ for $x \to a$, where $g$ is smooth, 
means $\lim_{x \to a}\frac{f(x)}{g(x)}=1$.
Furthermore, the letter $C$ (possibly with indices) denotes a generic
positive constant. 
We say that $O(x^\gamma)$, $\gamma \in \mathbb{R}$, 
\emph{behaves like a symbol}, \emph{is of symbol
type}, or \emph{has symbol character}, if the $k$--th derivative satisfies $O^{(k)}(x^\gamma)=O(x^{\gamma-k})$.
As usual, we use the abbreviation $\langle x \rangle:=\sqrt{1+|x|^2}$ and
the symbol $A \lesssim B$ means that there exists a $C>0$ such that $A
\leq CB$.
We also note that all of the functions we are going to construct will depend on
the parameters $\ell$ and $\sigma$, 
however, in order to increase readability of the equations, we will omit this
dependence in the notation most of the time. 
The same comment applies to all implicit and explicit constants.
Finally, as already mentioned, we assume $\ell>0$ unless otherwise stated.

\section{Solutions of the Regge--Wheeler equation}

\subsection{Asymptotics of the potential}

As explained in the introduction, linear perturbations of the 
Schwarzschild spacetime are described by the Regge--Wheeler equation
$$ \psi_{tt}-\psi_{xx}+V_{\ell,\sigma}(x)\psi=0 $$
where $V_{\ell,\sigma}$ is the Regge--Wheeler potential and 
$x$ is the tortoise coordinate which is related to the standard 
$r$--coordinate by 
\begin{equation}
\label{eq_tortoise}
x=r+2M \log \left (\frac{r}{2M}-1 \right ).
\end{equation}
$V_{\ell,\sigma}$ is given by 
$$ V_{\ell,\sigma}(x)=\left (1-\frac{2M}{r(x)} \right )\left (\frac{\ell(\ell+1)}{r^2(x)}+\frac{2M\sigma}{r^3(x)} \right ) $$
where $r(x)$ is implicitly defined by Eq.~(\ref{eq_tortoise}).
The valid range of the parameters is $\ell \in \mathbb{N}_0$ and
$\sigma=-3,0,1$. 
We start by obtaining the asymptotics of the potential $V_{\ell,\sigma}$.

\begin{lemma}
\label{lem_r}
The function $x \mapsto r(x)$ has the asymptotic behavior 
$r(x)=x(1+O(x^{-1+\varepsilon}))$ for $x \to
\infty$ and $r(x)=2M+O(e^{x/(2M)})$ for $x \to -\infty$ where $\varepsilon \in (0,1)$ is arbitrary
and the $O$--term in the
expression for $x \to \infty$ behaves
like a symbol.
\end{lemma}

\begin{proof}
The function $r(x)$ is implicitly defined by $x=r(x)+2M \log \left (\frac{r(x)}{2M}-1 \right )$ and 
thus, we have $r(x) \to \infty$ as $x \to \infty$.
This implies $x \sim r(x)$ and hence, $r(x) \sim x$ as $x \to \infty$.
We infer that $x-r(x) = 2M \log \left( \frac{r(x)}{2M}-1 \right ) \sim 
2M \log \frac{r(x)}{2M}  \sim 2M
\log \frac{x}{2M}$ and this shows $x-r(x)=O(x^\varepsilon)$.
For the symbol behavior note that $\frac{dx}{dr}(r)=\left (1-\frac{2M}{r} \right)^{-1}$ which implies that
$r'(x)=1-\frac{2M}{r(x)}$. The claim now follows by induction.

For the case $x \to -\infty$ we have $e^{x/(2M)}=e^{r(x)/(2M)}(\frac{r(x)}{2M}-1) \sim e
(\frac{r(x)}{2M}-1)$ which shows $r(x)-2M \sim 2Me^{x/(2M)-1}$ and this implies the claim.
\end{proof}

\begin{corollary}
\label{cor_RWasymptotics}
The Regge--Wheeler potential $V_{\ell,\sigma}$ has the asymptotic behavior
$$V_{\ell,\sigma}(x)=\frac{\ell(\ell+1)}{x^2}(1+O(x^{-1+\varepsilon}))$$
for $x \to \infty$ and
$V_{\ell,\sigma}(x)=O(e^{x/(2M)})$ for $x \to -\infty$ where $\varepsilon \in (0,1)$ is arbitrary and
the $O$--term in the expression for $x \to \infty$ behaves like a symbol.
\end{corollary}

\begin{proof}
Just insert the asymptotic expansions from Lemma \ref{lem_r} 
in the expression for $V_{\ell,\sigma}$.
For the symbol behavior apply the Leibniz rule and Lemma \ref{lem_symbolinverse}.
\end{proof} 

\subsection{Hilbert space formulation}
We define the Schr\"odinger operator $\mc{H}_{\ell,\sigma}$ on $L^2(\mathbb{R})$
with domain $\mc{D}(\mc{H}_{\ell,\sigma}):=H^2(\mathbb{R})$ by
$$\mc{H}_{\ell,\sigma}f:=-f''+V_{\ell,\sigma}f. $$
From the decay properties of $V_{\ell,\sigma}$ it follows that
$\mc{H}_{\ell,\sigma}$ is self--adjoint (see e.g. \cite{teschl}).
Furthermore, integration by parts shows
$$ (\mc{H}_{\ell,\sigma} f|f)_{L^2(\mathbb{R})} \geq (V_{\ell,\sigma}
f|f)_{L^2(\mathbb{R})} \geq 0 $$
since $V_{\ell,\sigma} \geq 0$ for all $\ell \in \mathbb{N}_0$ if $\sigma=0,1$.
For gravitational perturbations ($\sigma=-3$) we have to assume $\ell \geq 2$ to
obtain $V_{\ell,\sigma}\geq 0$ which we shall do from now on.
We conclude that the spectrum of $\mc{H}_{\ell,\sigma}$ 
is purely absolutely continuous and we have
$\sigma(\mc{H}_{\ell,\sigma})=\sigma_{ac}(\mc{H}_{\ell,\sigma})=[0,\infty)$ (see
\cite{teschl}) provided that $(\ell,\sigma) \notin \{(0,-3),(1,-3)\}$. 
An operator formulation
of the Regge--Wheeler equation is given by
$$ \frac{d^2}{dt^2}\Psi(t)+\mc{H}_{\ell,\sigma} \Psi(t)=0 $$
where $\Psi: \mathbb{R} \to L^2(\mathbb{R})$.
Applying the functional calculus for self--adjoint operators, the solution
$\Psi$ with initial data $\Psi(0)=f$ and $\frac{d\Psi}{dt}(0)=g$ is given by 
$$
\Psi(t)=\cos(t\sqrt{\mc{H}_{\ell,\sigma}})f+\frac{\sin(t\sqrt{\mc{H}_{\ell,\sigma}})}
{\sqrt{\mc{H}_{\ell,\sigma}}}g. $$
Thus, in order to obtain decay estimates for the solution, we have to understand the
operators $\cos(t\sqrt{\mc{H}_{\ell,\sigma}})$ and $\frac{\sin(t\sqrt{\mc{H}_{\ell,\sigma}})}
{\sqrt{\mc{H}_{\ell,\sigma}}}$.

\subsection{The spectral measure}
\label{sec_specmeas}
Recall that the spectral theorem for self--adjoint operators asserts the
existence of finite complex--valued Borel measures 
$\mu_{u,v}$ such that, for $u,v \in \mc{D}(\mc{H}_{\ell,\sigma})$, we have
$$ (\mc{H}_{\ell,\sigma} u | v)_{L^2(\mathbb{R})}=\int_0^\infty \lambda
d\mu_{u,v}(\lambda). $$
The solution operator $\cos(t \sqrt{\mc{H}_{\ell,\sigma}})$ is then given by
$$ \left ( \left . \cos(t\sqrt{\mc{H}_{\ell,\sigma}}  ) u \right | v \right )_{L^2(\mathbb{R})}=\int_0^\infty
\cos(t\sqrt{\lambda})d\mu_{u,v}(\lambda) $$
for $u,v \in L^2(\mathbb{R})$ and analogous for the sine evolution.
The point is that the spectral measure can be calculated in terms of the
resolvent $R_{\mc{H}_{\ell,\sigma}}(z)=(z-\mc{H}_{\ell,\sigma})^{-1}$ 
of $\mc{H}_{\ell,\sigma}$.
Indeed, for $u \in L^2(\mathbb{R})$ set
$$ F_u(z):=-(R_{\mc{H}_{\ell,\sigma}}(z)u|u)_{L^2(\mathbb{R})}=\int_0^\infty
\frac{1}{\lambda-z}d\mu_u(\lambda) $$
where $\mu_u:=\mu_{u,u}$ and $\mathrm{Im}z>0$.
$F_u$ is the Borel transform of the measure $\mu_u$ and,
since the measure $\mu_u$ is purely absolutely continuous, we
have $$ d\mu_u(\lambda)=\frac{1}{\pi}\lim_{\varepsilon \to
0+}\mathrm{Im}(F_u(\lambda+i\varepsilon))d\lambda, $$
see \cite{teschl} for the underlying theory of this.
The measure $\mu_{u,v}$ can be reconstructed from $\mu_u$ by the polarization
identity, i.e.,  $\mu_{u,v}=\frac{1}{4}(\mu_{u+v}-\mu_{u-v}+i\mu_{u-iv}-i\mu_{u+iv})$.
Furthermore, the resolvent is given by
$$ R_{\mc{H}_{\ell,\sigma}}(z)u(x)=\int_\mathbb{R}G_{\ell,\sigma}(x,x',\sqrt{z})u(x')dx' $$ where
$G_{\ell,\sigma}$ is the Green's function (we always choose the branch of the
square root with $\mathrm{Im}\sqrt{z}>0$ if $\mathrm{Im}z>0$) and thus, we have
$$ d\mu_u(\lambda)=-\frac{1}{\pi}\lim_{\varepsilon \to 0+}\int_\mathbb{R} \int_\mathbb{R}
\mathrm{Im} \left [G_{\ell,\sigma}(x,x',\sqrt{\lambda+i\varepsilon})u(x')\overline{u(x)}\right ]
dx'dxd\lambda . $$
It is known (and, for the convenience of the reader, will be shown below) that 
the limit
$$G_{\ell,\sigma}(x,x',\sqrt{\lambda}):=
\lim_{\varepsilon \to 0+}G_{\ell,\sigma}(x,x',\sqrt{
\lambda+i\varepsilon})$$ exists and satisfies $\sup_{x,x' \in
\mathbb{R}}|G_{\ell,\sigma}(x,x',\lambda)|\lesssim 1$ for all 
$\lambda \geq \lambda_0$ where
$\lambda_0>0$ is arbitrary.
Thus, if $u \in L^1(\mathbb{R}) \cap L^2(\mathbb{R})$, we have
$$ d\mu_u(\lambda)=-\frac{1}{\pi}\int_\mathbb{R} \int_\mathbb{R}
\mathrm{Im} \left [G_{\ell,\sigma}(x,x',\sqrt{\lambda})u(x')\overline{u(x)}\right ]
dx'dxd\lambda $$
by Lebesgue's theorem on dominated convergence and polarization
yields
$$ d\mu_{u,v}(\lambda)=-\frac{1}{\pi}\int_\mathbb{R} \int_\mathbb{R}
\mathrm{Im} \left [ G_{\ell,\sigma}(x,x',\sqrt{\lambda}) \right ] u(x')\overline{v(x)}
dx'dxd\lambda $$
for all $u,v \in L^1(\mathbb{R}) \cap L^2(\mathbb{R})$
since 
$G_{\ell,\sigma}(x,x',\sqrt{\lambda})$ is symmetric in $x$ and $x'$ for $\lambda >0$ as follows
from the explicit form (see below). 

\subsection{Pointwise decay estimates}   
As follows from the discussion above, the functional calculus for 
self--adjoint operators yields
the representation
$$ \left ( \left. \cos(t\sqrt{\mc{H}_{\ell,\sigma}})f \right | v \right
)_{L^2(\mathbb{R})}=-\frac{2}{\pi}\int_0^\infty  \int_\mathbb{R} \int_\mathbb{R} 
\lambda \cos(t\lambda)\mathrm{Im}\left [G_{\ell,\sigma}
(x,x',\lambda) \right ]f(x')dx'
\overline{v(x)}dx d\lambda $$
for $f,v \in \mc{S}(\mathbb{R})$ (the Schwartz space), where we have changed
variables in the integration with respect to $\lambda$.
Our intention is to obtain an expression for 
$\left [ \cos(t\sqrt{\mc{H}_{\ell,\sigma}})f \right ](x)$ and thus, we have to change
the order of integration. 
However, note carefully that a simple argument based on Fubini's theorem does
not apply here since the integrals cannot be expected
to converge absolutely.
In order to circumvent this difficulty, first observe that, for any $N \in
\mathbb{N}$, we have
$$ \int_{1/N}^N \int_\mathbb{R} \int_\mathbb{R} \left |
\lambda \cos(t\lambda)\mathrm{Im}\left [G_{\ell,\sigma}
(x,x',\lambda) \right ]f(x')\overline{v(x)} \right |dx' dx d\lambda \leq C_N $$
which follows immediately from 
$\sup_{x,x' \in
\mathbb{R}}|G_{\ell,\sigma}(x,x',\lambda)|\leq C_N$ for all 
$\lambda \geq \frac{1}{N}$, see Corollary \ref{cor_Green} below.
Thus, Fubini's theorem yields at least
$$ \left ( \left. \cos(t\sqrt{\mc{H}_{\ell,\sigma}})f \right | v \right
)_{L^2(\mathbb{R})}=-\frac{2}{\pi}\lim_{N \to \infty}
\int_\mathbb{R} \int_\mathbb{R} \int_{1/N}^N
\lambda \cos(t\lambda)\mathrm{Im}\left [G_{\ell,\sigma}
(x,x',\lambda) \right ]f(x')d\lambda dx'
\overline{v(x)}dx. $$
Next, we distinguish between high and low
energies by introducing a smooth cut--off $\chi_\delta$ satisfying
$\chi_\delta(\lambda)=1$ for $\lambda \in [0, \frac{\delta}{2}]$ and
$\chi(\lambda)=0$ for $\lambda \geq \delta$ where $\delta>0$ is sufficiently
small.
In Sec. \ref{sec_osc} below we prove the estimate
\begin{equation}
\label{eq_estsmall}
\sup_{x,x' \in \mathbb{R}}\left |\int_0^\infty 
\lambda \cos(t\lambda)\mathrm{Im}\left [G_{\ell,\sigma}(x,x',\lambda) \right ] \langle x
\rangle^{-\alpha}\langle x' \rangle^{-\alpha}\chi_\delta(\lambda)
d\lambda \right |\lesssim \langle t
\rangle^{-\alpha} 
\end{equation}
where $\alpha \in \mathbb{N}$ and $\alpha \leq 2\ell+3$.
This bound is sufficient to conclude
\begin{align*} \lim_{N \to \infty}
\int_\mathbb{R} \int_\mathbb{R} \int_{1/N}^N
\lambda \cos(t\lambda)\mathrm{Im}\left [G_{\ell,\sigma}
(x,x',\lambda) \right ]\chi_\delta(\lambda)f(x')d\lambda dx'
\overline{v(x)}dx \\
=\int_\mathbb{R} \lim_{N \to \infty} \int_\mathbb{R} \int_{1/N}^N
\lambda \cos(t\lambda)\mathrm{Im}\left [G_{\ell,\sigma}
(x,x',\lambda) \right ]\chi_\delta(\lambda)f(x')d\lambda dx'
\overline{v(x)}dx 
\end{align*}
by dominated convergence since $f,v \in \mc{S}(\mathbb{R})$.
For the large energy part we show in Sec. \ref{sec_osclarge} that, for any
$\alpha \in \mathbb{N}_0$,
\begin{align}
\label{eq_estlarge}
\sup_{x \in \mathbb{R}}\left |\lim_{N \to \infty} \int_\mathbb{R} 
\int_{1/N}^N
\lambda \cos(t\lambda)\mathrm{Im}\left [G_{\ell,\sigma}(x,x',\lambda) \right ] 
\langle x \rangle^{-\alpha} \langle x' \rangle ^{-\alpha} \phi(x')
[1-\chi_\delta(\lambda)]
d\lambda dx' \right | \\
\lesssim \langle t
\rangle^{-\alpha} \int_\mathbb{R}\left (|\phi'(x')|+|\phi(x')|\right )dx' 
\nonumber
\end{align}
which, by dominated convergence, implies
\begin{align*} \lim_{N \to \infty}
\int_\mathbb{R} \int_\mathbb{R} \int_{1/N}^N
\lambda \cos(t\lambda)\mathrm{Im}\left [G_{\ell,\sigma}
(x,x',\lambda) \right ][1-\chi_\delta(\lambda)]f(x')d\lambda dx'
\overline{v(x)}dx \\
=\int_\mathbb{R} \lim_{N \to \infty}\int_\mathbb{R} \int_{1/N}^N
\lambda \cos(t\lambda)\mathrm{Im}\left [G_{\ell,\sigma}
(x,x',\lambda) \right ][1-\chi_\delta(\lambda)]f(x')d\lambda dx'
\overline{v(x)}dx.
\end{align*}
By adding up the two contributions and using the density of
$\mc{S}(\mathbb{R})$ in $L^2(\mathbb{R})$, we arrive at the representation
$$ \left [\cos(t\sqrt{\mc{H}_{\ell,\sigma}})f \right ](x)=
-\frac{2}{\pi}\lim_{N \to \infty}\int_\mathbb{R} \int_{1/N}^N
\lambda \cos(t\lambda)\mathrm{Im}\left [G_{\ell,\sigma}
(x,x',\lambda) \right ]d\lambda f(x')dx' $$
for $f \in \mc{S}(\mathbb{R})$
and the estimates (\ref{eq_estsmall}), (\ref{eq_estlarge}) imply the bound
$$ \|w_\alpha \cos(t\sqrt{\mc{H}_{\ell,\sigma}})f\|_{L^\infty(\mathbb{R})}
\lesssim \langle t \rangle^{-\alpha}\left (\left \|\frac{f'}{w_\alpha}\right
\|_{L^1(\mathbb{R})}+
\left \|\frac{f}{w_\alpha}\right
\|_{L^1(\mathbb{R})} \right ) $$
for $1 \leq \alpha \leq 2\ell+3$ where
$w_\alpha(x):=\langle x \rangle^{-\alpha}$.
An analogous derivation applies to the sine evolution and
therefore, the proof of our result reduces to oscillatory estimates of the type
(\ref{eq_estsmall}) and (\ref{eq_estlarge}).

\subsection{The main theorem}
The main result proved in this work is the following.

\begin{theorem}
\label{thm_main}
Let $(\ell,\sigma) \notin \{(0,0), (0,-3), (1,-3)\}$, $\alpha \in \mathbb{N}$,
$1 \leq \alpha \leq 2\ell+3$ and set $w_\alpha(x):=\langle x \rangle^{-\alpha}$.
Then the solution operators for the Regge--Wheeler equation satisfy the 
estimates
$$ \|w_\alpha \cos(t\sqrt{\mc{H}_{\ell,\sigma}})f\|_{L^\infty(\mathbb{R})}
\leq C_{\ell,\alpha} \langle t \rangle^{-\alpha}\left (\left \|\frac{f'}{w_\alpha}\right
\|_{L^1(\mathbb{R})}+ \left \|\frac{f}{w_\alpha}\right
\|_{L^1(\mathbb{R})} \right ) $$
and
$$ \left \|w_\alpha
\frac{\sin(t\sqrt{\mc{H}_{\ell,\sigma}})}{\sqrt{\mc{H}_{\ell,\sigma}}}g \right 
\|_{L^\infty(\mathbb{R})}
\leq C_{\ell,\alpha} \langle t \rangle^{-\alpha+1}\left \|\frac{g}{w_\alpha}\right
\|_{L^1(\mathbb{R})} $$
for all $t \geq 0$ and initial data $f,g$ such that the right--hand sides are finite.
\end{theorem}

\begin{remark}
 As usual, we prove Theorem \ref{thm_main} for Schwartz functions $f, g \in \mc{S}(\mathbb{R})$. The general case is then obtained by a standard approximation argument.
\end{remark}

\begin{remark}
 For the convenience of the reader we make the behavior of the initial data near the horizon more explicit by transforming back to the Schwarzschild $r$--coordinate.
 Recall that $x=r+2M \log(\frac{r}{2M}-1)$ which implies 
 $\langle x \rangle^{\alpha} \simeq \left |\log(\frac{r}{2M}-1) \right |^\alpha$ as $r \to 2M+$.
 Thus, the polynomial weights in $x$ translate into logarithmic weights in $r$.
 Moreover, we have $dx=(1-\frac{2M}{r})^{-1}dr$ and this shows that the integrability condition near the horizon for initial data $g(x)=\tilde{g}(r(x))$ transforms as
 $$ \int_{-\infty}^0 |g(x)|\langle x \rangle^\alpha dx \simeq \int_{2M}^{r_0} |\tilde{g}(r)|\left | \log \left (\tfrac{r}{2M}-1 \right )\right |^\alpha \frac{dr}{1-\tfrac{2M}{r}} $$
 where $r=r_0$ corresponds to $x=0$.
\end{remark}

We remark that our proof actually applies to more general situations like
the analogous problem in Ho\v{r}ava--Lifshitz gravity, cf. \cite{Hor}.
The only requirements on the potential are the asymptotics of
Corollary \ref{cor_RWasymptotics}, the nonexistence of 
bound states and the nonexistence of a zero energy 
resonance (see Definition \ref{def_resonance} below).

\section{Basic properties of the Green's function}
For the convenience of the reader 
we discuss some well--known properties of the Green's function (cf. \cite{deift}, 
\cite{teschl}).

\subsection{The Jost solutions}
Recall that the Green's function is constructed with the help of the Jost solutions $f_\pm(\cdot,z)$ 
which are defined by $\mc{H}_{\ell,\sigma}f_\pm(\cdot,z)=z^2 f_\pm(\cdot,z)$
and the asymptotic behavior $f_\pm(x,z) \sim e^{\pm izx}$ as $x \to \pm \infty$.
First we prove that the Jost solutions exist and that they 
are continuous with respect to $z$ in 
$\overline{\mathbb{C}_+} \backslash \{0\}$
where $\mathbb{C}_+:=\{z \in \mathbb{C}: \mathrm{Im}z >0\}$.

\begin{lemma}
\label{lem_Jost}
For every $z \in \overline{\mathbb{C}_+} \backslash \{0\}$ there exist smooth functions 
$f_\pm(\cdot,z)$ satisfying 
$$ \mc{H}_{\ell,\sigma}f_\pm(\cdot,z)=z^2 f_\pm(\cdot,z) $$
and $f_\pm(x,z) \sim e^{\pm izx}$ for $x \to \pm \infty$.
Furthermore, for every $x \in \mathbb{R}$, the functions $f_\pm(x, \cdot)$ 
and $f_\pm'(x,\cdot)$ are continuous
in $\overline{\mathbb{C}_+} \backslash \{0\}$.
\end{lemma}

\begin{proof}
We only prove the assertion for $f_+$ since the proof for $f_-$ is completely analogous.
The variation of constants formula shows that $m_+(x,z):=e^{-izx}f_+(x,z)$, 
if it exists, satisfies the integral equation
\begin{equation}
\label{eq_proofJost}
m_+(x,z)=1+\int_x^\infty K(x,y,z)m_+(y,z)dy
\end{equation}
where $K(x,y,z)=\frac{1}{2iz}\left (e^{2iz(y-x)}-1 \right ) V_{\ell,\sigma}(y)$.
Conversely, if we can show that Eq.~(\ref{eq_proofJost}) has a smooth solution, 
we obtain existence
of the Jost solution.
However, Eq.~(\ref{eq_proofJost}) is a Volterra integral equation with a kernel
satisfying 
$$ \int_a^\infty \sup_{x \in (a,y)}|K(x,y,z)|dy \leq \frac{C}{|z|} \:,\:\:
\int_a^\infty \sup_{x \in (a,y)}|\partial_x K(x,y,z)|dy \lesssim 1 $$
for all $z \in \overline{\mathbb{C}_+} \backslash \{0\}$ and any fixed 
$a \in \mathbb{R}$
(see Corollary \ref{cor_RWasymptotics})
and thus, Lemma \ref{lem_volterra} implies the existence of a unique solution 
$m_+(\cdot,z)$ satisfying $\|m_+(\cdot,z)\|_{L^\infty(a,\infty)} \leq e^{C/|z|}$.
Furthermore, for fixed $z \in \overline{\mathbb{C}_+} \backslash \{0\}$, we have
$$ \int_a^\infty \sup_{x \in (a,y)}|\partial_x^k K(x,y,z)|dy \leq C_k $$ for all $k
\in \mathbb{N}_0$ and thus, Lemma \ref{lem_volterradiffx} shows that
$m_+(\cdot,z)$ is smooth.
For the continuity of $m_+(x,\cdot)$ fix $x \in (a,\infty)$, 
$z \in \overline{\mathbb{C}_+}\backslash \{0\}$ and note that
$$
m_+(x,z+h)-m_+(x,z)=g_h(x,z)+\int_x^\infty K(x,y,z+h)[m_+(y,z+h)-m_+(y,z)]dy 
$$
where
$$ g_h(x,z):=\int_x^\infty [K(x,y,z+h)-K(x,y,z)]m_+(y,z)dy. $$
Now observe that $\|g_h(\cdot,z)\|_{L^\infty(a,\infty)} \to 0$ as $h \to 0$
since $\|m_+(\cdot,z)\|_{L^\infty(a,\infty)} \leq e^{C/|z|}$ and
hence, Lemma \ref{lem_volterra} implies
$$ |m_+(x,z+h)-m_+(x,z)|\leq \|g_h(\cdot,z)\|_{L^\infty(a,\infty)}e^{C/|z|} \to
0 \mbox{ for } h \to 0 $$ 
which shows continuity of $m_+(x,\cdot)$ in 
$\overline{\mathbb{C}_+}\backslash\{0\}$ as claimed.
For the continuity of $m_+'(x,\cdot)$ simply observe that
$$ m_+'(x,z)=\int_x^\infty \partial_x K(x,y,z)m_+(y,z)dy $$
and the right--hand side of this equation is obviously continuous in $z$. 
\end{proof} 

\subsection{The Wronskian $W(f_-(\cdot,\sqrt{z}),f_+(\cdot,\sqrt{z}))$}

Having established existence of the Jost solutions we can now construct the 
Green's function and the standard procedure yields
$$ G_{\ell,\sigma}(x,x',\sqrt{z})=\frac{f_-(x',\sqrt{z})f_+(x,\sqrt{z})\Theta(x-x')
+f_-(x,\sqrt{z})f_+(x',\sqrt{z})\Theta(x'-x)}{W(f_-(\cdot,\sqrt{z}),f_+(\cdot,\sqrt{z}))} 
$$
for $\mathrm{Im}z>0$ where $\Theta$ denotes the Heaviside function.
Clearly, $W(f_-(\cdot,\sqrt{z}),f_+(\cdot,\sqrt{z})) \not= 0$ if $\mathrm{Im}z>0$ since otherwise
$f_-(\cdot,\sqrt{z})$ would be an eigenfunction with eigenvalue $z$ contradicting
the self--adjointness of $\mc{H}_{\ell,\sigma}$.
However, it is not a priori clear whether the limit 
$G_{\ell,\sigma}(x,x',\sqrt{z})$ for $\mathrm{Im}\sqrt{z} \to 0+$ exists.
The following observation shows that problems can only occur at $z=0$. 

\begin{lemma}
\label{lem_Wronskian}
Let $\lambda>0$. Then the limit 
$$ W(\sqrt{\lambda}):=\lim_{\varepsilon \to 0+}W(f_-(\cdot,\sqrt{\lambda+i\varepsilon}),
f_+(\cdot,\sqrt{\lambda+i\varepsilon})) $$ exists and is nonzero.
\end{lemma}

\begin{proof}
For brevity we write $W(f_-,f_+)(z)$ instead of $W(f_-(\cdot,\sqrt{z}), f_+(\cdot,\sqrt{z}))$ and 
likewise for other Wronskians.
By Lemma \ref{lem_Jost} we know that 
$W(f_-, f_+)$ is continuous in 
$\overline{\mathbb{C}_+} \backslash \{0\}$ and hence, $W(\sqrt{\lambda})$ exists
for any $\lambda>0$.
Observe that $f_\pm'(x,\sqrt{\lambda}) \sim \pm i\sqrt{\lambda} 
e^{\pm i\sqrt{\lambda}x}$ 
for $x \to \pm \infty$ which follows 
immediately from
the integral representation in the proof of Lemma \ref{lem_Jost}.
Thus, $W(\overline{f_+}, f_+)(\lambda)=
2i\sqrt{\lambda}$ which shows that 
$f_+(\cdot,\sqrt{\lambda})$ and $\overline{f_+(\cdot,\sqrt{\lambda})}$ 
are linearly independent for $\lambda > 0$.
Hence, there exist $A(\lambda)$ and $B(\lambda)$ such that
$ f_-(x,\sqrt{\lambda})=A(\lambda)f_+(x,\sqrt{\lambda})+
B(\lambda)\overline{f_+(x,\sqrt{\lambda})}$.
We conclude
\begin{equation}
 \label{eq:AB}
2i\sqrt{\lambda}=W(f_-, \overline{f_-})(\lambda)=
W(Af_++B\overline{f_+}, \overline{Af_+}+\overline{B}f_+)(\lambda)
= -2i\sqrt{\lambda}|A(\lambda)|^2+2i \sqrt{\lambda} |B(\lambda)|^2 
\end{equation}
which implies $|B(\lambda)|^2 \geq 1$.
However, we have 
\begin{equation}
 \label{eq:B}
W(\sqrt{\lambda})=W(f_-,f_+)(\lambda)=W(Af_++B\overline{f_+},f_+)(\lambda)=
2i\sqrt{\lambda}B(\lambda) 
\end{equation}
and thus, $|W(\sqrt{\lambda})| \geq 2 \sqrt{\lambda}$ which finishes the proof.
\end{proof}

\begin{corollary}
\label{cor_Green}
The limit 
$$ G_{\ell,\sigma}(x,x',\sqrt{\lambda})=\lim_{\varepsilon \to 0+}
G_{\ell,\sigma}(x,x',\sqrt{
\lambda+i\varepsilon}) $$
exists and satisfies 
$$ \sup_{x,x' \in \mathbb{R}}|G_{\ell,\sigma}(x,x',\sqrt{\lambda})| \leq C $$
for all $\lambda \geq \lambda_0$ where $\lambda_0 > 0$ is arbitrary.
\end{corollary}

\begin{proof}
From Lemma \ref{lem_Wronskian} and the asymptotic 
behavior of the
Jost solutions we immediately conclude
$$ \sup_{x' < 0, x > 0}\left | 
\frac{f_-(x',\sqrt{\lambda})f_+(x,\sqrt{\lambda})}
{W(f_-(\cdot,\sqrt{\lambda}),f_+(\cdot,\sqrt{\lambda}))} \right | \lesssim 1 $$
for all $\lambda \geq \lambda_0>0$.
For the remaining cases use reflection and transmission coefficients
$A(\lambda)$, $B(\lambda)$ (see the proof of Lemma \ref{lem_Wronskian}) to
express $f_\pm$ in terms of $f_\mp$ and $\overline{f_\mp}$.
The asymptotic behavior of $A(\lambda)$ and $B(\lambda)$ for $\lambda \to
\infty$ is given by Eqs.~\eqref{eq:B} and \eqref{eq:AB} and the
claim follows.
\end{proof}

\section{Perturbative solutions for $|x\lambda|$ small}

In this section we obtain approximations to solutions of
$\mc{H}_{\ell,\sigma}f=\lambda^2 f$ for $|x\lambda|$ small.
The solutions are constructed by perturbation in $\lambda$ around $\lambda=0$.
We closely follow \cite{schlag2}. 

\subsection{Zero energy solutions}

We first consider zero energy solutions, i.e.,  solutions of $\mc{H}_{\ell,\sigma} f=0$.
By setting $v(r):=f(x(r))$, the eigenvalue problem $\mc{H}_{\ell,\sigma} f=\lambda^2 f$ is equivalent to
\begin{equation}
\label{eq_RWzero}
-\left ( 1-\frac{2M}{r} \right )v''-\frac{2M}{r^2}v'+\left 
( \frac{\ell(\ell+1)}{r^2}+\frac{2M\sigma}{r^3} \right )v=\lambda^2 
\left(1-\frac{2M}{r}\right )^{-1}v 
\end{equation}
and it turns out that for $\lambda=0$ this equation can be solved by special functions which will be
useful later on.
However, the following result describes a fundamental system for $\mc{H}_{\ell,\sigma} f=0$ without
making use of explicit solutions.

\begin{lemma}
\label{lem_zeroenergy}
There exist smooth functions $u_j$ satisfying $\mc{H}_{\ell,\sigma} u_j=0$
for $j=0,1$ 
with the bounds $u_0(x)=(2\ell+1)^{-1}x^{\ell+1}(1
+O(x^{-1+\varepsilon}))$ and
$u_1(x)=x^{-\ell}(1+O(x^{-1+\varepsilon}))$ for all $x \geq 1$ 
where $\varepsilon \in (0,1)$ 
is arbitrary
and the $O$--terms are of symbol type.
The Wronskian is $W(u_0,u_1)=-1$.
\end{lemma}
 
\begin{proof}
Suppose for the moment that the solution $u_1$ exists and define the function $a$ by
$u_1(x)=x^{-\ell}(1+a(x))$. Then $\mc{H}_{\ell,\sigma}u_1=0$ is equivalent to
\begin{equation}
\label{eq_prooflemma1a}
 a''(x)-\frac{2\ell}{x}a'(x)=\left [ V_{\ell,\sigma}(x)-\frac{\ell(\ell+1)}{x^2} \right ](1+a(x)). 
 \end{equation}
Viewing this equation as an inhomogeneous equation for $a'$, applying the variation of constants
formula and integrating by parts, we obtain the integral equation
\begin{equation}
\label{eq_prooflemma1} a(x)=\frac{1}{2\ell+1}\int_x^\infty (y^{2\ell+1}-x^{2\ell+1})y^{-2\ell}
\left [V_{\ell,\sigma}(y)-\frac{\ell
(\ell+1)}{y^2} \right ](1+a(y))dy. 
\end{equation}
Therefore, if we can show that Eq.~(\ref{eq_prooflemma1}) has a smooth solution, we obtain existence of
$u_1$.
However, Eq.~(\ref{eq_prooflemma1}) is a Volterra integral equation of the form 
$$ a(x)=\int_x^\infty K(x,y)dy + \int_x^\infty K(x,y)a(y)dy $$
with a kernel $K$ satisfying $|\partial_x^k \partial_y^l K(x,y)| \leq C_{k,l}
y ^{-2+\varepsilon-k-l}$ for $1 \leq x \leq y$ and an arbitrary 
$\varepsilon \in (0,1)$ (cf. Corollary \ref{cor_RWasymptotics}).
Therefore, Lemma \ref{lem_volterradiffx} implies the
existence of a unique smooth solution $a \in L^\infty(1,\infty)$ 
and Eq.~(\ref{eq_prooflemma1}) shows that in fact $|a(x)|\lesssim x
^{-1+\varepsilon}$ for $x \geq 1$.
Furthermore, the first derivative of $a$ is given by
$$ a'(x)= \int_x^\infty \partial_x K(x,y)(1+a(y))dy $$
and this implies $|a'(x)|\lesssim x^{-2+\varepsilon}$ for $x \geq 1$.
The estimates for the higher derivatives follow from Eq.~(\ref{eq_prooflemma1a}),
the Leibniz rule and a simple induction.

For the second solution $u_0$ we use the Wronskian condition
$$ 1=W(u_1,u_0)=\left ( \frac{u_0}{u_1} \right )'
u_1^2 $$ which yields
$$ u_0(x)=u_1(x)\int_{x_1}^x u_1^{-2}(y)dy $$
where $x_1>0$ is chosen so large that $u_1(x)>0$ for all $x \geq x_1$.
Note that $u_1^{-2}(x)=x^{2\ell}(1+b(x))$ for a suitable $b$ satisfying 
$|b^{(k)}(x)|\leq C_k \langle x \rangle^{-1+\varepsilon-k}$  
(apply the Leibniz rule and Lemma \ref{lem_symbolinverse}).
Inserting the asymptotic expansion for $u_1$ yields
$u_0(x)=(2\ell+1)^{-1} x^{\ell+1}(1+O(x^{-1+\varepsilon}))$ where the 
$O$--term behaves like a symbol.
\end{proof}

We construct another pair $v_0$, $v_1$ of zero energy solutions 
with specific asymptotic
behavior as $x \to -\infty$. This is considerably easier than the above
construction for the solutions $u_j$ due to the exponential decay of the
Regge--Wheeler potential $V_{\ell,\sigma}(x)$ as $x \to -\infty$.

\begin{lemma}
\label{lem_zeroenergy-}
There exist smooth functions $v_j$ for $j=0,1$ satisfying 
$\mc{H}_{\ell,\sigma}v_j=0$ and 
$v_0(x)=x(1+O(x^{-1}))$ as well as $v_1(x)=1+O(x^{-1})$ for
all $x \leq -1$
where the $O$--terms behave like symbols under differentiation.
\end{lemma}

\begin{proof}
For $x \leq -1$ consider the Volterra equations
$$ \frac{v_0(x)}{x}=1-\int_{-\infty}^x
\left (\frac{y^2}{x}-y \right )V_{\ell,\sigma}(y)\frac{v_0(y)}{y}dy $$
and 
$$ v_1(x)=1-\int_{-\infty}^x (y-x)V_{\ell,\sigma}(y)v_1(y)dy $$
which have smooth solutions according to Lemma \ref{lem_volterradiffx}
since the potential $V_{\ell,\sigma}(y)$ decays exponentially as
$y \to -\infty$ (see Corollary \ref{cor_RWasymptotics}). Obviously, we have
the asymptotic behavior
$v_0(x) \sim x$, $v_1(x) \sim 1$ as $x \to -\infty$,
$v_j$ behave like
symbols under differentiation and satisfy $\mc{H}_{\ell,\sigma}v_j=0$, as
a straightforward calculation shows.
\end{proof}

\begin{corollary}
\label{cor_zeroenergy}
The solutions $u_j$ for $j=0,1$ can be uniquely extended to all of 
$\mathbb{R}$ and we have $u_j(x)=O(x)$ as $x \to -\infty$ 
where the $O$--term is of symbol type.
\end{corollary}

\begin{proof}
Since the potential $V_{\ell,\sigma}$ is smooth on $\mathbb{R}$, the solutions
$u_j$, originally defined on $[1,\infty)$ only, can be uniquely extended to
all of $\mathbb{R}$ by solving appropriate initial value problems.
Since the solution pair $v_0,v_1$ forms a fundamental system for the equation
$\mc{H}_{\ell,\sigma}f=0$, $u_j$ can
be written as a linear combination of $v_0$, $v_1$ on $(-\infty,-1]$
and everything follows
from Lemma \ref{lem_zeroenergy-}. 
\end{proof}

\subsection{Construction of the perturbative solutions}

Next, by perturbing in $\lambda$ around $\lambda=0$, we obtain useful approximations to
solutions of $\mc{H}_{\ell,\sigma} f=\lambda^2 f$ for $|x \lambda|$ small.

\begin{lemma}
\label{lem_pertenergy}
There exist constants $x_0,\lambda_0,\delta>0$ and smooth functions 
$u_j(\cdot,\lambda)$ satisfying
$$\mc{H}_{\ell,\sigma} u_j(\cdot,\lambda)=\lambda^2 u_j(\cdot,\lambda)$$ 
for $j=0,1$ and $W(u_0(\cdot,\lambda),u_1(\cdot,\lambda))=-1$ such that
$u_j(x,\lambda)=u_j(x)(1+a_j(x,\lambda))$ where
$$ |a_j(x,\lambda)|\lesssim
x^2 \lambda^2 $$
for all $\lambda \in (0, \lambda_0)$ and $x \in [x_0,\delta \lambda^{-1}]$
provided that $\ell \geq 1$.
In the case $\ell=0$ we have the weaker bounds
$$ |a_0(x,\lambda)|\lesssim
x^2 \lambda^2 \mbox{ and }
|a_1(x,\lambda)| \lesssim x\lambda $$
in the above ranges of $x$ and $\lambda$.
\end{lemma}

\begin{proof}
Let $x_0>0$ be so large that 
$u_0(x)>0$
for all $x\geq x_0$. 
A straightforward calculation shows that, if the function $h$ solves the
integral equation
\begin{equation}
\label{eq_proofperten}
h(x,\lambda)=1-\lambda^2 \int_{x_0}^x \left [
u_0(y)
u_1(y)-
u_0^2(y)\frac{u_1(x)}{u_0(x)} \right ]h(y,\lambda)dy, 
\end{equation}
then $u_0(x,\lambda):=u_0(x)h(x,\lambda)$ satisfies 
$\mc{H}_{\ell,\sigma}
u_0(\cdot,\lambda)=\lambda^2 u_0(\cdot,\lambda)$.
Eq.~(\ref{eq_proofperten}) is a Volterra integral equation
$$ h(x,\lambda)=1+\int_{x_0}^x K(x,y,\lambda)h(y,\lambda)dy $$
where the kernel is of the form 
$$ K(x,y,\lambda)=\lambda^{2}\left [y(1+O(y^{-1+\varepsilon}))
-x^{-2\ell-1}y^{2\ell+2}(1+O(x^{-1+\varepsilon})+O(y^{-1+\varepsilon}))\right ] $$
for $x,y \geq x_0$ and the $O$--terms are of symbol type (see Lemma
\ref{lem_zeroenergy}). 
This shows $|\partial_x^k \partial_y^l K(x,y,\lambda)|
\leq C_{k,l}y^{1-k-l} \lambda^2$ for $x_0 \leq y \leq x$ and
hence,
$$ \int_{x_0}^{\lambda^{-1}} 
\sup_{\{x: x_0 < y < x\}}|\partial_x^k K(x,y,\lambda)|dy \lesssim 1 $$ 
for $ \lambda \in (0,\lambda_0)$ and all $k \in \mathbb{N}_0$ where
$\lambda_0>0$ is arbitrary. Thus, Lemma 
\ref{lem_volterradiffx} shows that
Eq.~(\ref{eq_proofperten}) has a unique smooth solution 
$h(\cdot,\lambda)$ satisfying 
$\|h(\cdot,\lambda)\|_{L^\infty(x_0,\lambda^{-1})}
\leq C$ for all $\lambda \in (0,\lambda_0)$. 
Therefore, Eq.~(\ref{eq_proofperten})
implies $|h(x,\lambda)-1|\lesssim x^2 \lambda^2$ for 
$\lambda \in (0,\lambda_0)$ and $x_0 \leq
x \leq \lambda^{-1}$.

Now choose $\delta>0$ and $\lambda_0>0$ so small that 
$|h(x,\lambda)-1| \leq \frac{1}{2}$ for 
all $\lambda \in (0,\lambda_0)$ and $x \in [x_0,\delta \lambda^{-1}]$. 
We use the Wronskian condition
$$
-1=W(u_0(\cdot,\lambda),u_1(\cdot,\lambda))
=\left (\frac{u_1(\cdot,\lambda)}{u_0(\cdot,\lambda)} \right
)'u_0^2(\cdot,\lambda) $$
to construct the second solution $u_1(x,\lambda)$, i.e., 
$$ u_1(x,\lambda)=u_0(x,\lambda)
\int_x^{\delta \lambda^{-1}}
u_0^{-2}(y,\lambda)dy $$
which implies
\begin{equation}
\label{eq_proofpertenergya1}
a_1(x,\lambda)=\frac{u_0(x)}{u_1(x)}(1+a_0(x,\lambda))\int_x^{\delta
\lambda^{-1}} u_0^{-2}(y)(1+\tilde{a}_0(y,\lambda))dy-1 
\end{equation}
where $\tilde{a}_0(y,\lambda):=(1+a_0(y,\lambda))^{-2}-1$ 
inherits the bound of $a_0$.
Now, by inserting the asymptotics of $u_0$ from Lemma \ref{lem_zeroenergy}, we
obtain 
$$ u_0(x)\int_x^\infty u_0^{-2}(y)dy \sim x^{-\ell} $$
for $x \to \infty$ and this shows 
$$ u_1(x)=u_0(x)\int_x^\infty u_0^{-2}(y)dy $$
since $u_1$ is \emph{uniquely} determined by the asymptotic behavior $u_1(x) \sim
x^{-\ell}$ for $x \to \infty$ (cf. Lemma \ref{lem_zeroenergy}).
Using this, Eq.~(\ref{eq_proofpertenergya1}) 
and the asymptotics of $a_0$, $\tilde{a}_0$, we obtain
\begin{align*} 
a_1(x,\lambda)&=\frac{u_0(x)}{u_1(x)}\left (
-\int_{\delta \lambda^{-1}}^\infty
u_0^{-2}(y)dy+a_0(x,\lambda)\int_x^{\delta \lambda^{-1}}u_0^{-2}(y)dy \right )\\
&+\frac{u_0(x)}{u_1(x)}
(1+a_0(x,\lambda))\int_x^{\delta \lambda^{-1}}u_0^{-2}(y)
\tilde{a}_0(y,\lambda)dy \\
&=O(x^{2\ell+1})\left
[O(\lambda^{2\ell+1})+O(x^2\lambda^2)(O(\lambda^{2\ell+1})+O(x^{-2\ell-1}))
+O(x^{-2\ell-1})O(x^2\lambda^2) \right ]\\
&=O(x^2\lambda^2)+O((x\lambda)^{2\ell+1})
\end{align*}
which implies the claim.
\end{proof}

\subsection{Estimates on the derivatives}

Next, we study derivatives of the above constructed solutions.

\begin{proposition}
\label{prop_pertenergy}
The functions $a_j$ for $j=0,1$ from Lemma \ref{lem_pertenergy} are of symbol
type, i.e., 
$$ |\partial_x^k \partial_\lambda^m a_j(x,\lambda)|\leq C_{k,m}
x^{2-k}\lambda^{2-m} $$
for $k,m \in \mathbb{N}_0$ and $\lambda \in (0,\lambda_0)$, 
$x \in [x_0,\delta \lambda^{-1}]$ provided that $\ell \geq 1$ 
where $x_0,\lambda_0,\delta>0$ are constants.
In the case $\ell=0$ we have the weaker bounds
$$ |\partial_x^k \partial_\lambda^m a_0(x,\lambda)|\leq C_{k,m}
x^{2-k}\lambda^{2-m} 
\mbox{ and }
|\partial_x^k \partial_\lambda^m a_1(x,\lambda)|\leq C_{k,m}
x^{1-k}\lambda^{1-m} $$
in the above ranges of $x$ and $\lambda$.
\end{proposition}

\begin{proof}
We use the notations from the proof of Lemma \ref{lem_pertenergy} and proceed by
induction. 
We have to consider mixed derivatives and
therefore, we 
order the set
$\mathbb{N}_0 \times \mathbb{N}_0$ according to 
$$(0,0), (1,0), (0,1), (2,0),
(1,1), (0,2), \dots$$
which defines a bijection $n: \mathbb{N}_0 \times
\mathbb{N}_0 \to \mathbb{N}_0$ by $n(0,0)=0, n(1,0)=1, n(0,1)=2$, etc.
Now fix $(k,m) \in \mathbb{N}_0 \times \mathbb{N}_0$ and assume that 
$|\partial_x^l \partial_\lambda^j (h(x,\lambda)-1)|\leq
C_{l,j}x^{2-l}\lambda^{2-j}$ holds for all $(l,j)$ with $n(l,j)\leq n(k,m)$ and
$\lambda \in (0,\lambda_0)$, $x \in [x_0, \lambda^{-1}]$.
We have to show that this implies 
$|\partial_x^{k'} \partial_\lambda^{m'} h(x,\lambda)|\leq
C_{k',m'}x^{2-k'}\lambda^{2-m'}$, where
$n(k',m')=n(k,m)+1$. 
There are two possibilities: Either $(k',m')=(m+1,0)$ (if $k=0$) or
$(k',m')=(k-1,m+1)$.
In the former case we have with 
$\kappa_l(x,\lambda):=\partial_x^l K(x,y,\lambda)|_{y=x}
=O(x^{1-l}\lambda^2)$
(cf. Lemma \ref{lem_volterradiffx}),
$$ \partial_x^{m+1} h(x,\lambda)=\sum_{l=0}^m
\partial_x^{m-l}\left [\kappa_l(x,\lambda)h(x,\lambda) \right ]+\int_{x_0}^x
\partial_x^{m+1} K(x,y,\lambda)h(y,\lambda)dy. $$
Now observe that by assumption 
$|\partial_x^{m-l} \left [\kappa_l(x,\lambda)h(x,\lambda)\right ]| 
\lesssim x^{2-(m+1)}\lambda^2$ 
for $\lambda \in (0,\lambda_0)$ and $x \in [x_0, \lambda^{-1}]$ and hence,
$|\partial_x^{m+1}h(x,\lambda)| \lesssim x^{2-(m+1)}\lambda^2$.
In the latter case we have, provided $k \geq 2$,
\begin{align*} 
|\partial_\lambda^{m+1}\partial_x^{k-1}h(x,\lambda)| & =
\left |\sum_{l=0}^{k-2}
\partial_\lambda^{m+1} 
\partial_x^{k-2-l}\left [\kappa_l(x,\lambda)h(x,\lambda) \right ] 
+\int_{x_0}^x
\partial_\lambda^{m+1}\left [\partial_x^{k-1} K(x,y,\lambda)
h(y,\lambda) \right ]dy \right | \\
& \lesssim x^{2-(k-1)}\lambda^{2-(m+1)} 
\end{align*}
by assumption. If $k=1$ we have
\begin{align*}
\partial_\lambda^{m+1} h(x,\lambda)&=\sum_{j=1}^{\min\{m+1,2\}} 
\left ( \begin{array}{c}m+1 \\
j \end{array} \right )\int_{x_0}^x
\partial_\lambda^j K(x,y,\lambda)\partial_\lambda^{m+1-j} h(y,\lambda)dy \\
&+\int_{x_0}^x K(x,y,\lambda)\partial_\lambda^{m+1} h(y,\lambda)dy
\end{align*}
and
thus, by assumption, the derivative $\partial_\lambda^{m+1}h(x,\lambda)$ 
satisfies a Volterra
equation of the form
$$ \partial_\lambda^{m+1}h(x,\lambda)=O(x^2 \lambda^{2-(m+1)})+
\int_{x_0}^x K(x,y,\lambda)\partial_\lambda^{m+1} h(y,\lambda)dy $$
and the basic estimate from Lemma \ref{lem_volterra} yields
$|\partial_\lambda^{m+1}h(x,\lambda)|\lesssim x^2\lambda^{2-(m+1)}$
for all $\lambda \in (0,\lambda_0)$ and $x \in [x_0,\lambda^{-1}]$.

For the second solution we use the representation
\begin{align}
\label{eq_proofproppertenergya1}
a_1(x,\lambda)&=\frac{u_0(x)}{u_1(x)}\left (
-\int_{\delta \lambda^{-1}}^\infty
u_0^{-2}(y)dy+a_0(x,\lambda)\int_x^{\delta \lambda^{-1}}u_0^{-2}(y)dy \right )
\\
&+\frac{u_0(x)}{u_1(x)}
(1+a_0(x,\lambda))\int_x^{\delta \lambda^{-1}}u_0^{-2}(y)
\tilde{a}_0(y,\lambda)dy \nonumber
\end{align}
from the proof of Lemma \ref{lem_pertenergy}
where, as before, 
$\tilde{a}_0(x,\lambda):=(1+a_0(x,\lambda))^{-2}-1$.
Lemma \ref{lem_symbolinverse} and the Leibniz rule show that $\tilde{a}_0$
inherits the bounds of $a_0$, i.e.,  $|\partial_x^k \partial_\lambda^m
\tilde{a}_0(x,\lambda)|\leq C_{k,m}x^{2-k}\lambda^{2-m}$ 
for all $k,m \in \mathbb{N}_0$ and $\lambda
\in (0,\lambda_0)$, $x \in [x_0,\lambda^{-1}]$. 
Thus, all functions on the right--hand side of Eq.
(\ref{eq_proofproppertenergya1}) behave like symbols under
differentiation with respect to $x$ and $\lambda$.
Therefore, as in the proof of Lemma \ref{lem_pertenergy}, we have
$$ a_1(x,\lambda)=O(x^2\lambda^2)+O((x\lambda)^{2\ell+1}) $$
where the $O$--terms are of symbol type which finishes the proof.
\end{proof}

\subsection{Refined bounds for $\lambda$--derivatives}

As a next step we prove a refinement of the estimates for the solution $u_0$
which shows that we can effectively trade $\lambda^{-1}$ for $x$ in the
bounds for the $\lambda$--derivatives of $a_0$.

\begin{lemma}
\label{lem_pertenergyest2}
The function $a_0(\cdot,\lambda)$, defined by 
$u_0(x,\lambda)=u_0(x)(1+a_0(x,\lambda))$, 
satisfies the
estimates
$$ |\partial_\lambda^{2m} a_0(x,\lambda)|\leq C_m x^{2m} 
\mbox{ and } |\partial_\lambda^{2m+1} a_0(x,\lambda)| \leq C_m x^{2m+2}\lambda $$
for all $\lambda \in [0,\lambda_0]$,
$x \in [x_0,\lambda^{-1}]$ and $m \in \mathbb{N}_0$
where $\lambda_0>0$ is a sufficiently small
constant.
\end{lemma}

\begin{proof}
We use the notations from the proof of Proposition \ref{prop_pertenergy}
and proceed by induction. 
The case $m=0$ has already been proved in Proposition \ref{prop_pertenergy}.
Now fix $m\in \mathbb{N}$ and
assume that $|\partial_\lambda^{2j} (h(x,\lambda)-1)|\leq C_j x^{2j}$ and
$|\partial_\lambda^{2j+1} (h(x,\lambda)-1)|\leq C_j x^{2j+2}\lambda$
for $j < m$ in the
above range of $\lambda$ and $x$.
According to the proof of Proposition \ref{prop_pertenergy}, we have
\begin{align*} 
\partial_\lambda^{2m} h(x,\lambda)&= 
\left ( \begin{array}{c}2m \\
2 \end{array} \right )\int_{x_0}^x
\partial_\lambda^2 K(x,y,\lambda)\partial_\lambda^{2m-2} h(y,\lambda)dy 
+2m \int_{x_0}^x
\partial_\lambda K(x,y,\lambda)\partial_\lambda^{2m-1} h(y,\lambda)dy \\
& +\int_{x_0}^x K(x,y,\lambda)\partial_\lambda^{2m} h(y,\lambda)dy \\
&=\int_{x_0}^x O(y)O(y^{2m-2})dy+\int_{x_0}^x O(y\lambda)O(y^{2m}\lambda)dy+
\int_{x_0}^x K(x,y,\lambda)\partial_\lambda^{2m} h(y,\lambda)dy \\
&=O(x^{2m})+\int_{x_0}^x K(x,y,\lambda)\partial_\lambda^{2m} h(y,\lambda)dy 
\end{align*}
by assumption and the estimate in Lemma \ref{lem_volterra} implies 
$|\partial_\lambda^{2m}h(x,\lambda)| \lesssim x^{2m}$ for 
$\lambda \in [0,\lambda_0]$ and $x \in [x_0,\lambda^{-1}]$. 
For the odd derivatives we proceed analogously and obtain
\begin{align*} 
\partial_\lambda^{2m+1} h(x,\lambda)&= 
\left ( \begin{array}{c}2m+1 \\
2 \end{array} \right )\int_{x_0}^x
\partial_\lambda^2 K(x,y,\lambda)\partial_\lambda^{2m-1} h(y,\lambda)dy \\
& +(2m+1) \int_{x_0}^x
\partial_\lambda K(x,y,\lambda)\partial_\lambda^{2m} h(y,\lambda)dy 
+\int_{x_0}^x K(x,y,\lambda)\partial_\lambda^{2m+1} h(y,\lambda)dy \\
&=\int_{x_0}^x O(y)O(y^{2m}\lambda)dy+\int_{x_0}^x O(y\lambda)O(y^{2m})dy+
\int_{x_0}^x K(x,y,\lambda)\partial_\lambda^{2m+1} h(y,\lambda)dy \\
&=O(x^{2m+2}\lambda)+\int_{x_0}^x K(x,y,\lambda)\partial_\lambda^{2m+1} h(y,\lambda)dy 
\end{align*}
by assumption and again, Lemma \ref{lem_volterra} yields the claim.
\end{proof}

Note that, by construction, we have $u_0(x_0,\lambda)=u_0(x_0)$ and
$u_0'(x_0,\lambda)=u_0'(x_0)$ (see the proof of Lemma \ref{lem_pertenergy}) 
which shows
in particular that
$u_0(x_0,\lambda)$ and $u_0'(x_0,\lambda)$ are smooth functions of $\lambda$.
Next, we prove similar bounds for the function $a_1$ but
unfortunately, the situation here is a bit more complicated.

\begin{lemma}
\label{lem_u1}
Let $u_1(x,\lambda)=u_1(x)(1+a_1(x,\lambda))$ and $x_0,\lambda_0,\delta>0$ 
be as in Lemma \ref{lem_pertenergy}.
Then, for all $\lambda \in (0,\lambda_0)$ and all 
$x \in [x_0, \delta \lambda^{-1}]$, we have the estimates
$$ |\partial_\lambda^{2m} a_1(x,\lambda)| \leq C_m x^{2m} 
\mbox{ and } 
|\partial_\lambda^{2m+1}a_1(x,\lambda)|\leq C_m x^{2m+2}\lambda $$
provided that $m \leq \ell-1$.
Furthermore, for higher derivatives we have the bounds
$$ |\partial_\lambda^{2\ell+m} a_1(x,\lambda)| \leq C_m x^{2\ell} 
\lambda^{-m} $$
for $m \in \mathbb{N}_0$ in the above ranges of $x$ and $\lambda$.
\end{lemma}

\begin{proof}
The function 
$a_1(x,\lambda)$ is given by
$$ 
a_1(x,\lambda)=\frac{u_0(x)}{u_1(x)}(1+a_0(x,\lambda))\int_x^{\delta
\lambda^{-1}} u_0^{-2}(y)(1+\tilde{a}_0(y,\lambda))dy-1 
$$
with $\tilde{a}_0(x,\lambda)=(1+a_0(x,\lambda))^{-2}-1$,
see Eq.~(\ref{eq_proofpertenergya1}).
Thus, in view of Lemma \ref{lem_pertenergyest2} it suffices to prove the claimed
bounds for
$$ h(x,\lambda):=\frac{u_0(x)}{u_1(x)}\int_x^{\delta
\lambda^{-1}} u_0^{-2}(y)(1+\tilde{a}_0(y,\lambda))dy-1. $$
For $k \in \mathbb{N}$ we have
\begin{align*}
\partial_\lambda^k
h(x,\lambda)&= 
-\frac{u_0(x)}{u_1(x)} \sum_{j=0}^{k-1}\frac{d^j}{d\lambda^j}\left
(u_0^{-2}(\delta
\lambda^{-1})\partial_\lambda^{k-1-j}(1+\tilde{a}_0(y,\lambda))|_{y=\delta
\lambda^{-1}}\delta \lambda^{-2} \right ) \\
&+ \frac{u_0(x)}{u_1(x)}\int_{x}^{\delta \lambda^{-1}}
u_0^{-2}(y)\partial_\lambda^k (1+\tilde{a}_0(y,\lambda))dy 
\end{align*}
and, by using the symbol behavior of $\tilde{a}_0$, we infer
$$ \frac{u_0(x)}{u_1(x)}\sum_{j=0}^{k-1}\frac{d^j}{d\lambda^j}\left
(u_0^{-2}(\delta
\lambda^{-1})\partial_\lambda^{k-1-j}(1+\tilde{a}_0(y,\lambda))|_{y=\delta
\lambda^{-1}}\delta \lambda^{-2} \right )
=O(x^{2\ell+1})O(\lambda^{2\ell+1-k}). $$
If $k=2m$ we have 
$ O(x^{2\ell+1})O(\lambda^{2\ell+1-k})=O(x^{2m}) $
provided that $m \leq \ell$.
If $k=2m+1$ we obtain
$$ O(x^{2\ell+1})O(\lambda^{2\ell+1-k})=O(x^{2\ell+1-2\ell+2m+1}\lambda)
=O(x^{2m+2}\lambda) $$
provided that $m \leq \ell-1$.
In order to estimate the integral term note that
$\tilde{a}_0$ inherits the bounds of $a_0$ from Lemma \ref{lem_pertenergyest2}
as a consequence of the Leibniz rule and Lemma \ref{lem_symbolinverse}.
For even derivatives we therefore have
\begin{align*} 
\frac{u_0(x)}{u_1(x)}\int_{x}^{\delta \lambda^{-1}}
u_0^{-2}(y)\partial_\lambda^{2m}
(1+\tilde{a}_0(y,\lambda))dy&=O(x^{2\ell+1})\int_{x}^{\delta \lambda^{-1}}
O(y^{-2\ell-2+2m})dy \\
&=O(x^{2\ell+1})(O(\lambda^{2\ell+1-2m})+O(x^{-2\ell-1+2m})) 
\end{align*}
and, provided that $m \leq \ell$, we obtain
$O(x^{2\ell+1})O(\lambda^{2\ell+1-2m})=O(x^{2m})$ 
as before.
For odd derivatives we use
$\partial_\lambda^{2m+1}\tilde{a}_0(y,\lambda)=O(y^{2m+2}\lambda)$ and infer
\begin{align*} 
\frac{u_0(x)}{u_1(x)}\int_{x}^{\delta \lambda^{-1}}
u_0^{-2}(y)\partial_\lambda^{2m+1}
(1+\tilde{a}_0(y,\lambda))dy&=O(x^{2\ell+1})\int_{x}^{\delta \lambda^{-1}}
O(y^{-2\ell-2+2m+2}\lambda)dy\\
&=O(x^{2\ell+1})(O(\lambda^{2\ell-2m})+O(x^{-2\ell-1+2m+2}\lambda)) 
\end{align*}
and $O(x^{2\ell+1})O(\lambda^{2\ell-2m})=O(x^{2m+2}\lambda)$ provided that $m
\leq \ell-1$.
The claim for the higher derivatives follows directly from the symbol behavior
of the above $O$--terms.
\end{proof}

At this point it is convenient to introduce a new notation.
\begin{definition}
\label{def_O}
For $N \in \mathbb{N}_0$ we write $f(x)=O_{2N}(1)$ if, for a constant $a>0$,
\begin{enumerate}
\item $f: (0,a) \to \mathbb{R}$ is smooth,
\item $|f^{(k)}(x)|\leq C_k$ for $k \leq 2N$ and all $x \in (0,a)$,
\item $|f^{(2N+k)}(x)|\leq C_k x^{-k}$ for all $k \in \mathbb{N}_0$ and $x \in
(0,a)$,
\item $\lim_{x \to 0+}f^{(2k-1)}(x)=0$ for all $1 \leq k \leq N$.
\end{enumerate}
Similarly, we write $f(x)=O_{2N+1}(x)$ if, for a constant $a>0$,
\begin{enumerate}
\item $f: (0,a) \to \mathbb{R}$ is smooth,
\item $|f^{(k)}(x)|\leq C_k$ for $k \leq 2N+1$ and all $x \in (0,a)$,
\item $|f^{(2N+1+k)}(x)|\leq C_k x^{-k}$ for all $k \in \mathbb{N}_0$ and $x \in
(0,a)$,
\item $\lim_{x \to 0+}f^{(2k)}(x)=0$ for all $0 \leq k \leq N$.
\end{enumerate}
Note carefully that the crucial difference between $O_{2N}$ and $O_{2N+1}$ is in condition $(4)$.
We also use the symbols $O_{2N+1}$ and $O_{2N}$ to denote \emph{generic
real--valued}
functions with the respective properties.
\end{definition}

\begin{corollary}
\label{cor_u1}
Let $u_1(x,\lambda)$ and $x_0, \lambda_0>0$ be as in Lemma \ref{lem_pertenergy}.
Then we have $u_1(x_0,\lambda)=O_{2\ell}(1)$ and
$u_1'(x_0,\lambda)=O_{2\ell}(1)$ for all $\lambda \in (0,\lambda_0)$.
\end{corollary}

\begin{proof}
The first assertion $u_1(x_0,\lambda)=O_{2\ell}(1)$ follows immediately from
Lemma \ref{lem_u1}.
For the second one note that $u_0(x_0)\not=0$ and by construction we have
$-1=u_0(x_0)u_1'(x_0,\lambda)-u_0'(x_0)u(x_0,\lambda)$ which implies
$$ u_1'(x_0,\lambda)=\frac{u_0'(x_0)u_1(x_0,\lambda)-1}{u_0(x_0)}=O_{2\ell}(1) $$ 
by the first part.
\end{proof}

\subsection{Extension of $u_j(x,\lambda)$ to negative values of $x$}

Finally, we extend the solutions $u_j(x,\lambda)$ to negative values of $x$ and
prove appropriate estimates. 

\begin{lemma}
\label{lem_pertenergyest2-}
The functions $u_j(x,\lambda)$,
$j=0,1$, from Lemma \ref{lem_pertenergy} can be
smoothly extended to $x \in [-\lambda^{-1},x_0]$ for $\lambda
\in (0,\lambda_0)$ where $\lambda_0>0$ is a constant.
Furthermore, in the above ranges of $x$ and $\lambda$, the function $u_0$
satisfies the bounds
$$ |\partial_\lambda^m u_0(x,\lambda)|\leq C_m \langle x \rangle^{m+1}, $$
$m \in \mathbb{N}_0$, whereas for $u_1$ we have the estimates
$$ |\partial_\lambda^m u_1(x,\lambda)|\leq C_m \langle x \rangle^{m+1} $$
if $m \leq 2\ell$ and
$$ |\partial_\lambda^{2\ell+m} u_1(x,\lambda)|\leq C_m 
\langle x \rangle^{2\ell+1}\lambda^{-m} $$
for $m \in \mathbb{N}_0$.
\end{lemma}

\begin{proof}
The coefficients of the equation $\mc{H}_{\sigma,\ell}f=\lambda^2 f$ are smooth
on $\mathbb{R}$ and thus, any solution of $\mc{H}_{\sigma,\ell}f=\lambda^2 f$
can be smoothly extended to all
of $\mathbb{R}$ by solving an appropriate initial value problem.
Applying the variation of constants formula and noting that
$u_0(x_0,\lambda)=u_0(x_0)$, $u_0'(x_0,\lambda)=u_0'(x_0)$ 
shows that the solution
$u_0(\cdot,\lambda)$
satisfies the integral equation
$$ u_0(x,\lambda)=u_0(x)+\lambda^2 \int_x^{x_0}\left
[v_0(x)v_1(y)-v_0(y)v_1(x) \right ]u_0(y,\lambda)dy, $$
see also Corollary \ref{cor_zeroenergy},
where $v_0$, $v_1$ are the smooth extensions to $(-\infty,x_0]$ of the 
functions constructed in Lemma \ref{lem_zeroenergy-}.
This is a Volterra equation with a kernel $K(x,y,\lambda):=\lambda^2
(v_0(x)v_1(y)-v_0(y)v_1(x))$.
According to Lemma \ref{lem_zeroenergy-} we have the bound
$|K(x,y,\lambda)|\lesssim \lambda^2(\langle x \rangle+\langle y \rangle)$
which implies
$$ \int_{-\lambda^{-1}}^{x_0} \sup_{x \in (-\lambda^{-1},x_0)} 
|K(x,y,\lambda)|dy \lesssim 1 $$
and hence, Lemma \ref{lem_volterra} and Corollary \ref{cor_zeroenergy} show that
$|u_0(x,\lambda)|\lesssim \langle x \rangle$ for all 
$\lambda \in (0,\lambda_0)$ and $x \in [-\lambda^{-1},x_0]$ where
$\lambda_0>0$ is the constant from Lemma \ref{lem_pertenergy}.
We proceed by induction. Fix $m \in \mathbb{N}_0$ and assume that we have 
$|\partial_\lambda^l u_0(x,\lambda)|\leq C_l \langle x \rangle^{l+1}$ 
for all $l \leq m$ in the above range of $\lambda$ and $x$.
This implies
\begin{align*}
\partial_\lambda^{m+1} u_0(x,\lambda)&=\sum_{l=1}^{\min\{m+1,2\}} 
\left ( \begin{array}{c}m+1 \\
l \end{array} \right )\int_x^{x_0}
\partial_\lambda^l K(x,y,\lambda)\partial_\lambda^{m+1-l} 
u_0(y,\lambda)dy \\ 
&+\int_x^{x_0} K(x,y,\lambda)\partial_\lambda^{m+1} u_0(y,\lambda)dy \\
&=\sum_{l=1}^{\min\{m+1,2\}} \underbrace{O(\lambda^{2-l} \langle x
\rangle^2)}_{O(\langle x \rangle^l)}
O(\langle x \rangle^{1+m+1-l})+
\int_x^{x_0} K(x,y,\lambda)\partial_\lambda^{m+1} u_0(y,\lambda)dy
\end{align*}
and the estimate from Lemma \ref{lem_volterra} yields
$|\partial_\lambda^{m+1}u_0(x,\lambda)|\lesssim \langle x \rangle^{m+2}$
as claimed.
Note carefully that the index $l$ in the last sum is at most equal to $2$ and therefore, we only estimate nonnegative powers of $\lambda$ here.

For the second assertion we proceed similarly and note that a 
straightforward calculation as well as the variation of constants formula
show that $u_1(\cdot,\lambda)$ satisfies the Volterra equation
\begin{align*}
u_1(x,\lambda)&=\frac{u_1(x_0,\lambda)v_1'(x_0)-u_1'(x_0,\lambda)v_1(x_0)}
{W(v_0,v_1)}v_0(x)
-\frac{u_1(x_0,\lambda)v_0'(x_0)-u_1'(x_0,\lambda)v_0(x_0)}
{W(v_0,v_1)}v_1(x) \\
&+\int_x^{x_0}K(x,y,\lambda)u_1(y,\lambda)dy. 
\end{align*}
According to Corollary \ref{cor_u1} we have
$$
u_1(x,\lambda)=O_{2\ell}(1)v_0(x)+O_{2\ell}(1)v_1(x)+\int_x^{x_0}K(x,y,\lambda)
u_1(y,\lambda)dy $$
and, since $v_j(x)=O(\langle x \rangle)$, we obtain $|u_1(x,\lambda)|\lesssim
\langle x \rangle$
for all $x \in [-\lambda^{-1},x_0]$ and $\lambda \in (0,\lambda_0)$ by Lemma
\ref{lem_volterra}.
Again, we proceed by induction and first we consider the case $m \leq 2\ell$.
If $\ell=0$ there is nothing left to prove, so we restrict ourselves to 
$\ell \geq 1$. 
Assuming that $|\partial_\lambda^l
u_1(x,\lambda)|\lesssim \langle x \rangle^{l+1}$ holds for all 
$l \leq m$
and a fixed $m \leq 2\ell-1$, we infer
\begin{align*}
\partial_\lambda^{m+1} u_1(x,\lambda)&=O(\langle x \rangle)
+\sum_{l=1}^{\min\{m+1,2\}} 
\left ( \begin{array}{c}m+1 \\
l \end{array} \right )\int_x^{x_0}
\underbrace{\partial_\lambda^l K(x,y,\lambda)\partial_\lambda^{m+1-l} 
u_1(y,\lambda)}_{O(\langle x \rangle)
O(\langle y \rangle^m)}dy 
\\
&+\int_x^{x_0} K(x,y,\lambda)\partial_\lambda^{m+1} u_1(y,\lambda)dy \\
&=O(\langle x \rangle^{m+2})+
\int_x^{x_0} K(x,y,\lambda)\partial_\lambda^{m+1} u_1(y,\lambda)dy
\end{align*}
and Lemma \ref{lem_volterra} yields 
$|\partial_\lambda^{m+1} u_1(x,\lambda)| \lesssim
\langle x \rangle^{m+2}$ for all $x \in [\lambda^{-1},x_0]$ and
$\lambda \in (0,\lambda_0)$.
The claim for the higher derivatives follows by a similar induction.
\end{proof}

\section{Perturbative solutions for $|x\lambda|$ large}

The solutions $u_j(\cdot,\lambda)$ obtained by perturbing in energy cannot
directly be matched with the Jost solution $f_+(\cdot,\lambda)$ since the
approximations in Lemma \ref{lem_pertenergy} are valid for $|x\lambda|$ 
small whereas the behavior
of $f_+(x,\lambda)$ is known only for fixed $\lambda$ and $x \to \infty$.
Thus, we construct another set of solutions to 
$\mc{H}_{\ell,\sigma} f=\lambda^2 f$ by
perturbing the potential.

\subsection{Construction of the perturbative solutions}

To do so, we first rescale the equation $\mc{H}_{\ell,\sigma} f=\lambda^2 f$ by
introducing a new independent variable $z:=\lambda x$.
Setting $\tilde{f}(z):=f(\lambda^{-1} z)$, the equation $\mc{H}_{\ell,\sigma} f=\lambda^2
f$ is equivalent to
\begin{equation}
\label{eq_rescaled}
\tilde{f}''+\left (1-\frac{\ell(\ell+1)}{z^2}\right
)\tilde{f}=\lambda^{-2}U_{\ell,\sigma}(\lambda^{-1}z)\tilde{f} 
\end{equation}
where $U_{\ell,\sigma}(x):=V_{\ell,\sigma}(x)-\frac{\ell(\ell+1)}{x^2}$.
Now observe that the equation
$$ \tilde{f}''+\left (1-\frac{\ell(\ell+1)}{z^2}\right
)\tilde{f}=0 $$
has the fundamental system $\{\sqrt{z}J_{\ell+1/2}(z),
\sqrt{z}Y_{\ell+1/2}(z)\}$ where $J_{\ell+1/2}$ and $Y_{\ell+1/2}$ are the
Bessel functions \cite{abramowitz}.
Thus, for a small right--hand side of Eq.~(\ref{eq_rescaled}) we expect to
obtain solutions of Eq.~(\ref{eq_rescaled}) 
by perturbing the Bessel functions.
According to Corollary \ref{cor_RWasymptotics}, the right--hand side satisfies
the estimate $\lambda^{-2}U_{\ell,\sigma}(\lambda^{-1}z) \lesssim
\lambda^{-2}(\lambda^{-1}z)^{-3+\varepsilon}=\lambda^{-2}x^{-3+\varepsilon}$ for
$x \to \infty$. 
Thus, our approximation is expected to be good if
$\lambda^{-2}x^{-3+\varepsilon}$ is small.
Smallness can be achieved by fixing $\lambda>0$ and letting $x \to \infty$ which
is required for the matching with the Jost solution $f_+$.
However, we can also enforce smallness by setting 
$x=\lambda^{-1+\varepsilon}$ (for a small $\varepsilon>0$) 
and letting $\lambda \to 0$. For $\lambda>0$ sufficiently
small, we have $|x \lambda|<\delta$ and the
matching with the
solutions $u_j(\cdot,\lambda)$ can be done as well. 
As a result, we obtain a good approximation to the Jost solution 
$f_+(x,\lambda)$ 
at a finite $x$ for $\lambda \to 0$.

\begin{lemma}
\label{lem_pertpot}
There exists a smooth solution $\phi_\ell(\cdot,\lambda)$ of 
Eq.~(\ref{eq_rescaled})
such that 
$$ \phi_\ell(z,\lambda)=\beta_\ell \sqrt{z}H_{\ell+1/2}^+(z)
(1+b_\ell(z,\lambda)) $$ 
where $H_{\ell+1/2}^+:=J_{\ell+1/2}+iY_{\ell+1/2}$ is the Hankel function and
$\beta_\ell:=i\sqrt{\frac{\pi}{2}}e^{i \ell \frac{\pi}{2}}$.
For all $\lambda \in (0,1)$, the function $b_\ell$ satisfies the bounds
$$ |b_\ell(z,\lambda)|\leq C \lambda^{1-\varepsilon(2\ell+3)}
$$
for all $z \in [\lambda^\varepsilon, 1]$ and
$$ |b_\ell(z,\lambda)|\leq C z^{-2+\varepsilon}\lambda^{1-\varepsilon}
$$
for all $z \in [1,\infty)$ where
$\varepsilon \in \left (0,\frac{1}{2\ell+3} \right )$ is arbitrary.
\end{lemma}

\begin{proof}
Set $\varphi_\ell(z):=\beta_\ell \sqrt{z}H_{\ell+1/2}^+(z)$ and observe that
$|\varphi_\ell(z)| > 0$ for all $z >0$.
Furthermore, we have $W(\varphi_\ell,\overline{\varphi_\ell})=-2i$ which 
follows by noting
that $\varphi_\ell(z) \sim e^{iz}$ and $\varphi_\ell'(z) \sim ie^{iz}$ for $z \to \infty$.
A straightforward calculation 
shows that, if $h$ satisfies the 
integral equation
\begin{equation}
\label{eq_proofpertpot}
h(z,\lambda)=1-\frac{1}{2i}\int_z^\infty \left [\varphi_\ell(y)\overline{\varphi_\ell(y)}
-\varphi_\ell^2(y)\frac{\overline{\varphi_\ell(z)}}{\varphi_\ell(z)} \right ] 
\lambda^{-2}U_{\ell,\sigma}(\lambda^{-1}y)h(y,\lambda)dy, 
\end{equation}
then $\phi_\ell(z,\lambda):=\varphi_\ell(z)h(z,\lambda)$ is a solution to Eq.
(\ref{eq_rescaled}).
Eq.~(\ref{eq_proofpertpot}) is of the form 
$$ h(z,\lambda)=1+\int_z^\infty K(z,y,\lambda)h(y,\lambda)dy. $$
Recall the asymptotic behavior
$\varphi_\ell(z)=c_1 \beta_\ell z^{\ell+1}(1+O(z))+i c_2 \beta_\ell 
z^{-\ell}(1+O(z))$ for $z \to 0$ 
where  
$c_1,c_2$ are nonzero real constants and the $O$--terms are smooth
(cf. \cite{abramowitz}).
Furthermore, we have $\varphi_\ell(z) \sim e^{iz}$ as $z \to \infty$ and
$|\lambda^{-2}U_{\ell,\sigma}(\lambda^{-1}z)| \lesssim
\lambda^{1-\varepsilon}z^{-3+\varepsilon}$ for all $\lambda,z>0$ with, say,
$\lambda^{-1}z \geq 1$ (see Corollary \ref{cor_RWasymptotics}).
This shows that, for $k \in \mathbb{N}_0$, 
$$ \left |\left (\frac{\overline{\varphi_\ell}}{\varphi_\ell} 
\right )^{(k)}(z) \right |\leq C_k $$ 
for all $z \geq 0$ and hence,
$|\partial_z^k K(z,y,\lambda)| \lesssim
\lambda^{1-\varepsilon}y^{-3+\varepsilon}(1+y^{-2\ell})$ for all $\lambda \in
(0,1)$ and all $y,z$ with $\lambda^\varepsilon \leq z \leq y$.
Thus, we have
$$ \mu(\lambda):=\int_{\lambda^\varepsilon}^\infty \sup_{z \in 
(\lambda^\varepsilon, y)}|K(z,y,\lambda)|dy \lesssim \lambda^{1-\varepsilon(2\ell+3)} $$
for all $\lambda \in (0,1)$ and hence, $\mu:=\sup_{\lambda \in
(0,1)}\mu(\lambda) < \infty$ provided that $\varepsilon \leq \frac{1}{2\ell+3}$.
Applying Lemma \ref{lem_volterradiffx} we conclude that Eq.~(\ref{eq_proofpertpot})
has a unique smooth solution $h(\cdot,\lambda)$ satisfying
$\|h(\cdot,\lambda)\|_{L^\infty(\lambda^\varepsilon,\infty)}\leq
e^\mu$ for all $\lambda \in (0,1)$.
Thus, Eq.~(\ref{eq_proofpertpot}) implies
$$ |h(z,\lambda)-1|\lesssim \int_z^\infty |K(z,y,\lambda)|dy \lesssim \int_z^\infty \lambda^{1-\varepsilon}y^{-3+\varepsilon}dy \lesssim \lambda^{1-\varepsilon}z^{-2+\varepsilon} $$
for all $\lambda \in (0,1)$ and all $z \in
[1,\infty)$.
Similarly, for $z \in [\lambda^\varepsilon,1]$, we have
\begin{align*} |h(z,\lambda)-1|&\lesssim \int_z^\infty |K(z,y,\lambda)|dy \lesssim \int_z^\infty \lambda^{1-\varepsilon}y^{-3+\varepsilon}(1+y^{-2\ell})dy \\
&\lesssim \lambda^{1-\varepsilon}\left (z^{-2+\varepsilon}+z^{-2-2\ell+\varepsilon} \right )
\lesssim \lambda^{1-\varepsilon(2\ell+3)}
\end{align*}
as claimed.
\end{proof}

\subsection{Estimates for the derivatives}

\begin{lemma}
\label{lem_pertpotsymbol}
For all $\lambda \in (0,1)$ and $m,k \in \mathbb{N}_0$, the function 
$b_\ell$ from Lemma \ref{lem_pertpot} satisfies the estimates
$$|\partial_z^k \partial_\lambda ^m b_\ell(z,\lambda)| \leq
C_{k,m}z^{-k}\lambda^{1-\varepsilon(2\ell+3)-m} $$
for all $z \in [\lambda^\varepsilon,1]$ and 
$$|\partial_z^k \partial_\lambda ^m b_\ell(z,\lambda)| \leq
C_{k,m}z^{-2+\varepsilon-k}\lambda^{1-\varepsilon-m} $$
for all $z \in [1,\infty)$ where
$\varepsilon \in \left 
(0,\frac{1}{2\ell+3} \right )$ is arbitrary.
\end{lemma}

\begin{proof}
The function $h:=1+b_\ell$ satisfies the equation
$$ h(z,\lambda)=1+\int_z^\infty K(z,y,\lambda)h(y,\lambda)dy $$
where 
$$ K(z,y,\lambda):=-\frac{1}{2i}\left [\varphi_\ell(y)\overline{\varphi_\ell(y)}
-\varphi_\ell^2(y)\frac{\overline{\varphi_\ell(z)}}{\varphi_\ell(z)} \right ] 
\lambda^{-2}U_{\ell,\sigma}(\lambda^{-1}y) $$
and $\varphi_\ell(z):=\beta_\ell \sqrt{z}H_{\ell+1/2}^+(z)$, see the proof of
Lemma \ref{lem_pertpot}.
Recall the asymptotic behavior 
$\varphi_\ell(z)=e^{iz}(1+O_\mathbb{C}(z^{-1}))$ for $z \to \infty$ 
where the $O_\mathbb{C}$--term is of symbol type.
On the other hand, we have 
$\varphi_\ell(z)=c_1 \beta_\ell z^{\ell+1}(1+O(z))+ic_2
\beta_\ell z^{-\ell}(1+O(z))$ for $z \to 0$ where $c_1,c_2$ are nonzero real
constants and the $O$--terms are smooth (see
\cite{abramowitz}). 
Thus, by the Leibniz rule and Lemma \ref{lem_symbolinverse}, 
we have  
$K(z,y,\lambda)=(1-e^{2i(y-z)})O_\mathbb{C}(y^{-3+\varepsilon}\lambda^{1-\varepsilon})$ 
for all $1 \leq z \leq y$ and
$K(z,y,\lambda)=O_\mathbb{C}(y^{-3-2\ell+\varepsilon}\lambda^{1-\varepsilon})
=O_\mathbb{C}(y^{-1}\lambda^{1-\varepsilon(2\ell+3)})$ for
$\lambda^\varepsilon \leq z
\leq y \leq 1$ where all 
$O$--terms are of symbol type.
Let $n: \mathbb{N}_0 \times \mathbb{N}_0 \to \mathbb{N}_0$ denote the bijection
from the proof of Proposition \ref{prop_pertenergy}. As before, we proceed by
induction.
Fix $(k,m) \in \mathbb{N}_0 \times \mathbb{N}_0$ and suppose we have 
$|\partial_z^l \partial_\lambda^j (h(z,\lambda)-1)|\leq
C_{l,j}z^{-l}\lambda^{1-\varepsilon(2\ell+3)-j}$ for all $(l,j)$ with $n(l,j)
\leq n(k,m)$ and $\lambda \in (0,1)$, $z \in [\lambda^\varepsilon,1]$.
We have to show that this implies 
$|\partial_z^{k'} \partial_\lambda^{m'}(h(z,\lambda)-1)|\leq
C_{k',m'}z^{-k'}\lambda^{1-\varepsilon(2\ell+3)-m'}$ where $n(k',m')=n(k,m)+1$.
If $k=0$ we have $(k',m')=(m+1,0)$ and, with $\kappa_l(z,\lambda):=\partial_z^l
K(z,y,\lambda)|_{y=z}$,
$$ \partial_z^{m+1} h(z,\lambda)=-\sum_{l=0}^m \partial_z^{m-l} \left
[\kappa_l(z,\lambda)h(z,\lambda) \right ]+\int_z^\infty
\partial_z^{m+1}K(z,y,\lambda)h(y,\lambda)dy, $$
see Lemma \ref{lem_volterradiffx}.
By assumption we have 
$\left | \partial_z^{m-l} \left
[\kappa_l(z,\lambda)h(z,\lambda) \right ] \right |\lesssim
z^{-(m+1)}\lambda^{1-\varepsilon(2\ell+3)}$ for $\lambda \in (0,1)$, $z \in
[\lambda^\varepsilon, 1]$ and this implies $|\partial_z^{m+1}
h(z,\lambda)|\lesssim z^{-(m+1)}\lambda^{1-\varepsilon(2\ell+3)}$.
If $k \geq 1$ we have $(k',m')=(k-1,m+1)$ and
\begin{equation}
\label{eq_proofpertpotsymbol}
\partial_\lambda^{m+1} \partial_z^{k-1} h(x,\lambda)=-\sum_{l=0}^{k-2}
\partial_\lambda^{m+1} \partial_z^{k-2-l} \left [\kappa_l(z,\lambda)h(z,\lambda)
\right ]+\int_z^\infty \partial_\lambda^{m+1} \left [
\partial_z^{k-1}K(z,y,\lambda)h(y,\lambda) \right ] dy.
\end{equation}
If $k \geq 2$, Eq.~(\ref{eq_proofpertpotsymbol}) and the assumption shows that
$|\partial_\lambda^{m+1} \partial_z^{k-1}h(z,\lambda)|\lesssim
z^{-(k-1)}\lambda^{1-\varepsilon(2\ell+3)-(m+1)}$ for $\lambda \in (0,1)$, $z
\in [\lambda^\varepsilon,1]$.
If $k=1$, Eq.~(\ref{eq_proofpertpotsymbol}) is of the form
$$
\partial_\lambda^{m+1}h(z,\lambda)=O_\mathbb{C}(\lambda^{1-\varepsilon(2\ell+3)-(m+1)})
+\int_z^\infty K(z,y,\lambda)\partial_\lambda^{m+1}h(y,\lambda)dy $$
by assumption and therefore, Lemma \ref{lem_volterra} yields 
$|\partial_\lambda^{m+1}h(z,\lambda)|\lesssim
\lambda^{1-\varepsilon(2\ell+3)-(m+1)}$ for $\lambda \in (0,1)$ and $z \in
[\lambda^\varepsilon,1]$.
This proves the first estimate.

For the second estimate we proceed by an analogous induction and write
\begin{align*} 
\int_z^\infty K(z,y,\lambda)h(y,\lambda)dy&=\int_0^\infty
K(z,\eta+z,\lambda)h(\eta+z,\lambda)d\eta \\
&=\int_0^\infty
(1-e^{2i\eta})O_\mathbb{C}((\eta+z)^{-3+\varepsilon}
\lambda^{1-\varepsilon})h(\eta+z,\lambda)d\eta
\end{align*}
for $z \geq 1$ where the $O_\mathbb{C}$--term is of symbol type.
Thus, the assumption yields
$$ \partial_z^{m+1} h(z,\lambda)=
O_\mathbb{C}(z^{-2+\varepsilon-(m+1)}\lambda^{1-\varepsilon})+\int_z^\infty
K(z,y,\lambda)\partial_y^{m+1}h(y,\lambda)dy $$
and Lemma \ref{lem_volterra} implies $|\partial_z^{m+1} h(z,\lambda)|\lesssim
z^{-2+\varepsilon-(m+1)}\lambda^{1-\varepsilon}$ for $\lambda \in (0,1)$ and $z
\geq 1$.
Analogously, we obtain
$$ \partial_\lambda^{m+1}
\partial_z^{k-1}h(z,\lambda)=O_\mathbb{C}(z^{-2+\varepsilon-(k-1)}
\lambda^{1-\varepsilon-(m+1)})+\int_z^\infty
K(z,y,\lambda)\partial_\lambda^{m+1}\partial_y^{k-1}h(y,\lambda)dy $$
and again, Lemma \ref{lem_volterra} yields
$|\partial_z^{k-1}\partial_\lambda^{m+1}h(z,\lambda)|\lesssim
z^{-2+\varepsilon-(k-1)}\lambda^{1-\varepsilon-(m+1)}$ for $\lambda \in (0,1)$
and $z \geq 1$ which finishes the proof.
\end{proof}

\section{Matching with the Jost solutions}

In this section we match the Jost functions $f_\pm(\cdot,\lambda)$ with the
solutions $u_j(\cdot,\lambda)$ which allows us to calculate the asymptotic behavior
of $f_\pm(\cdot,\lambda)$ and $W(f_-(\cdot,\lambda), f_+(\cdot,\lambda))$ 
in the limit $\lambda \to 0+$.

\subsection{Matching with $f_+(\cdot,\lambda)$}

Note that the solution $\phi_\ell$ constructed in Lemma \ref{lem_pertpot} 
has the asymptotic behavior
$\phi_\ell(z,\lambda) \sim e^{iz}$ for $z \to \infty$ which shows that 
$f_+(x,\lambda)=\phi_\ell(\lambda x,\lambda)$. 
Thus, we have found a representation of the Jost solution $f_+(x,\lambda)$ 
which is valid for
all $\lambda \in (0,1)$ and all $x$ with $\lambda^{-1+\varepsilon} \leq x <
\infty$ with a sufficiently small $\varepsilon>0$.
For given $\varepsilon, \delta>0$ 
we can always accomplish
$\lambda^{-1+\varepsilon} \leq \delta \lambda^{-1}$ 
for all $\lambda \in (0,\lambda_0)$ provided $\lambda_0>0$ is chosen small enough.
Thus, at $x=\lambda^{-1+\varepsilon}$ for $\lambda$ sufficiently close to
$0$, we can match the solutions
$f_+(\cdot,\lambda)$ and $u_j(\cdot,\lambda)$ (see Lemma
\ref{lem_pertenergy}). 

\begin{lemma}
\label{lem_c+}
The Wronskians
$c_j^+(\lambda):=W(f_+(\cdot,\lambda),u_j(\cdot,\lambda))$ for $j=0,1$ 
have the
asymptotic behavior
$$ c_0^+(\lambda)=i\alpha_0\beta_\ell \lambda^{-\ell}
(1+O(\lambda^\varepsilon)+iO(\lambda^{\varepsilon(2\ell+2)}))$$
and
$$ c_1^+(\lambda)=\alpha_1\beta_\ell \lambda^{\ell+1}(1+
O(\lambda^\varepsilon)+iO(\lambda^{-2\ell \varepsilon})) $$
as $\lambda \to 0+$ 
for a sufficiently small $\varepsilon>0$
where $\alpha_j$ are real nonzero constants and 
all $O$--terms are of symbol type.
\end{lemma}

\begin{proof}
We have $f_+(x,\lambda)=\beta_\ell \sqrt{\lambda x}H_{\ell+1/2}^+(\lambda
x)(1+b_\ell(\lambda x,\lambda))$ by Lemma \ref{lem_pertpot}.
Note that 
$\partial_x b_\ell(\lambda x,\lambda)=\lambda \partial_z
b_\ell(\lambda x,\lambda)$
as well as 
$\partial_\lambda b_\ell(\lambda x,\lambda)=
x \partial_z b_\ell(\lambda
x,\lambda)+\partial_\lambda b_\ell(\lambda x,\lambda)$ 
and 
hence, by Lemma \ref{lem_pertpotsymbol} and the chain rule, we infer
$$ |\partial_x^k \partial_\lambda^m b_\ell(\lambda x,\lambda)|\leq
C_{k,m}x^{-k}\lambda^{1-\varepsilon(2\ell+3)-m} $$
for all $k,m \in \mathbb{N}_0$.
By the same reasoning we obtain the symbol character (with respect to
differentiation in $x$ and $\lambda$) of the $O$--terms in 
$$ \beta_\ell \sqrt{\lambda x} H_{\ell+1/2}^+(\lambda x)=\alpha_1 \beta_\ell (\lambda x)^{\ell+1}(1+O(\lambda
x))+i\alpha_0 \beta_\ell (\lambda x)^{-\ell}(1+O(\lambda x)) $$
where $\alpha_0, \alpha_1$ are nonzero real constants.
This shows that 
$$ f_+(x,\lambda)=\left (\alpha_1 \beta_\ell (\lambda x)^{\ell+1}
+i\alpha_0
\beta_\ell (\lambda x)^{-\ell} \right
)(1+O_\mathbb{C}(\lambda^{1-\varepsilon(2\ell+3)}))(1+O(\lambda x))$$
where the $O$--terms are of symbol type and the representation is valid for all 
$\lambda \in (0,1)$ and all $x \in
[\lambda^{-1+\varepsilon},\lambda^{-1}]$.
Differentiating this expression with respect to $x$ and using the symbol
character of the $O$--terms we obtain
$$ f_+'(x,\lambda)=\left ((\ell+1) \alpha_1 \beta_\ell \lambda (\lambda x)^{\ell}
-i\ell \alpha_0
\beta_\ell \lambda (\lambda x)^{-\ell-1} \right )
(1+O_\mathbb{C}(\lambda^{1-\varepsilon(2\ell+3)}))(1+O(\lambda x)).$$ 
Evaluation at $x=\lambda^{-1+\varepsilon}$ yields
$$ f_+(\lambda^{-1+\varepsilon},\lambda)=
\left (\alpha_1 \beta_\ell \lambda^{\varepsilon(\ell+1)}
+i\alpha_0
\beta_\ell \lambda^{-\varepsilon \ell} \right
)\left (1+O(\lambda^{\varepsilon})+iO(\lambda^{1-\varepsilon(2\ell+3)})\right ) $$
and 
$$ f_+'(\lambda^{-1+\varepsilon},\lambda)=
\left ((\ell+1)\alpha_1 \beta_\ell \lambda^{1+\varepsilon \ell}
-i\ell \alpha_0
\beta_\ell \lambda^{1-\varepsilon (\ell+1)} \right
)\left (1+O(\lambda^\varepsilon)+iO(\lambda^{1-\varepsilon(2\ell+3)}) \right ) $$
for $\varepsilon>0$ sufficiently small.
Furthermore, by Lemma \ref{lem_pertenergy}, we have
$$ u_0(x,\lambda)=u_0(x)(1+O(x^2 \lambda^2)) 
\mbox{ , } 
u_0'(x,\lambda)=u_0'(x)(1+O(x^2 \lambda^2)) $$
and
$$ u_1(x,\lambda)=u_1(x)(1+O(x\lambda)) \mbox{ , }
u_1'(x,\lambda)=u_1'(x)(1+O(x \lambda)) $$
for all $\lambda \in (0,\lambda_0)$ and all 
$x \in [x_0,\delta \lambda^{-1}]$ where $\lambda_0, \delta>0$ are sufficiently
small and $x_0>0$ is sufficiently large.
Thus, by choosing $\lambda$ sufficiently close to $0$, we
obtain $\lambda^{-1+\varepsilon} \in [x_0,\delta \lambda^{-1}]$ and we can
evaluate the above expressions at $x=\lambda^{-1+\varepsilon}$ which yields
\begin{align*}
u_0(\lambda^{-1+\varepsilon},\lambda)&=(2\ell+1)^{-1}
\lambda^{-\ell-1+\varepsilon(\ell+1)}(1+O(\lambda^{2\varepsilon})), \\
u_0'(\lambda^{-1+\varepsilon},\lambda)&=(\ell+1)(2\ell+1)^{-1}
\lambda^{-\ell+\varepsilon\ell}(1+O(\lambda^{2\varepsilon})), \\
u_1(\lambda^{-1+\varepsilon},\lambda)&=
\lambda^{\ell-\varepsilon \ell}(1+O(\lambda^\varepsilon)), \\
u_1'(\lambda^{-1+\varepsilon},\lambda)&=
-\ell \lambda^{\ell+1-\varepsilon (\ell+1)}(1+O(\lambda^\varepsilon)) 
\end{align*}
by Lemma \ref{lem_zeroenergy} and all $O$--terms are of symbol type.
The claim now follows from a straightforward computation.
\end{proof}

\subsection{The Jost solution $f_-(\cdot,\lambda)$ in the limit $\lambda \to 0+$}
The Jost solution $f_-(\cdot,\lambda)$ satisfies the Volterra 
integral equation
$$ f_-(x,\lambda)=e^{-i \lambda x}-\int_{-\infty}^x
\frac{\sin(\lambda(y-x))}{\lambda}V_{\ell,\sigma}(y)f_-(y,\lambda)dy $$
as can be seen from the definition and the variation of constants formula.
The decay properties of the potential $V_{\ell,\sigma}$ are crucial for the behavior of
$f_\pm$. Since $V_{\ell,\sigma}$ decays exponentially as $x \to -\infty$ (Corollary
\ref{cor_RWasymptotics}), the situation for $f_-$ is much simpler.
In fact, $f_-$ behaves essentially as in the free case $V_{\ell,\sigma}=0$.  

\begin{lemma}
\label{lem_Jost-}
Let $a \in \mathbb{R}$ and $\lambda_0>0$. 
Then the Jost solution $f_-(x,\lambda)=e^{-i\lambda x}m_-(x,\lambda)$ 
exists for all 
$\lambda \in [-\lambda_0, \lambda_0]$ and, for $m \in \mathbb{N}_0$, we have the bounds
$|\partial_\lambda^m m_-(x,\lambda)|\leq C_m$ as well as
$|\partial_\lambda^m m_-'(x,\lambda)|\leq C_m$ for all 
$x \in (-\infty,a]$
and all $\lambda \in [-\lambda_0,\lambda_0]$.
\end{lemma}

\begin{proof}
According to Lemma \ref{lem_Jost}, the function $m_-(x,\lambda)=e^{i\lambda
x}f_-(x,\lambda)$ satisfies the Volterra equation
\begin{equation}
\label{eq_proofJost-}
m_-(x,\lambda)=1+\int_{-\infty}^x K(x,y,\lambda)m_-(y,\lambda)dy
\end{equation}
where
$K(x,y,\lambda):=
\frac{1}{2i\lambda}(e^{2i\lambda (y-x)}-1)V_{\ell,\sigma}(y)$.            
Note the bound 
$$ |\partial_\lambda^m K(x,y,\lambda)|\leq C_m |y-x|^{m+1}|V_{\ell,\sigma}(y)|
$$
and thus, Lemma \ref{lem_volterradifflambda} shows
that the solution of Eq.~(\ref{eq_proofJost-}) 
satisfies $\|\partial_\lambda^m m_-(\cdot,\lambda)\|_{L^\infty(-\infty,a)} \leq C_m$
for all $\lambda \in [-\lambda_0,\lambda_0]$ 
since $V_{\ell,\sigma}(y)$ decays exponentially as $y \to -\infty$
(Corollary \ref{cor_RWasymptotics}).
The estimate for the derivative $m_-'(x,\lambda)$ follows by 
differentiating Eq.
(\ref{eq_proofJost-}) and using the bounds for $\partial_\lambda^m
m_-(x,\lambda)$.
\end{proof}

\subsection{Zero energy resonances}
Next we discuss the occurrence of resonances.
First we give the precise definition of a zero energy resonance.

\begin{definition}
\label{def_resonance}
We say that the operator $\mc{H}_{\ell,\sigma}$ has a \emph{zero energy resonance} if
there exists a function $f \in L^\infty(\mathbb{R})$ such that 
$\mc{H}_{\ell,\sigma} f=0$.
\end{definition}
Recall that the equation $\mc{H}_{\ell,\sigma} f=0$ is equivalent to
$$ -\left ( 1-\frac{2M}{r} \right )v''-\frac{2M}{r^2}v'+\left 
( \frac{\ell(\ell+1)}{r^2}+\frac{2M\sigma}{r^3} \right )v=0 $$
where $v(r)=f(x(r))$ (see Eq.~(\ref{eq_RWzero})).
As already mentioned, solutions of this equation can be given in terms of 
special functions and
therefore, we even know the behavior of $f_-(x,0)$ for $x \to \infty$. 
This is
crucial to see whether the operator $\mc{H}_{\ell,\sigma}$ has a zero energy
resonance or not.
As the following lemma shows, no resonances occur for scalar perturbations.
However, in the case of electromagnetic or gravitational perturbations one 
has to require $\ell
\geq 1$ (which we do anyway) or $\ell \geq 2$, respectively, in order to avoid resonances.

\begin{lemma}
\label{lem_nonresonant}
Let $(\ell, \sigma) \in \mathbb{N}_0 \times \{-3,0,1\} 
\backslash \{(0,0),(0,-3),(1,-3)\}$. 
Then the zero energy Jost solution
has the asymptotic behavior $f_-(x,0) \sim cx^{\ell+1}$ for $x \to \infty$
where $c$ is a nonzero constant. 
In particular, there does not exist a function $f \in L^\infty(\mathbb{R})$ that
satisfies $\mc{H}_{\ell,\sigma} f=0$.
\end{lemma}

\begin{proof}
As already mentioned, the equation $\mc{H}_{\ell,\sigma} f=0$ is equivalent to 
$$ -\left ( 1-\frac{2M}{r} \right )v''-\frac{2M}{r^2}v'+\left 
( \frac{\ell(\ell+1)}{r^2}+\frac{2M\sigma}{r^3} \right )v=0 $$
where $v(r):=f(x(r))$.
Set $u(z):=z^{-(1+s)}v(2Mz)$ where $s:=\sqrt{1-\sigma}$ (note that $s$ is the spin
of the perturbing field, i.e.,  $s \in \{0,1,2\}$). 
Then the above equation is equivalent to
$$ z(1-z)u''+[c-(a+b+1)z]u'-abu=0 $$
where $a=-\ell+s$, $b=\ell+1+s$ and $c=1+2s$.
This is the hypergeometric differential equation and we have the 
solution $u(z)={}_2 F_1(a,b; c; z)={}_2F_1(-\ell+s, \ell+1+s; 1+2s; z)$ 
which reduces to a polynomial of order $\ell-s$ provided that $\ell-s
\in \mathbb{N}_0$ (see \cite{abramowitz}) and this is equivalent to 
$(\ell,\sigma) \notin \{(0,0),(0,-3),(1,-3)\}$.
The Frobenius indices for the hypergeometric differential equation at the
regular singular point $z=1$ (which corresponds to $r=2M$ and hence, $x \to
-\infty$) are $(0,c-a-b)=(0,0)$ (see \cite{abramowitz}) which shows that $u(1)
\not= 0$ and hence, there exists a nonzero constant $c_0$ such that
$f_-(x,0)=c_0 r(x)^{1+s}u(\frac{r(x)}{2M})$.
Since $u$ is a polynomial of order $\ell-s$, we obtain $f_-(x,0) \sim c_1
x^{\ell+1}$ for $x \to \infty$ by Lemma \ref{lem_r} where $c_1$ is a 
nonzero constant.
\end{proof}

\begin{remark}
Note that Lemma \ref{lem_nonresonant} is sharp in the sense that the operator
$\mc{H}_{\ell,\sigma}$ does indeed have zero energy resonances if $(\ell,\sigma) \in
\{(0,0),(0,-3),(1,-3)\}$.
The resonance functions $f_{\ell,\sigma}$ are given by $f_{0,0}(x)=1$, 
$f_{0,-3}(x)=1-\frac{3M}{r(x)}$ and
$f_{1,-3}(x)=\frac{1}{r(x)}$ as can be checked immediately.
\end{remark}

\begin{remark}
In the scalar case ($\sigma=1$), the hypergeometric function in the proof of
Lemma \ref{lem_nonresonant} reduces to the Legendre polynomial $P_\ell$ 
and we have
$$ f_-(x,0)=\frac{r(x)}{2M}P_\ell \left (\frac{r(x)}{M}-1 \right). $$ 
\end{remark}

\subsection{Matching with $f_-(\cdot,\lambda)$}

The above results are sufficient to match the Jost solution $f_-(\cdot,\lambda)$
to the solutions $u_j(\cdot,\lambda)$ obtained in Lemma
\ref{lem_pertenergy} by perturbing in energy.
In what follows we will always assume that we are in the nonresonant regime,
i.e.,  
$$ (\ell,\sigma)\notin \{(0,0),(0,-3),(1,-3)\}. $$

\begin{lemma}
\label{lem_c-}
The Wronskians $c_j^-(\lambda):=W(f_-(\cdot,\lambda),u_j(\cdot,\lambda))$
for $j=0,1$ are of the form \footnote{See Definition \ref{def_O}.}
$$ c_j^-(\lambda)=O_{2\ell}(1)+iO_{2\ell+1}(\lambda) $$
for $\lambda \in (0,\lambda_0)$ where $\lambda_0>0$ is a constant.
Furthermore, we have $c_1^-(0)\not=0$.
\end{lemma}

\begin{proof}
According to Lemma \ref{lem_Jost-}, $f_-(x,\lambda)$ is smooth in $\lambda$
around $\lambda=0$ and by definition we have
$f_-(x,\lambda)=\overline{f_-(x,-\lambda)}$ for $\lambda \in \mathbb{R}$.
In particular, this implies $\mathrm{Re}f_-(x_0,\lambda)=O_{2\ell}(1)$
and $\mathrm{Im}f_-(x_0,\lambda)=O_{2\ell+1}(\lambda)$ where $x_0>0$ is the
constant from Lemma \ref{lem_pertenergy}.
Repeating these arguments for the derivative $f_-'(x,\lambda)$, we similarly 
obtain
$\mathrm{Re}f_-'(x_0,\lambda)=O_{2\ell}(1)$ and
$\mathrm{Im}f_-'(x_0,\lambda)=O_{2\ell+1}(\lambda)$.
By construction (cf. Lemma \ref{lem_pertenergy}), we have
$u_0(x_0,\lambda)=u_0(x_0)$, $u_0'(x_0,\lambda)=u_0'(x_0)$.
Combining this with Corollary \ref{cor_u1} we obtain
$u_j(x_0,\lambda)=O_{2\ell}(1)$ and $u_j'(x_0,\lambda)=O_{2\ell}(1)$ 
for $j=0,1$.
This shows
$$
c_j^-(\lambda)=O_{2\ell}(1)(O_{2\ell}(1)+iO_{2\ell+1}(\lambda))
=O_{2\ell}(1)+iO_{2\ell+1}(\lambda). $$
Suppose $c_1^-(0)=W(f_-(\cdot,0),u_1(\cdot))=0$. This is equivalent to
$f_-(\cdot,0)$ and $u_1$ being linearly dependent which implies that
$f_-(x,0)=O(x^{-\ell})$ for $x \to \infty$, a contradiction to Lemma
\ref{lem_nonresonant}.
\end{proof}

\section{The spectral measure at zero energy}

Recall that we are interested in estimating the integral
$$ \int_0^\infty \lambda \cos(t \lambda)
\mathrm{Im}\left [ G_{\ell,\sigma}(x,x',\lambda) \right ] 
\langle x \rangle^{-\alpha} \langle x' 
\rangle^{-\alpha}
d\lambda $$
and thus, we have to study the expressions
$$ \mathrm{Im}\frac{f_-(x',\lambda)f_+(x,\lambda)}
{W(f_-(\cdot,\lambda),f_+(\cdot,\lambda))}. $$
In this section we obtain estimates 
for $\lambda \to 0+$ and, 
as will be clear afterwards, the decay properties for solutions of the
Regge--Wheeler equation are completely determined by this asymptotic behavior.
In view of this, the following lemma is in fact the central result of our work.
 
We have to consider different ranges of $x$, $x'$ and $\lambda$ separately 
and we
start with estimates for $|\lambda|$, $|x\lambda|$ and $|x' \lambda|$ small which
turns out to be the most important case.
For all $\lambda \in (0,\lambda_0)$ with a sufficiently small constant
$\lambda_0>0$, we have the representation
$$
f_\pm(x,\lambda)=-c_1^\pm(\lambda)u_0(x,\lambda)
+c_0^\pm(\lambda)u_1(x,\lambda) $$
where \footnote{From now on we write
$W(f_\pm,u_j)(\lambda)$ instead of
$W(f_\pm(\cdot,\lambda),u_j(\cdot,\lambda))$.}
$c_j^\pm(\lambda)=W(f_\pm,u_j)(\lambda)$ 
(see Lemmas \ref{lem_c+} and \ref{lem_c-}). 
Note carefully the slightly inconvenient fact that $c_0^\pm(\lambda)$ is the 
coefficient
of $u_1(\cdot,\lambda)$ and not $u_0(\cdot,\lambda)$!
It follows that 
$$W(f_-,f_+)(\lambda)=c_1^-(\lambda)c_0^+(\lambda)-
c_0^-(\lambda)c_1^+(\lambda).$$
We abbreviate
$$ A_{jk}(\lambda):=\mathrm{Im} \left [
\frac{c_j^-(\lambda)c_k^+(\lambda)}{
c_1^-(\lambda)c_0^+(\lambda)-
c_0^-(\lambda)c_1^+(\lambda)} \right ] $$
and, since $u_j(\cdot,\lambda)$ are real--valued,  
we have to study expressions of the form
$$ A_{00}(\lambda)u_1(x,\lambda)u_1(x',\lambda) \mbox{, }
A_{10}(\lambda)u_0(x,\lambda)u_1(x',\lambda) \mbox{, etc. } $$

\begin{lemma}
\label{lem_Ajk}
The function $A_{jk}$ is of the form
$$ A_{00}(\lambda)=O_{2\ell+1}(\lambda) \mbox{ and }
A_{jk}(\lambda)=O(\lambda^{2\ell+1}) \mbox{ if } j+k \geq 1 $$
for $\lambda \in (0,\lambda_0)$ where $\lambda_0>0$ is a sufficiently small
constant and the $O$--term behaves like a symbol.
\end{lemma}

\begin{proof}
We have to distinguish four cases.
\begin{enumerate}
\item For $A_{00}$ we write
$$
A_{00}(\lambda)=\mathrm{Im}\frac{\frac{c_0^-}{c_1^-}(\lambda)}
{1-\frac{c_0^-}{c_1^-}
\frac{c_1^+}{c_0^+}(\lambda)}. $$
According to Lemma \ref{lem_c-} we have 
$$ \frac{c_0^-}{c_1^-}(\lambda)=
\frac{c_0^-(\lambda)\overline{c_1^-(\lambda)}}{|c_1^-(\lambda)|^2}=
\frac{O_{2\ell}(1)+iO_{2\ell+1}(\lambda)}{|c_1^-(\lambda)|^2} $$ and, since
$|c_1^-(\lambda)|^2=O_{2\ell}(1)+O_{2\ell+2}(\lambda^2)=O_{2\ell}(1)$ as well as
$|c_1^-(0)|\not=0$, we infer
$\frac{c_0^-}{c_1^-}(\lambda)=O_{2\ell}(1)+iO_{2\ell+1}(\lambda)$ with the help
of Lemma \ref{lem_symbolinverse}.
Furthermore, Lemma \ref{lem_c+} and Lemma \ref{lem_symbolinverse} yield
$$\frac{c_1^+}{c_0^+}(\lambda)=-i\frac{\alpha_1}{\alpha_0}\lambda^{2\ell+1}
(1+O(\lambda^\varepsilon)+iO(\lambda^{-2\ell
\varepsilon}))=O(\lambda^{2\ell+1-2\ell\varepsilon})+iO(\lambda^{2\ell+1})$$
where the $O$--terms are of symbol type.
Applying Lemma \ref{lem_symbolinverse} again, we conclude
\begin{align*}
A_{00}(\lambda)&=\mathrm{Im}\frac{(O_{2\ell}(1)+iO_{2\ell+1}(\lambda))(
1+O(\lambda^{2\ell+1-2\ell\varepsilon})+iO(\lambda^{2\ell+1}))}{1+
O(\lambda^{2\ell+1-2\ell\varepsilon})} \\
&=O_{2\ell+1}(\lambda)+O(\lambda^{2\ell+1}) 
\end{align*}
for a sufficiently small $\varepsilon$ and the $O$--term is of symbol type.

\item For $A_{10}$ we use the representation
\begin{align*}
A_{10}(\lambda)&=\mathrm{Im}\frac{1}{1-\frac{c_0^-}{c_1^-}
\frac{c_1^+}{c_0^+}(\lambda)}=\mathrm{Im}\frac{1}
{1+O(\lambda^{2\ell+1-2\ell\varepsilon})+iO(\lambda^{2\ell+1})} \\
&=\mathrm{Im}\frac{1+O(\lambda^{2\ell+1-2\ell\varepsilon})+iO(\lambda^{2\ell+1})}{1+
O(\lambda^{2\ell+1-2\ell\varepsilon})}=O(\lambda^{2\ell+1})
\end{align*}
where all $O$--terms are of symbol type (use Lemma \ref{lem_symbolinverse}).

\item In order to estimate $A_{01}$ first note that
$$ W(f_-,f_+)(\lambda)=c_1^-(\lambda)c_0^+(\lambda)-c_0^-(\lambda)c_1^+(\lambda)
=i\alpha_0 \beta_\ell c \lambda^{-\ell}
(1+O(\lambda^\varepsilon)+iO(\lambda^{\varepsilon(2\ell+2)})) $$
as follows straightforward from Lemmas \ref{lem_c+} and \ref{lem_c-}
where $c$ is a nonzero real constant and all $O$--terms are of symbol type.
However, this implies
\begin{align*}
A_{01}(\lambda)&=\mathrm{Im}\frac{i\frac{\alpha_1}{\alpha_0 c}\lambda^{2\ell+1}
(O(1)+iO(\lambda))(1+O(\lambda^\varepsilon)+iO(\lambda^{-2\ell \varepsilon}))}
{1+O(\lambda^\varepsilon)+iO(\lambda^{\varepsilon(2\ell+2)})}\\
&=\mathrm{Im}\frac{i\frac{\alpha_1}{\alpha_0 c}\lambda^{2\ell+1}(O(1)+O(\lambda^\varepsilon)
+iO(\lambda^{-2\ell\varepsilon}))(1+O(\lambda^\varepsilon)
+iO(\lambda^{\varepsilon(2\ell+2)}))}{1+O(\lambda^\varepsilon)} \\
&=\frac{O(\lambda^{2\ell+1})}{1+O(\lambda^\varepsilon)}=O(\lambda^{2\ell+1})
\end{align*}
where all $O$--terms are of symbol type (see Lemma \ref{lem_symbolinverse}).
\item Finally, for $A_{11}$ we proceed exactly as above and obtain
$$ A_{11}(\lambda)=\mathrm{Im}\frac{i\frac{\alpha_1}{\alpha_0 c}\lambda^{2\ell+1}
(O(1)+iO(\lambda))(1+O(\lambda^\varepsilon)+iO(\lambda^{-2\ell \varepsilon}))}
{1+O(\lambda^\varepsilon)+iO(\lambda^{\varepsilon(2\ell+2)})}
=O(\lambda^{2\ell+1}) $$
where the $O$--term behaves like a symbol.
\end{enumerate}
\end{proof}

\begin{remark}
The fact that $A_{00}(\lambda)$ is somewhat exceptional is a direct consequence
of the asymmetric decay properties of the Regge--Wheeler potential. This
phenomenon is not present in \cite{schlag1} or \cite{schlag2}.
\end{remark}

\section{Oscillatory integral estimates for small energies}
\label{sec_osc}
In this section we obtain bounds for the oscillatory integrals 
that describe the
time evolution of solutions to the Regge--Wheeler equation.
We distinguish different regimes, depending on the ranges of $x$, $x'$ and 
$\lambda$ and in this section we only consider the case $|\lambda|$ small.
As already mentioned, the most important contribution comes from
the regime $|x\lambda|$ and
$|x'\lambda|$ small which yields the decay rates stated in Theorem
\ref{thm_main}.
The remaining cases can be treated
very similar to \cite{schlag2}, 
however, for the sake of completeness we give explicit proofs for all of the
following statements. 

\subsection{Estimates for $|x \lambda|$ and $
|x' \lambda|$ small}
 
We will need the following result on oscillatory integrals.

\begin{lemma}
\label{lem_osc}
For an $N \in \mathbb{N}_0$ let $\omega(\lambda)=O_{2N+1}(\lambda)$ and
suppose there exists a constant $\lambda_0>0$ such that 
$\omega(\lambda)=0$ for all $\lambda \geq \lambda_0$.
Then we have the estimates
\begin{align*} 
\left | \int_0^\infty \lambda \cos(t \lambda)\omega(\lambda)d\lambda 
\right | & \leq C(\omega)
\langle t \rangle^{-(2N+3)} \\
\left | \int_0^\infty \sin(t \lambda)\omega(\lambda)d\lambda 
\right |  & \leq C(\omega)
\langle t \rangle^{-(2N+2)} 
\end{align*}
for all $t \geq 0$ where $C(\omega)$ can be estimated as
$$ C(\omega)\leq C \max \left \{ \|\omega\|_{L^\infty(0,\infty)}, \:
\sup_{\lambda>0}|\lambda^j
\omega^{(2N+1+j)}(\lambda)|: j=1,2,3 \right \} $$
for an absolute constant $C>0$.
\end{lemma}

\begin{proof}
We only prove the sine estimate since the proof for the cosine estimate
is completely analogous.
It suffices to consider $t\geq 1$.
($2N+2$)--fold integration by parts yields
\begin{align*}
\left | \int_0^\infty \sin(t\lambda)\omega(\lambda)d\lambda \right
|&=\left |\frac{1}{t^{2N+2}} \int_0^\infty
\sin(t\lambda)\omega^{(2N+2)}(\lambda)d\lambda \right |
\end{align*}
since the boundary terms vanish thanks to $\omega^{(2m)}(0)=0$ for $m \leq
N$ and the fact that $\omega(\lambda)=0$ for all $\lambda \geq \lambda_0$.
Thus, it suffices to show that
$$ \left |\int_0^\infty \sin(t\lambda)
\omega^{(2N+2)}(\lambda)d\lambda \right |\leq C
$$ for a constant $C$ independent of $t$.
Let $\chi$ be a smooth cut--off satisfying $\chi(\lambda)=1$ for $0 \leq \lambda
\leq \frac{1}{2}$ and $\chi(\lambda)=0$ for $\lambda \geq 2$.
Then we have
$$
\left | \int_0^\infty \sin(t\lambda)
\omega^{(2N+2)}(\lambda)\chi(t\lambda)d\lambda \right | =
\left | \int_0^\infty \sin(\eta)
\omega^{(2N+2)}\left (\frac{\eta}{t} \right )\chi(\eta)\frac{d\eta}{t} \right |
\lesssim
\int_0^\infty \left
|\frac{\sin(\eta)}{\eta}\chi(\eta) \right |d\eta \lesssim 1
$$
for all $t \geq 1$.
Furthermore, by an additional integration by parts we obtain
\begin{align*}
\left | \int_0^\infty \sin(t\lambda)
\omega^{(2N+2)}(\lambda)[1-\chi(t\lambda)]d\lambda \right |&\leq 
\left |\frac{1}{t}\int_0^\infty
\cos(t\lambda)\omega^{(2N+3)}(\lambda)[1-\chi(t\lambda)]d\lambda \right | \\
&+ \left |\frac{1}{t}\int_0^\infty \cos(t\lambda)\omega^{(2N+2)}(\lambda)
t\chi'(t\lambda)d\lambda
\right |
\end{align*}
where the boundary term vanishes thanks to the cut--off and
$\omega^{(2N+2)}(\lambda)=0$ for all $\lambda \geq \lambda_0$.
However, we have
\begin{align*} \left |\frac{1}{t}\int_0^\infty
\cos(t\lambda)\omega^{(2N+3)}(\lambda)[1-\chi(t\lambda)]d\lambda \right |& =
\left |\frac{1}{t}\int_0^\infty
\cos(\eta)\omega^{(2N+3)}\left( \frac{\eta}{t} \right)[1-\chi(\eta)]
\frac{d\eta}{t} \right | \\
&\lesssim  \frac{1}{t^2}\int_0^\infty
\left |\cos(\eta)\frac{t^2}{\eta^2} [1-\chi(\eta)]\right |d\eta \lesssim 1
\end{align*}
as well as
\begin{align*}
\left |\frac{1}{t}\int_0^\infty \cos(t\lambda)\omega^{(2N+2)}t\chi'(t\lambda)d\lambda
\right |&=\left |\int_0^\infty \cos(\eta)
\omega^{(2N+2)}\left (\frac{\eta}{t}\right )\chi'(\eta)\frac{d\eta}{t}
\right | \\
&\lesssim \int_0^\infty \left |\frac{\cos \eta}{\eta}\chi'(\eta) \right|
d\eta \lesssim 1
\end{align*}
since $\mathrm{supp}(\chi') \subset [\frac{1}{2},2]$.
\end{proof}

Now we are ready to prove the first oscillatory integral estimate, valid for
small $\lambda$ and $|x\lambda| \leq \delta$, $|x' \lambda| \leq \delta$
where $\delta>0$ is sufficiently small.
In what follows we denote by $\chi_\delta$ a smooth cut--off function supported 
in a $\delta$--neighborhood of the origin, i.e.,  
$$ \chi_\delta(x)=\left \{ \begin{array}{l}1 \mbox{ if } |x|\leq
\frac{\delta}{2} \\
0 \mbox{ if } |x| \geq \delta \end{array} \right . .$$

\begin{lemma}
\label{lem_osc1}
Let $\alpha \geq 2\ell+1$ 
and $\delta>0$ be sufficiently small.
Then we have the estimates
$$ \sup_{x,x' \in \mathbb{R}} \left | \int_0^\infty
\lambda
\cos(t\lambda)\mathrm{Im} \left [\frac{f_-(x',\lambda)f_+(x,\lambda)}
{W(f_-(\cdot,\lambda),f_+(\cdot,\lambda))} \right ] \langle x
\rangle^{-\alpha}\langle x' \rangle^{-\alpha}\chi_\delta(\lambda)\chi_\delta  
(x\lambda)\chi_\delta (x'\lambda)d\lambda \right |
\lesssim \langle t \rangle^{-(2\ell+3)}
$$
and
$$ \sup_{x,x' \in \mathbb{R}} \left | \int_0^\infty
\sin(t\lambda)\mathrm{Im}\left [ \frac{f_-(x',\lambda)f_+(x,\lambda)}
{W(f_-(\cdot,\lambda),f_+(\cdot,\lambda))} \right ] \langle x
\rangle^{-\alpha}\langle x' \rangle^{-\alpha}\chi_\delta(\lambda)\chi_\delta  
(x\lambda)\chi_\delta (x'\lambda)d\lambda \right |
\lesssim \langle t \rangle^{-(2\ell+2)}
$$
for all $t \geq 0$.
\end{lemma}
 
\begin{proof}
We set $\omega(x,x',\lambda):=A_{00}(\lambda)u_1(x,\lambda)
u_1(x',\lambda)\langle x
\rangle^{-\alpha}\langle x'
\rangle^{-\alpha}\chi_\delta(\lambda)\chi_\delta(x\lambda)
\chi_\delta(x'\lambda)$.
According to Lemmas \ref{lem_Ajk} and \ref{lem_u1}, we have
$\omega(x,x',\lambda)=O_{2\ell+1}(\lambda)$ for fixed $x,x'$.
Combining Lemmas \ref{lem_pertenergy}, \ref{lem_u1} and
\ref{lem_pertenergyest2-}, we obtain 
$|\partial_\lambda^{2\ell+m} u_1(x,\lambda)|\leq C_m \langle x
\rangle^{2\ell+1}\lambda^{-m}$ for $m \in \mathbb{N}_0$, $\lambda \in
(0,\delta)$ and $x \in [-\delta \lambda^{-1},\delta \lambda^{-1}]$.
This implies 
$$|\partial_\lambda^{2\ell+2} \omega(x,x',\lambda)|\lesssim \langle
x \rangle^{2\ell+1-\alpha}\langle x' \rangle^{2\ell+1-\alpha}\lambda^{-1}, $$
and, analogously,
$$|\partial_\lambda^{2\ell+3} [
\omega(x,x',\lambda)]|\lesssim \langle
x \rangle^{2\ell+1-\alpha}\langle x' \rangle^{2\ell+1-\alpha}\lambda^{-2}, $$
$$|\partial_\lambda^{2\ell+4} [
\omega(x,x',\lambda)]|\lesssim \langle
x \rangle^{2\ell+1-\alpha}\langle x' \rangle^{2\ell+1-\alpha}\lambda^{-3} $$
for all $\lambda > 0$ and $x,x' \in \mathbb{R}$.
Finally, $\omega(\lambda)=0$ for $\lambda \geq \delta$.
Thus, Lemma \ref{lem_osc} yields
$$ \sup_{x,x' \in \mathbb{R}}
\left |\int_0^\infty \lambda \cos(t\lambda)\omega(x,x',\lambda)d\lambda 
\right |\lesssim \langle t \rangle^{-(2\ell+3)} 
\mbox{ and }
\sup_{x,x' \in \mathbb{R}}\left |\int_0^\infty \sin(t\lambda)\omega(x,x',\lambda)d\lambda 
\right |\lesssim \langle t \rangle^{-(2\ell+2)} $$
for all $t \geq 0$.
For the remaining cases assume $j+k\geq 1$ and set
$$ \omega(x,x',\lambda):=A_{jk}(\lambda)u_{j'}(x,\lambda)u_{k'}(x',\lambda)
\langle x
\rangle^{-\alpha}\langle x'
\rangle^{-\alpha}\chi_\delta(\lambda)\chi_\delta(x\lambda)
\chi_\delta(x'\lambda). $$
According to Lemmas \ref{lem_Ajk} and \ref{lem_pertenergy}, we have
$\omega(x,x',\lambda)=O(\lambda^{2\ell+1})O(\langle x \rangle^{\ell+1})O(\langle
x' \rangle^{\ell+1})$ where the $O$--terms behave like symbols (use Proposition
\ref{prop_pertenergy} and Lemma \ref{lem_pertenergyest2-}).
In particular, $\omega(x,x',\lambda)=O_{2\ell+1}(\lambda)$ for fixed $x,x'$ and 
$$|\partial_\lambda^{2\ell+2} \omega(x,x',\lambda)|\lesssim \langle
x \rangle^{\ell+1-\alpha}\langle x' \rangle^{\ell+1-\alpha}\lambda^{-1}, $$
as well as
$$|\partial_\lambda^{2\ell+3} [
\omega(x,x',\lambda)]|\lesssim \langle
x \rangle^{\ell+1-\alpha}\langle x' \rangle^{\ell+1-\alpha}\lambda^{-2}, $$
$$|\partial_\lambda^{2\ell+4} [
\omega(x,x',\lambda)]|\lesssim \langle
x \rangle^{\ell+1-\alpha}\langle x' \rangle^{\ell+1-\alpha}\lambda^{-3} $$
for all $\lambda > 0$ and $x,x' \in \mathbb{R}$.
Thus, as before, applying Lemma \ref{lem_osc} yields the claim.
\end{proof}

\begin{remark}
Obviously, by performing fewer integrations by parts (cf. the proof of Lemma
\ref{lem_osc}), one may obtain weaker decay
bounds (in $t$). 
By doing so, however, one can relax the decay requirements (in $x$) 
of the data,
that is, the index $\alpha$ in Lemma \ref{lem_osc1} can be chosen smaller.
To be more precise, one obtains the additional bounds, valid for $\alpha \in
\mathbb{N}$, $\alpha \leq 2\ell+1$,
$$ \sup_{x,x' \in \mathbb{R}} \left | \int_0^\infty
\lambda
\cos(t\lambda)\mathrm{Im} \left [\frac{f_-(x',\lambda)f_+(x,\lambda)}
{W(f_-(\cdot,\lambda),f_+(\cdot,\lambda))} \right ] \langle x
\rangle^{-\alpha}\langle x' \rangle^{-\alpha}\chi_\delta(\lambda)\chi_\delta  
(x\lambda)\chi_\delta (x'\lambda)d\lambda \right |
\lesssim \langle t \rangle^{-\alpha-2}
$$
and
$$ \sup_{x,x' \in \mathbb{R}} \left | \int_0^\infty
\sin(t\lambda)\mathrm{Im}\left [ \frac{f_-(x',\lambda)f_+(x,\lambda)}
{W(f_-(\cdot,\lambda),f_+(\cdot,\lambda))} \right ] \langle x
\rangle^{-\alpha}\langle x' \rangle^{-\alpha}\chi_\delta(\lambda)\chi_\delta  
(x\lambda)\chi_\delta (x'\lambda)d\lambda \right |
\lesssim \langle t \rangle^{-\alpha-1}
$$
for all $t \geq 0$.
\end{remark}

\begin{remark}
 Note that the sine estimate from Lemma \ref{lem_osc1} is the main obstacle to proving better decay. The remaining oscillatory estimates of Sec.~\ref{sec_osc} and Sec.~\ref{sec_osclarge} below are consistent with faster decay.
\end{remark}

\subsection{Estimates for $|x\lambda|$ and
$|x'\lambda|$ large}
 
For the remaining small energy contributions 
it is useful to note that, for $\lambda \in \mathbb{R}$,
$f_\pm(x,-\lambda)=\overline{f_\pm(x,\lambda)}$ by
definition of the Jost solutions.
This implies $G_{\ell,\sigma}(x,x',-\lambda)=
\overline{G_{\ell,\sigma}(x,x',\lambda)}$ and hence, the real part 
$\mathrm{Re}\left [ G_{\ell,\sigma}(x,x',\lambda) \right ]$ is an
even function of $\lambda$ whereas the imaginary part
$\mathrm{Im} \left [ G_{\ell,\sigma}(x,x',\lambda) \right ]$ is odd.
Thus, we have
$$ \int_0^\infty \lambda
\cos(t\lambda)\mathrm{Im}\left [G_{\ell,\sigma}(x,x',\lambda) \right ]
d\lambda=\frac{1}{2}
\int_\mathbb{R}\lambda \cos(t\lambda)G_{\ell,\sigma}(x,x',\lambda)d\lambda $$
and similarly for the sine evolution.
This shows that we can replace the imaginary part of 
$G_{\ell,\sigma}(x,x',\lambda)$ by $
G_{\ell,\sigma}(x,x',\lambda)$ itself in the oscillatory integrals and 
change the domain of
integration from $\lambda>0$ to $\lambda \in \mathbb{R}$.
Furthermore, recall
\begin{equation}
\label{eq_wronskian}
W(f_-,f_+)(\lambda)=c_1^-(\lambda)
c_0^+(\lambda)-
c_0^-(\lambda)c_1^+(\lambda)=c\lambda^{-\ell}(1+O_\mathbb{C}(\lambda^{\varepsilon})) 
\end{equation} for a
nonzero constant $c$ where the $O$--term is of symbol
type. This has been shown in the proof of Lemma \ref{lem_Ajk}.

In order to deal with terms that
involve $f_-(x',\lambda)$ for $x' \geq 0$ and $f_+(x,\lambda)$ for $x \leq 0$
we have to consider reflection and transmission coefficients.
For $\lambda \not=0$, the functions $f_+(\cdot,\lambda)$ and
$\overline{f_+(\cdot,\lambda)}$ are linearly independent which shows that
there exist coefficients $a(\lambda)$ and $b(\lambda)$ \footnote{In order to avoid confusion, we remark that the coefficients $a$ and $b$ are not exactly the same as $A$ and $B$ in the proof of Lemma \ref{lem_Wronskian} but they are related by $a(\lambda)=A(\lambda^2)$ and $b(\lambda)=B(\lambda^2)$.} such that
$f_-(x,\lambda)=a(\lambda)f_+(x,\lambda)+b(\lambda)\overline{f_+(x,\lambda)}
$.
This representation implies $|b(\lambda)|^2-|a(\lambda)|^2=1$ (cf. the proof of
Lemma \ref{lem_Wronskian}) and thus,
$f_+(x,\lambda)=-\overline{a(\lambda)}f_-(x,\lambda)+b(\lambda)
\overline{f_-(x,\lambda)}$.
Furthermore, we have
$W(f_-,f_+)(\lambda)=b(\lambda)W(\overline{f_+},f_+)(\lambda)=2i\lambda b(\lambda)$ which
is equivalent to
\begin{equation}
\label{eq_b}
\frac{b(\lambda)}{W(f_-,f_+)(\lambda)}=\frac{1}{2i\lambda}.
\end{equation}
Similarly, we obtain
$ W(f_-,\overline{f_+})(\lambda)=-2i\lambda a(\lambda)$ and therefore,
$$ \frac{a(\lambda)}{W(f_-,f_+)(\lambda)}=-
\frac{W(f_-,\overline{f_+})(\lambda)}{2i\lambda
W(f_-,f_+)(\lambda)}. $$
However, from Lemma \ref{lem_c+} and
$W(f_-,\overline{f_+})(\lambda)=c_1^-(\lambda)\overline{c_0^+(\lambda)}-
c_0^-(\lambda)\overline{c_1^+(\lambda)}$ it follows that 
$$
\frac{W(f_-,\overline{f_+})(\lambda)}{W(f_-,f_+)(\lambda)}=
c+O_\mathbb{C}(\lambda^\varepsilon) $$
where $c$ is a nonzero constant and the $O$--term is of symbol type
and hence,
\begin{equation}
\label{eq_a}
\frac{a(\lambda)}{W(f_-,f_+)(\lambda)}=c\lambda^{-1}(1+
O_\mathbb{C}(\lambda^\varepsilon)).
\end{equation}

\begin{lemma}
\label{lem_osc2}
Let $\alpha \in \mathbb{N}_0$ and $\delta>0$ sufficiently small. Then we have the
estimates
$$ \sup_{x,x' \in \mathbb{R}} \left | \int_\mathbb{R}
\lambda \cos(t\lambda)\frac{f_-(x',\lambda)f_+(x,\lambda)}{W(f_-,f_+)(\lambda)}
\langle x
\rangle^{-\alpha}\langle x' \rangle^{-\alpha}
\chi_\delta(\lambda)(1-\chi_\delta  (x\lambda))
(1-\chi_\delta (x'\lambda))d\lambda \right |
\lesssim \langle t \rangle^{-\alpha}
$$
and
$$ \sup_{x,x' \in \mathbb{R}} \left | \int_\mathbb{R}
\sin(t\lambda)\frac{f_-(x',\lambda)f_+(x,\lambda)}{W(f_-,f_+)(\lambda)}
\langle x
\rangle^{-\alpha}\langle x' \rangle^{-\alpha}
\chi_\delta(\lambda)(1-\chi_\delta(x\lambda))
(1-\chi_\delta (x'\lambda))d\lambda \right |
\lesssim \langle t \rangle^{-\alpha+1}
$$
for all $t \geq 0$.
\end{lemma}

\begin{proof}
Let $|\lambda|\leq \delta$, $x \geq 0$, $x' \leq 0$ and $|\lambda x|, |\lambda
x'|\geq \frac{\delta}{2}$.
We set $m_\pm(x,\lambda):=e^{\mp i\lambda x}f_\pm(x,\lambda)$.
According to Lemma \ref{lem_Jost-}, we have the bound $|\partial_\lambda^m
m_-(x,\lambda)|\leq C_m$ for all $x\leq 0$ and $m\in \mathbb{N}_0$.
Furthermore, since $\lambda$ is small, we have 
$$ m_+(x,\lambda)=e^{-i\lambda x}\phi_\ell(\lambda x,\lambda)=
(1+O_\mathbb{C}((\lambda x)^{-1}))(1+b_\ell(\lambda x,\lambda)) $$
by Lemma \ref{lem_pertpot} and the asymptotics of the Hankel function where the
$O$--term is of symbol type.
By Lemma \ref{lem_pertpotsymbol} and the chain rule, we have the estimate
$|\partial_\lambda^m b_\ell(\lambda x,\lambda)|\leq C_m \langle x \rangle^m$ 
since $|\lambda|^{-1} \lesssim \langle x \rangle$
and this implies
$|\partial_\lambda^m m_+(x,\lambda)| \leq C_m \langle x \rangle^m$
for all $m \in \mathbb{N}_0$.
Set
$$ \omega(x,x',\lambda):=
\lambda \frac{m_-(x',\lambda)m_+(x,\lambda)}{W(f_-,f_+)(\lambda)}
\chi_\delta(\lambda)(1-\chi_\delta  
(x\lambda))(1-\chi_\delta (x'\lambda)).$$
Then we have
$$ |\partial_\lambda^m \omega(x,x',\lambda)|\lesssim \langle x
\rangle^m \langle x'
\rangle^m $$
for all $m \in \mathbb{N}_0$
since
$W(f_-,f_+)(\lambda)$ is of symbol type (cf. Eq.~(\ref{eq_wronskian})).
We have to estimate the integral
$$ \int_\mathbb{R}e^{i \lambda (\pm t+x-x')}\omega(x,x',\lambda)
\langle x \rangle^{-\alpha} \langle x' \rangle^{-\alpha}d\lambda. $$
If $|\pm t+x-x'|\geq \frac{1}{2}t$ we integrate by parts 
$\alpha$--times to obtain
$$ \left | 
\int_\mathbb{R}e^{i \lambda (\pm t+x-x')}\omega(x,x',\lambda)
\langle x \rangle^{-\alpha}\langle x' \rangle^{-\alpha}d\lambda \right
|\lesssim |\pm t+x-x'|^{-\alpha} \lesssim \langle t \rangle^{-\alpha}
$$
and, 
if $|\pm t+x-x'|\leq \frac{1}{2}t$, we have 
$\langle x \rangle^{-\alpha}\langle x'
\rangle^{-\alpha}\lesssim \langle t \rangle^{-\alpha}$ as $t \to \infty$ and
thus,
$$ \left | 
\int_\mathbb{R}e^{i \lambda (\pm t+x-x')}\omega(x,x',\lambda)
\langle x \rangle^{-\alpha}\langle x' \rangle^{-\alpha}d\lambda \right
|\lesssim \langle t \rangle^{-\alpha}. $$
If $x' \geq 0$ or $x \leq 0$ we use the representations
$m_-(x',\lambda)=a(\lambda)e^{2i\lambda x'}
m_+(x',\lambda)+b(\lambda)\overline{m_+(x',\lambda)}$ or
$m_+(x,\lambda)=-\overline{a(\lambda)}e^{-2i\lambda
x}m_-(x,\lambda)+b(\lambda)\overline{m_-(x,\lambda)}$ and
with the help of Eqs.
(\ref{eq_b}) and (\ref{eq_a}) the
corresponding integrals can be estimated as above.

For the sine evolution note that we are missing one $\lambda$ and thus, 
for instance, if
$$
\omega(x,x',\lambda):=\frac{a(\lambda)m_+(x',\lambda)m_+(x,\lambda)}
{W(f_-,f_+)(\lambda)}\chi_\delta(\lambda)(1-\chi_\delta  
(x\lambda))(1-\chi_\delta (x'\lambda)) $$
we have $|\omega(x,x',\lambda)|\lesssim |\lambda|^{\ell-1}$ 
(cf.~Eq.~(\ref{eq_wronskian})) which, in the case
$\ell=0$, only yields the weaker bound
$|\partial_\lambda^m \omega(x,x',\lambda)|\lesssim
\langle x \rangle^{m+1}\langle x' \rangle^m$.
\end{proof}

\subsection{Estimates for $|x\lambda|$ small and $|x'\lambda|$
large}

The next regime to be considered is $|x\lambda|\leq \delta$ and $|x'\lambda|
\geq \frac{\delta}{2}$.
We distinguish the two cases $x' \leq 0$ and $x' \geq 0$. 

\begin{lemma}
Let $\alpha \in \mathbb{N}$ and $\delta>0$ sufficiently small. Then we have the
estimates
$$ \sup_{x \in \mathbb{R},x' < 0} \left | \int_\mathbb{R}
\lambda \cos(t\lambda)\frac{f_-(x',\lambda)f_+(x,\lambda)}{W(f_-,f_+)(\lambda)}
\langle x
\rangle^{-\alpha}\langle x' \rangle^{-\alpha}
\chi_\delta(\lambda)\chi_\delta (x\lambda)(1-\chi_\delta  (x'\lambda))
d\lambda \right |
\lesssim \langle t \rangle^{-\alpha}
$$
and
$$ \sup_{x \in \mathbb{R},x' < 0} \left | \int_\mathbb{R}
\sin(t\lambda)\frac{f_-(x',\lambda)f_+(x,\lambda)}{W(f_-,f_+)(\lambda)}
\langle x
\rangle^{-\alpha}\langle x' \rangle^{-\alpha}\chi_\delta(\lambda)
\chi_\delta (x\lambda)(1-\chi_\delta  
(x'\lambda))
d\lambda \right |
\lesssim \langle t \rangle^{-\alpha+1}
$$
for all $t \geq 0$.
\end{lemma}

\begin{proof}
Let $x \in \mathbb{R}$, $x' \leq 0$, $|\lambda|\leq \delta$, $|\lambda x|\leq
\delta$ and $|\lambda x'|\geq \frac{\delta}{2}$.
As before, we write $f_\pm(x,\lambda)=e^{\mp i\lambda x}m_\pm(x,\lambda)$ 
and recall the
bound $|\partial_\lambda^m m_-(x',\lambda)|\leq C_m$ for all $m \in
\mathbb{N}_0$ (Lemma
\ref{lem_Jost-}). 
For
$m_+$ we use the representation
$$m_+(x,\lambda)=e^{-i\lambda x}(-c_1^+(\lambda)u_0(x,\lambda)+
c_0^+(\lambda)u_1(x,\lambda))$$
where we extend $c_j^+(\lambda)$ and $u_j(x,\lambda)$ to negative $\lambda$
according to $c_j^+(-\lambda)=\overline{c_j^+(\lambda)}$ and
$u_j(x,-\lambda)=u_j(x,\lambda)$.
Applying Proposition \ref{prop_pertenergy} and Lemmas 
\ref{lem_pertenergyest2-}, \ref{lem_c+},
we obtain the bounds
$$ |\partial_\lambda^m c_1^+(\lambda)u_0(x,\lambda)|\lesssim \langle x
\rangle^{\ell+1}|\lambda|^{\ell-m}\lesssim \langle x \rangle |\lambda|^{-m} 
\lesssim \langle x' \rangle^m |\lambda|^{-1} $$
and 
$$ |\partial_\lambda^m c_0^+(\lambda)u_1(x,\lambda)|\lesssim \langle x
\rangle |\lambda|^{-\ell-m}\lesssim \langle x'
\rangle^m |\lambda|^{-\ell-1} $$
for $m \in \mathbb{N}_0$ which implies
$|\partial_\lambda^m m_+(x,\lambda)| \lesssim \langle x' \rangle^m
|\lambda|^{-\ell-1}$.
We also have $|\partial_\lambda^m \chi_\delta(x\lambda)|\lesssim \langle x
\rangle^m \lesssim |\lambda|^{-m}\lesssim \langle x' \rangle^m$ and, putting all
this together, we arrive at
$$ |\partial_\lambda^m \omega(x,x',\lambda)|\lesssim \langle x' \rangle^m $$
where
$$ \omega(x,x',\lambda):=\lambda
\frac{m_-(x',\lambda)m_+(x,\lambda)}{W(f_-,f_+)(\lambda)}
\chi_\delta(\lambda)\chi_\delta (x\lambda)(1-\chi_\delta  
(x'\lambda)), $$
see also Eq.~(\ref{eq_wronskian}).
Thus, the 
claim follows by appropriate integration by parts as in the proof of Lemma
\ref{lem_osc2}.
The proof for the sine evolution goes along the same lines but one loses one
power of $\lambda$.
\end{proof}

\begin{lemma}
Let $\alpha \in \mathbb{N}$ and $\delta>0$ sufficiently small. Then we have the
estimates
$$ \sup_{x \in \mathbb{R},x' > 0} \left | \int_\mathbb{R}
\lambda \cos(t\lambda)\frac{f_-(x',\lambda)f_+(x,\lambda)}{W(f_-,f_+)(\lambda)}
\langle x
\rangle^{-\alpha}\langle x' \rangle^{-\alpha}
\chi_\delta(\lambda)\chi_\delta (x\lambda)(1-\chi_\delta  (x'\lambda))
d\lambda \right |
\lesssim \langle t \rangle^{-\alpha}
$$
and
$$ \sup_{x \in \mathbb{R},x' > 0} \left | \int_\mathbb{R}
\sin(t\lambda)\frac{f_-(x',\lambda)f_+(x,\lambda)}{W(f_-,f_+)(\lambda)}
\langle x
\rangle^{-\alpha}\langle x' \rangle^{-\alpha}\chi_\delta(\lambda)
\chi_\delta (x\lambda)(1-\chi_\delta  
(x'\lambda))
d\lambda \right |
\lesssim \langle t \rangle^{-\alpha+1}
$$
for all $t \geq 0$.
\end{lemma}

\begin{proof}
Let $x \in \mathbb{R}$, $x' \geq 0$, $|\lambda|\leq \delta$, $|x\lambda|\leq
\delta$ and $|x'\lambda|\geq \frac{\delta}{2}$.
As always, we write $f_\pm(x,\lambda)=e^{\mp i\lambda x}m_\pm(x,\lambda)$.
Again, by Lemmas \ref{lem_pertenergy}, \ref{lem_c+}, \ref{lem_pertenergyest2-} 
and the representation
$$ m_+(x,\lambda)=e^{-i\lambda
x}(-c_1^+(\lambda)u_0(x,\lambda)+c_0^+(\lambda)u_1(x,\lambda))$$
we obtain 
the bound
$$ |m_+(x,\lambda)| \lesssim \langle x \rangle^{\ell+1}|\lambda|^{\ell}+\langle
x \rangle |\lambda|^{-\ell} \lesssim \langle x \rangle |\lambda|^{-\ell} $$
and, by using the symbol behavior of the involved terms (see Proposition
\ref{prop_pertenergy} and Lemma \ref{lem_c+}), this implies
$$ |\partial_\lambda^m m_+(x,\lambda)|\lesssim \langle x \rangle
|\lambda|^{-\ell-m}\lesssim \langle x \rangle \langle x' \rangle^m
|\lambda|^{-\ell}. $$
For $m_-(x',\lambda)$ we use reflection and transmission coefficients, i.e.,
$$ m_-(x',\lambda)=a(\lambda)e^{2i\lambda
x'}m_+(x',\lambda)+b(\lambda)\overline{m_+(x',\lambda)} $$
and, from the proof of Lemma \ref{lem_osc2}, we have $|\partial_\lambda^m
m_+(x',\lambda)|\lesssim \langle x' \rangle^m$.
Hence, Eqs.~(\ref{eq_a}) and (\ref{eq_b}) show that
$$ |\partial_\lambda^m m_-(x',\lambda)|\lesssim |\lambda|^{-1-m}+
\langle x' \rangle^m|\lambda|^{-1} \lesssim \langle x' \rangle^m |\lambda|^{-1}.
$$ 
Setting
$$
\omega(x,x',\lambda):=\lambda\frac{m_-(x',\lambda)m_+(x,\lambda)}
{W(f_-,f_+)(\lambda)}\chi_\delta(\lambda)\chi_\delta (x\lambda)(1-\chi_\delta  
(x'\lambda)) $$
the above estimates and Eq.~(\ref{eq_wronskian}) imply 
$|\partial_\lambda^m \omega(x,x',\lambda)| \lesssim
\langle x \rangle \langle x' \rangle^m$ and the cosine estimate follows by
appropriate integration by parts as in the proof of Lemma \ref{lem_osc2}.
For the sine estimate we set
$$
\omega(x,x',\lambda):=\frac{m_-(x',\lambda)m_+(x,\lambda)}
{W(f_-,f_+)(\lambda)}\chi_\delta(\lambda)\chi_\delta (x\lambda)(1-\chi_\delta  
(x'\lambda)) $$
and we only have the weaker bound $|\partial_\lambda^m \omega(x,x',\lambda)| \lesssim
\langle x\rangle \langle x' \rangle^m |\lambda|^{-1} \lesssim 
\langle x \rangle \langle x' \rangle^{m+1}$.
\end{proof}

\subsection{Estimates for $|x'\lambda|$ small and $|x\lambda|$
large}

Due to the asymmetric decay of the Regge--Wheeler potential, this case is
slightly different from the above considered $|x\lambda|$ small and
$|x'\lambda|$ large.
Thus, it has to be studied separately and does not follow from symmetry
arguments as in
\cite{schlag2}.
Similar as above, we distinguish $x \geq 0$ and $x \leq 0$.

\begin{lemma}
Let $\alpha \in \mathbb{N}$ and $\delta>0$ sufficiently small. Then we have the
estimates
$$ \sup_{x > 0,x' \in \mathbb{R}} \left | \int_\mathbb{R}
\lambda \cos(t\lambda)\frac{f_-(x',\lambda)f_+(x,\lambda)}{W(f_-,f_+)(\lambda)}
\langle x
\rangle^{-\alpha}\langle x' \rangle^{-\alpha}
\chi_\delta(\lambda)
(1-\chi_\delta (x\lambda))\chi_\delta  (x'\lambda)d\lambda \right |
\lesssim \langle t \rangle^{-\alpha}
$$
and
$$ \sup_{x > 0,x' \in \mathbb{R}} \left | \int_\mathbb{R}
\sin(t\lambda)\frac{f_-(x',\lambda)f_+(x,\lambda)}{W(f_-,f_+)(\lambda)}
\langle x
\rangle^{-\alpha}\langle x' \rangle^{-\alpha}\chi_\delta(\lambda)
(1-\chi_\delta (x\lambda))\chi_\delta  
(x'\lambda)d\lambda \right |
\lesssim \langle t \rangle^{-\alpha+1}
$$
for all $t \geq 0$.
\end{lemma} 

\begin{proof}
Let $x \geq 0$, $x' \in \mathbb{R}$, $|\lambda|\leq \delta$, $|x'\lambda|\leq
\delta$ and $|x\lambda|\geq \frac{\delta}{2}$.
As in the proof of Lemma \ref{lem_osc2}, we have the bounds $|\partial_\lambda^m
m_+(x,\lambda)| \lesssim \langle x \rangle^m$.
For $m_-(x',\lambda)$ we use the representation
$$ m_-(x',\lambda)=e^{i\lambda
x'}(-c_1^-(\lambda)u_0(x',\lambda)+c_0^-(\lambda)u_1(x',\lambda)) $$
since $|x'\lambda|$ is small.
Lemmas \ref{lem_pertenergy}, \ref{lem_pertenergyest2-} and 
\ref{lem_c-} imply the bound
$$ |m_-(x',\lambda)|\lesssim \langle x' \rangle^{\ell+1}+\langle x' \rangle
\lesssim |\lambda|^{-\ell-1} $$
and, by using the symbol behavior (see Proposition \ref{prop_pertenergy} and
Lemmas \ref{lem_pertenergyest2}, \ref{lem_c-}) and $\langle x' \rangle \lesssim
|\lambda|^{-1}$, we infer
$$ |\partial_\lambda^m m_-(x',\lambda)|\lesssim |\lambda|^{-\ell-1-m} \lesssim
\langle x \rangle^m |\lambda|^{-\ell-1}. $$
We set
$$ \omega(x,x',\lambda):=\lambda \frac{m_-(x',\lambda)m_+(x,\lambda)}
{W(f_-,f_+)(\lambda)}\chi_\delta(\lambda)(1-\chi_\delta  
(x\lambda))\chi_\delta (x'\lambda) $$
and the above bounds as well as Eq.~(\ref{eq_wronskian}) imply the estimate
$|\partial_\lambda^m \omega(x,x',\lambda)|\lesssim \langle x \rangle^m$.
Thus, as before, the claim follows by appropriate integration by parts.
\end{proof}

\begin{lemma}
\label{lem_osc3a}
Let $\alpha \in \mathbb{N}$ and $\delta>0$ sufficiently small. Then we have the
estimates
$$ \sup_{x < 0,x' \in \mathbb{R}} \left | \int_\mathbb{R}
\lambda \cos(t\lambda)\frac{f_-(x',\lambda)f_+(x,\lambda)}{W(f_-,f_+)(\lambda)}
\langle x
\rangle^{-\alpha}\langle x' \rangle^{-\alpha}
\chi_\delta(\lambda)
(1-\chi_\delta (x\lambda))\chi_\delta  (x'\lambda)d\lambda \right |
\lesssim \langle t \rangle^{-\alpha}
$$
and
$$ \sup_{x < 0,x' \in \mathbb{R}} \left | \int_\mathbb{R}
\sin(t\lambda)\frac{f_-(x',\lambda)f_+(x,\lambda)}{W(f_-,f_+)(\lambda)}
\langle x
\rangle^{-\alpha}\langle x' \rangle^{-\alpha}\chi_\delta(\lambda)
(1-\chi_\delta (x\lambda))\chi_\delta  
(x'\lambda)d\lambda \right |
\lesssim \langle t \rangle^{-\alpha+1}
$$
for all $t \geq 0$.
\end{lemma} 

\begin{proof}
Let $x \leq 0$, $x' \in \mathbb{R}$, $|\lambda|\leq \delta$, $|x'\lambda|\leq
\delta$ and $|x\lambda|\geq \frac{\delta}{2}$.
Like in the proof of Lemma \ref{lem_osc3a} we have
$$ |\partial_\lambda^m m_-(x',\lambda)| \lesssim \langle x'
\rangle^{\ell+1}|\lambda|^{-m} \lesssim \langle x' \rangle |\lambda|^{-\ell-m}
$$ for all $m \in \mathbb{N}_0$.
Since $x \leq 0$ we use reflection and transmission coefficients to obtain the
representation
$$ m_+(x,\lambda)=-\overline{a(\lambda)}e^{-2i\lambda
x}m_-(x,\lambda)+b(\lambda)\overline{m_-(x, \lambda)} $$
which immediately implies the bound
$|m_+(x,\lambda)| \lesssim |\lambda|^{-1}$
by Eqs.~(\ref{eq_a}), (\ref{eq_b}) and Lemma \ref{lem_Jost-}.
Thus, from the symbol behavior of $a(\lambda)$, $b(\lambda)$ and
$|\partial_\lambda^m m_-(x,\lambda)|\leq C_m$ (Lemma \ref{lem_Jost-}), we infer
$$ |\partial_\lambda^m m_+(x,\lambda)|\lesssim \langle x \rangle^m
|\lambda|^{-1} $$
for all $m \in \mathbb{N}_0$ since $|\lambda|^{-1} \lesssim \langle x \rangle$.
Thus, Eq.~(\ref{eq_wronskian}) implies
$$ |\partial_\lambda^m \omega(x,x',\lambda)| \lesssim \langle x \rangle^m
\langle x' \rangle $$
for all $m \in \mathbb{N}_0$ where
$$ \omega(x,x',\lambda):=\lambda \frac{m_-(x',\lambda)m_+(x,\lambda)}
{W(f_-,f_+)(\lambda)}\chi_\delta(\lambda)(1-\chi_\delta  
(x\lambda))\chi_\delta (x'\lambda). $$
As a consequence, by appropriate integration by parts (cf. the proof of Lemma
\ref{lem_osc2}), we obtain
$$ \left | \int_\mathbb{R} e^{i\lambda (\pm t + x - x')}
\omega(x,x',\lambda)\langle x \rangle^{-\alpha}
\langle x' \rangle^{-\alpha}d\lambda \right |\lesssim \langle t
\rangle^{-\alpha} $$
provided that $\alpha \geq 1$.
For the sine evolution we set
$$ \omega(x,x',\lambda):=\frac{m_-(x',\lambda)m_+(x,\lambda)}
{W(f_-,f_+)(\lambda)}\chi_\delta(\lambda)(1-\chi_\delta  
(x\lambda))\chi_\delta (x'\lambda) $$
and use the bounds
$$ |\partial_\lambda^m m_+(x,\lambda)| \lesssim \langle x\rangle^m
|\lambda|^{-1} \lesssim \langle x \rangle^{m+1} $$ 
to obtain
$$ |\partial_\lambda^m \omega(x,x',\lambda)| \lesssim \langle x 
\rangle^{m+1}\langle x' \rangle $$
for all $m \in \mathbb{N}_0$.
Hence, as before, the claim follows from
$$ \left | \int_\mathbb{R} e^{i\lambda (\pm t + x - x')}\omega(x,x',\lambda)\langle x \rangle^{-\alpha}
\langle x' \rangle^{-\alpha}d\lambda \right |\lesssim \langle t
\rangle^{-\alpha+1} $$
which can be obtained by appropriate integration by parts 
similar to the proof of 
Lemma \ref{lem_osc2}.
\end{proof}

\section{Oscillatory integral estimates for large energies}
\label{sec_osclarge}

\subsection{The Jost solutions at large energies}
In order to estimate the contributions from large energies, we need 
the behavior of the Jost solutions for $\lambda \to \infty$.
As usual, we write $m_\pm(x,\lambda)=e^{\mp i \lambda x}f_\pm(x,\lambda)$.

\begin{lemma}
\label{lem_Jostinf}
Let $\lambda_0>0$. Then, for $k,m \in \mathbb{N}_0$, 
the function $m_+(\cdot,\lambda)$ satisfies the estimates
$$ |\partial_x^k \partial_\lambda^m (m_+(x,\lambda)-1)| \leq C_{k,m}\langle x
\rangle^{-1-k} \lambda^{-1-m} $$
for all $\lambda \geq \lambda_0$ and all $x \geq 0$.
The same bounds hold for $m_-(x,\lambda)$ if $x \leq 0$.
\end{lemma}

\begin{proof}
As already discussed (see Lemma \ref{lem_Jost}), the function
$m_+(\cdot,\lambda)$ satisfies the Volterra equation
\begin{align*}
m_+(x,\lambda)&=1+\frac{1}{2i\lambda}
\int_x^\infty \left ( e^{2i\lambda(y-x)}-1 \right )
V_{\ell,\sigma}(y)m_+(y,\lambda)dy \\
&=1+\int_0^\infty \left ( e^{2i\eta}-1 \right ) \left [V_{\ell,\sigma}
\left ( \frac{\eta}{\lambda}+x \right )
\frac{m_+ \left ( \frac{\eta}{\lambda}+x,\lambda \right)}
{2i\lambda^2} \right ]d\eta
\end{align*}
and thus, the Lemma is obviously true for $k=m=0$.
Let $n: \mathbb{N}_0 \times \mathbb{N}_0 \to \mathbb{N}_0$ denote the bijection
from Proposition \ref{prop_pertenergy}. 
Fix $(k,m) \in \mathbb{N}_0 \times \mathbb{N}_0$ and 
suppose the assertion is true for all 
$(j,l)$ with $n(j,l)\leq n(k,m)$.
We need to show that this implies the claim for $(k',m')$ where
$n(k',m')=n(k,m)+1$.
There are two possibilities: Either $(k',m')=(m+1,0)$ (if $k=0$) or
$(k',m')=(k-1,m+1)$. 
In the former case we have
\begin{align*}
\partial_x^{m+1}m_+(x,\lambda)&=\frac{1}{2i\lambda^2}\sum_{j=0}^m  
\left ( \begin{array}{c}m+1 \\ j\end{array} \right )
\int_0^\infty \left ( e^{2i\eta}-1 \right )
\partial_x^{m+1-j}V_{\ell,\sigma}\left (\frac{\eta}{\lambda}+x \right )
\partial_x^j m_+\left (\frac{\eta}{\lambda}+x,\lambda \right )
d\eta \\
&+\frac{1}{2i\lambda^2}\int_0^\infty \left ( e^{2i \eta}-1 \right )
V_{\ell,\sigma}\left (\frac{\eta}{\lambda}+x \right )
\partial_x^{m+1} m_+\left (\frac{\eta}{\lambda}+x,\lambda \right )d\eta \\
&=O_\mathbb{C}(\langle x \rangle^{-1-(m+1)}\lambda^{-1})+
\frac{1}{2i\lambda}\int_x^\infty \left ( e^{2i\lambda (y-x)}-1 \right )
V_{\ell,\sigma}(y)\partial_y^{m+1}m_+(y,\lambda)dy
\end{align*}
by assumption and thus, the estimate from Lemma \ref{lem_volterra} yields
$|\partial_x^{m+1}m_+(x,\lambda)|\lesssim \langle x
\rangle^{-1-(m+1)}\lambda^{-1}$.
For the latter case observe that 
$$  \partial_\lambda \partial_x^j
m_+\left (\frac{\eta}{\lambda}+x,\lambda \right )
\lesssim \left \langle 
\frac{\eta}{\lambda}+x \right \rangle ^{-1-j-1}\frac{\eta}{\lambda}\lambda^{-1}+
\left \langle 
\frac{\eta}{\lambda}+x \right \rangle^{-1-j}\lambda^{-2}\lesssim \left \langle 
\frac{\eta}{\lambda}+x \right \rangle^{-1-j}\lambda^{-1} $$
and, more generally \footnote{One may apply Fa\`a di Bruno's formula (cf. Lemma
\ref{lem_symbolinverse}) to obtain a completely explicit expression for the
higher $\lambda$--derivatives.},
$$ \partial_\lambda^l \partial_x^j 
m_+\left (\frac{\eta}{\lambda}+x,\lambda \right )
\lesssim \left \langle 
\frac{\eta}{\lambda}+x \right \rangle ^{-1-l-j}\frac{\eta^l}{\lambda^l}
\lambda^{-l}+\dots+
\left \langle 
\frac{\eta}{\lambda}+x \right \rangle^{-1-j}\lambda^{-1-l}
\lesssim \left \langle 
\frac{\eta}{\lambda}+x \right \rangle^{-1-j}\lambda^{-l} $$
for all $(j,l)$ with $n(j,l)\leq n(k,m)$ and $(j,l) \not= (0,0)$ by assumption.
This shows that
\begin{align*}
&\partial_x^{k-1} \partial_\lambda^{m+1} m_+(x,\lambda) \\ 
&=\sum_{j=0}^m \left ( \begin{array}{c}m+1\\j\end{array}\right )
\int_0^\infty \left ( e^{2i\eta}-1 \right ) 
\partial_x^{k-1} \left [\partial_\lambda^{m+1-j} 
\frac{V_{\ell,\sigma}\left ( \frac{\eta}{\lambda}+x \right )}{2i\lambda^2}
\partial_\lambda^j
m_+ \left ( \frac{\eta}{\lambda}+x,\lambda \right)
\right ]d\eta \\
&+\int_0^\infty \left ( e^{2i\eta}-1 \right )
\partial_x^{k-1} \left [
\frac{V_{\ell,\sigma}\left ( \frac{\eta}{\lambda}+x \right )}{2i\lambda^2}
\partial_\lambda^{m+1}
m_+ \left ( \frac{\eta}{\lambda}+x,\lambda \right)
\right ]d\eta \\
&=O_\mathbb{C}(\langle x \rangle^{-1-(k-1)}\lambda^{-1-(m+1)}) \\
&+\frac{1}{2i\lambda}
\int_x^\infty  
\left ( e^{2i\lambda(y-x)}-1 \right )
V_{\ell,\sigma}(y)
\partial_y^{k-1} \partial_\lambda^{m+1}
m_+(y,\lambda) dy
\end{align*}
and Lemma \ref{lem_volterra} yields the claim.
The proof for $m_-$ is (\emph{mutatis mutandis}) identical.
\end{proof}

\begin{corollary}
\label{cor_Winf}
Let $\lambda_0 >0$. Then the Wronskian $W(f_-,f_+)(\lambda)$ has the behavior
$$ \frac{1}{W(f_-,f_+)(\lambda)}=\frac{1}{2i\lambda}
(1+O_\mathbb{C}(\lambda^{-1})) $$
for all $\lambda \geq \lambda_0$ where the
$O$--term behaves like a symbol.
\end{corollary}

\begin{proof}
With $m_\pm(x,\lambda)=e^{\mp i\lambda x}f_\pm(x,\lambda)$ we have
$$ W(f_-,f_+)(\lambda)=2i\lambda m_-(0,\lambda)m_+(0,\lambda)+
O_\mathbb{C}(\lambda^{-1})=2i\lambda(1+O_\mathbb{C}(\lambda^{-1})) $$
by Lemma \ref{lem_Jostinf} where the $O$--term is of symbol type.
Thus, the claim follows from Lemma \ref{lem_symbolinverse}.
\end{proof}

Before proceeding to the final oscillatory integral estimate, we need the large
$\lambda$ behavior of the reflection and transmission coefficients, i.e.,  the
coefficients $a(\lambda)$ and $b(\lambda)$ satisfying 
$f_-(x,\lambda)=a(\lambda)f_+(x,\lambda)+b(\lambda)\overline{f_+(x,\lambda)}$. 
The behavior of $b$ is given by Eq.~(\ref{eq_b}).
By Lemma \ref{lem_Jostinf} we have
$$ W(f_-,\overline{f_+})(\lambda)=m_-(0,\lambda)\overline{m_+'(0,\lambda)}
-m_-'(0,\lambda)\overline{m_+(0,\lambda)}=O_\mathbb{C}(\lambda^{-1}) $$
where the $O$--term behaves like a symbol and therefore,
\begin{equation}
\label{eq_ainf}
\frac{a(\lambda)}{W(f_-,f_+)(\lambda)}=-\frac{W(f_-,\overline{f_+})(\lambda)}
{2i\lambda W(f_-,f_+)(\lambda)}=O_\mathbb{C}(\lambda^{-3})
\end{equation}
We also remark that, by symmetry, the above considerations extend to large
negative $\lambda$. 
Now we are ready to prove the final oscillatory integral estimate.

\subsection{The cosine estimate}

We distinguish between the cosine and the sine estimate since in the former 
case
we obtain a bound involving the derivative of the data.

\begin{proposition}
\label{prop_cos}
Let $\alpha \in \mathbb{N}_0$ and $\delta>0$ sufficiently small. 
Then we have the
estimate
\begin{align*} 
\sup_{x \in \mathbb{R}} \left |
\lim_{N \to \infty} \int_\mathbb{R} \int_{-N}^N 
\lambda e^{\pm it\lambda}G_{\ell,\sigma}(x,x',\lambda)
(1-\chi_\delta(\lambda))
\langle x \rangle^{-\alpha} \langle x' \rangle ^{-\alpha}\phi(x')dx'd\lambda 
\right | \\
\lesssim \langle t \rangle^{-\alpha}\int_\mathbb{R} \left (
|\phi'(x')|+|\phi(x')|\right )dx'
\end{align*}
for all $t \geq 0$ and any $\phi \in \mc{S}(\mathbb{R})$.
\end{proposition}

\begin{proof}
We split the integral according to
$$ \int_\mathbb{R} \int_{-N}^N \dots d\lambda dx'=\int_{-\infty}^x \int_{-N}^N
\dots d\lambda dx' + 
\int_{x}^\infty \int_{-N}^N
\dots d\lambda dx' $$
and only consider the first summand since the proof for the second one is
completely analogous.
In the domain $x' \leq x$, which we study now, the Green's function is given by
$$
G_{\ell,\sigma}(x,x',\lambda)=\frac{f_-(x',\lambda)f_+(x,\lambda)}
{W(f_-,f_+)(\lambda)}. $$
We distinguish between $x \leq 0$ and $x \geq 0$ and start with $x \leq 0$.
Using reflection and transmission
coefficients we obtain
\begin{align*}
f_-(x',\lambda)f_+(x,\lambda)&=e^{-i\lambda x'}m_-(x',\lambda)\left 
[-\overline{a(\lambda)}e^{-i\lambda x}m_-(x,\lambda)+b(\lambda)e^{i\lambda
x}\overline{m_-(x,\lambda)} \right ] \\
&=-e^{-i\lambda (x+x')}\overline{a(\lambda)}m_-(x',\lambda)m_-(x,\lambda)
+e^{i\lambda (x-x')}b(\lambda)m_-(x',\lambda)\overline{m_-(x,\lambda)},
\end{align*}
and consider each term separately.
We define
$$ \omega(x,x',\lambda):=
\frac{\lambda b(\lambda)(1-\chi_\delta(\lambda))}{W(f_-,f_+)(\lambda)}
m_-(x',\lambda)\overline{m_-(x,\lambda)} $$
and by Lemma \ref{lem_Jostinf} and Corollary \ref{cor_Winf} as well as Eqs.
(\ref{eq_b}), (\ref{eq_ainf}), we obtain the
estimates 
$$|\partial_\lambda^m \omega(x,x',\lambda)|\leq C_m |\lambda|^{-m}$$
for all $|\lambda|\geq \frac{\delta}{2}$ and
$x,x' \leq 0$ (recall that $1-\chi_\delta(\lambda)\equiv 1$ for
$|\lambda|\geq \delta$).
Note that, by Fubini, we can freely interchange the order of integration and
thus,
integration by parts with respect to $x'$ yields
\begin{align}
\label{eq_proofoscinf}
\int_{-\infty}^x \int_{-N}^N e^{i\lambda(\pm t+x-x')} 
\omega(x,x',\lambda)\langle x \rangle^{-\alpha} \langle x' \rangle ^{-\alpha}
\phi(x')
d\lambda dx'
= -\phi(x)\langle x \rangle^{-2\alpha}
\int_{-N}^N \frac{e^{\pm i\lambda t}}{i\lambda}
\omega(x,x,\lambda)d\lambda \\
+ \int_{-\infty}^x \int_{-N}^N \frac{e^{i\lambda(\pm t+x-x')}}{i\lambda}
\langle x \rangle^{-\alpha}\partial_{x'}\left [
\omega(x,x',\lambda) 
\langle x' \rangle^{-\alpha}\phi(x') \right ]
d\lambda dx' . \nonumber
\end{align}
We first claim that 
\begin{equation}
\label{eq_proofoscinfclaim1}
\sup_{x < 0}\left | \int_\mathbb{R}e^{\pm i\lambda t}\lambda^{-1}
\omega(x,x,\lambda)d\lambda \right | \lesssim 1. 
\end{equation}
Indeed, set 
$$ \tilde{\omega}(x,\lambda):=\lambda^{-1}\omega(x,x,\lambda)=
\frac{b(\lambda)(1-\chi_\delta(\lambda))}{W(f_-,f_+)(\lambda)}
m_-(x,\lambda)\overline{m_-(x,\lambda)} $$
and observe that $\tilde{\omega}(x,\lambda)=\frac{1}{2i\lambda}
(1+O_\mathbb{C}(|\lambda|^{-1}))$ for $|\lambda| \to \infty$ 
by Eq.~(\ref{eq_b}) and Lemma
\ref{lem_Jostinf}.
This shows that $\tilde{\omega}(x,-\lambda)=-\tilde{\omega}(x,\lambda)
+O_\mathbb{C}(|\lambda|^{-2})$.
Thus, we have
\begin{align*} 
\left | \int_\mathbb{R}e^{\pm i\lambda t}
\tilde{\omega}(x,\lambda)d\lambda \right | & =\left | 
\int_\frac{\delta}{2}^\infty \left [ 
e^{\pm i\lambda t}\tilde{\omega}(x,\lambda)+e^{\mp i\lambda
t}\tilde{\omega}(x,-\lambda) \right ]d\lambda \right | \\
&\lesssim \left | \int_\frac{\delta}{2}^\infty \sin(\lambda t)
\tilde{\omega}(x,\lambda)
d\lambda \right |+\left |\int_\frac{\delta}{2}^\infty e^{\pm i\lambda
t}O_\mathbb{C}(|\lambda|^{-2})d\lambda \right | \\
& \lesssim
\left |\int_\frac{\delta}{2}^\infty \sin(\lambda
t)(\lambda^{-1}+O_\mathbb{C}(|\lambda|^{-2}))d\lambda \right |+1 \lesssim 1
\end{align*}
for all $x \leq 0$ and this proves Eq.~(\ref{eq_proofoscinfclaim1}).
Therefore, for $N \to \infty$, we can estimate the first term in 
Eq.~(\ref{eq_proofoscinf}) as
\begin{align*} \left | \phi(x)\langle x \rangle^{-2\alpha}
\int_\mathbb{R} e^{\pm i\lambda t} 
\lambda^{-1}\omega(x,x,\lambda)d\lambda \right |& \lesssim 
\sup_{x < 0}|\phi(x)|\: \frac{1}{t^\alpha} \int_\mathbb{R} \left |e^{\pm i\lambda t}\partial_\lambda^\alpha \left [\lambda^{-1}\omega(x,x,\lambda)\right ] \right | d\lambda \\
&\lesssim
\langle t
\rangle^{-\alpha}\int_\mathbb{R}\left (|\phi'(x')|+|\phi(x')|\right )dx' 
\end{align*}
for all $t\geq 1$ and $x \leq 0$ by $\alpha$--fold integration by parts and Sobolev
embedding.
By Eq.~\eqref{eq_proofoscinfclaim1} this inequality is in fact valid for all $t \geq 0$.

For the second term we similarly claim that
\begin{align}
\label{eq_proofoscinfclaim2}
\sup_{x<0}\int_{-\infty}^x \left | \int_\mathbb{R}
e^{i\lambda(\pm t+x-x')}\lambda^{-1}
\langle x \rangle^{-\alpha}\partial_{x'}\left [
\omega(x,x',\lambda) 
\langle x' \rangle^{-\alpha}\phi(x') \right ]
d\lambda \right |dx' \\
\lesssim \langle x \rangle^{-\alpha}\int_\mathbb{R}
\langle x' \rangle^{-\alpha}\left (|\phi'(x')|+|\phi(x')|\right )dx'. \nonumber
\end{align}
Indeed, we have
\begin{align*}
\left | \int_\mathbb{R}
e^{i\lambda(\pm t+x-x')}\lambda^{-1}
\langle x \rangle^{-\alpha}\partial_{x'}
\omega(x,x',\lambda) 
\langle x' \rangle^{-\alpha}\phi(x') 
d\lambda \right | &\leq 
\langle x \rangle^{-\alpha}\langle x' \rangle^{-\alpha}
\int_\mathbb{R} \left |\lambda^{-1}
\partial_{x'}
\omega(x,x',\lambda) 
\phi(x') \right | d\lambda \\
& \lesssim \langle x \rangle^{-\alpha}\langle x' \rangle^{-\alpha}|\phi(x')| 
\end{align*}
for all $x \leq 0$ since $|\lambda^{-1}
\partial_{x'}
\omega(x,x',\lambda)| \lesssim |\lambda|^{-2}$ by Lemma \ref{lem_Jostinf}.
Moreover,
$$ \left | \int_\mathbb{R}
e^{i\lambda(\pm t+x-x')}\lambda^{-1}
\omega(x,x',\lambda) 
\langle x \rangle^{-\alpha}\partial_{x'}\left [
\langle x' \rangle^{-\alpha}\phi(x') \right ]
d\lambda \right | \lesssim \langle x \rangle^{-\alpha}
\langle x' \rangle^{-\alpha} \left (|\phi'(x')|+|\phi(x')| \right ) $$
for all $x \leq 0$ which can be shown by exploiting exactly the same cancellation
that led to Eq.~(\ref{eq_proofoscinfclaim1}).
This proves Eq.~(\ref{eq_proofoscinfclaim2}).
Note in particular that Eq.~(\ref{eq_proofoscinfclaim2}) implies
\begin{align*}
\lim_{N \to \infty}\int_{-\infty}^x \int_{-N}^N \frac{e^{i\lambda(\pm t+x-x')}}{i\lambda}
\langle x \rangle^{-\alpha}\partial_{x'}\left [
\omega(x,x',\lambda) 
\langle x' \rangle^{-\alpha}\phi(x') \right ]
d\lambda dx' \\
=\int_{-\infty}^x \int_\mathbb{R} \frac{e^{i\lambda(\pm t+x-x')}}{i\lambda}
\langle x \rangle^{-\alpha}\partial_{x'}\left [
\omega(x,x',\lambda) 
\langle x' \rangle^{-\alpha}\phi(x') \right ]
d\lambda dx' 
\end{align*}
by dominated convergence.
Now we distinguish two cases.
If $|\pm t+x-x'|\geq \frac{1}{2}t$, we integrate by parts $\alpha$--times to
obtain
\begin{align*}
&\left | \int_{-\infty}^x \int_\mathbb{R} e^{i\lambda(\pm t+x-x')} 
\lambda^{-1}
\langle x \rangle^{-\alpha}\partial_{x'}\left [
\omega(x,x',\lambda)\langle x' \rangle^{-\alpha} 
\phi(x') \right ]
d\lambda dx'\right | \\
&\lesssim |\pm t+x-x'|^{-\alpha}\int_{-\infty}^x \int_\mathbb{R}
 \left |\langle x \rangle^{-\alpha}
 \partial_{x'}\partial_\lambda^{\alpha}\left [\lambda^{-1}
\omega(x,x',\lambda)\langle x' \rangle^{-\alpha} 
 \phi(x') \right ] \right | d\lambda dx' \\
&\lesssim \langle t \rangle^{-\alpha}
\int_\mathbb{R} \int_{|\lambda|\geq \frac{\delta}{2}}
|\lambda|^{-(\alpha+1)} \langle x \rangle^{-\alpha} \left |
\partial_{x'} \left [ 
\langle x' \rangle^{-\alpha}\phi(x') \right ] \right | d\lambda dx' \\
&\lesssim \langle t \rangle^{-\alpha}\int_\mathbb{R}
\left (|\phi'(x')|+|\phi(x')|\right )dx'
\end{align*}
for all $t \geq 0$ and all $x \leq 0$.
If $|\pm t+x-x'| \leq \frac{1}{2}t$, we have $\langle x \rangle^{-\alpha}\langle x'
\rangle^{-\alpha} \lesssim \langle t \rangle^{-\alpha}$ as $t \to \infty$ and
Eq.~(\ref{eq_proofoscinfclaim2}) implies
\begin{align*}
\sup_{x<0}\int_{-\infty}^x \left | \int_\mathbb{R}
e^{i\lambda(\pm t+x-x')}\lambda^{-1}
\langle x \rangle^{-\alpha}\partial_{x'}\left [
\omega(x,x',\lambda) 
\langle x' \rangle^{-\alpha}\phi(x') \right ]
d\lambda \right |dx' \\
\lesssim \langle t \rangle^{-\alpha}\int_\mathbb{R}
\left (|\phi'(x')|+|\phi(x')|\right )dx'. \nonumber
\end{align*}
for all $t\geq 0$. The remaining cases are treated in a completely analogous
fashion.
Note that the terms involving the coefficient $a(\lambda)$ are even simpler
due to the stronger
decay given by Eq.~(\ref{eq_ainf}).
For terms that contain no reflection and transmission coefficients, use
Corollary \ref{cor_Winf} for the cancellation argument.
\end{proof}

\subsection{The sine estimate}
The sine estimate is slightly stronger since it does not require derivatives of
the data.

\begin{corollary}
\label{prop_sin}
Let $\alpha \in \mathbb{N}_0$ and $\delta>0$ sufficiently small. 
Then we have the
estimate
\begin{align*} 
\sup_{x \in \mathbb{R}} \left |
\lim_{N \to \infty} \int_\mathbb{R} \int_{-N}^N 
e^{\pm it\lambda}G_{\ell,\sigma}(x,x',\lambda)
(1-\chi_\delta(\lambda))
\langle x \rangle^{-\alpha} \langle x' \rangle ^{-\alpha}\phi(x')dx'd\lambda 
\right | \\
\lesssim \langle t \rangle^{-\alpha}\int_\mathbb{R} |\phi(x')| dx'
\end{align*}
for all $t \geq 0$ and any $\phi \in \mc{S}(\mathbb{R})$.
\end{corollary}

\begin{proof}
Just repeat the arguments from the proof of
Proposition \ref{prop_cos}.
However, note that we are lacking one factor of $\lambda$ compared to
Proposition \ref{prop_cos} which makes the integration by parts with respect to
$x'$ unnecessary. This explains why no term containing $\phi'$ appears on the
right--hand side of the estimate.
\end{proof}


\begin{appendix}

\section{Symbol behavior}

\begin{lemma}
\label{lem_symbolinverse}
Let $I \subset \mathbb{R}$ and suppose $f$ is smooth on $I$ and satisfies 
$|f(x)|\leq C < 1$ for all $x \in I$.
Then, for all $x \in I$, we have the estimate
$$ \left |\left ( \frac{1}{1+f} \right )^{(k)}(x) \right |\leq C_k 
\sum \prod_{j=1}^k
\left | f^{(j)}(x) \right |^{m_j} $$
for all $k \in \mathbb{N}$ where the sum runs over all possible $k$--tuples
$(m_1,m_2,\dots,m_k) \in \mathbb{N}_0^k$ satisfying $\sum_{j=1}^k jm_j=k$. 
\end{lemma}

\begin{proof}
This follows from the identity 
$$ \left (\frac{1}{1+f}\right )^{(k)}=\sum_{\sum_{j=1}^k jm_j=k}
a_{m_1,m_2,\dots,m_k}\left (\frac{1}{1+f} \right )^{1+\sum_{j=1}^k m_j}
\prod_{j=1}^k
\left ( f^{(j)} \right )^{m_j} $$
which is known as Fa\` a di Bruno's formula (see e.g.~\cite{Roman}, 
the explicit form of the
coefficients $a_{m_1,m_2,\dots,m_k}$ is irrelevant for our purposes) and the
fact that $|(1+f)^{-1}| \lesssim 1$ on $I$.
\end{proof}

\section{Volterra integral equations}
In this section we establish some well--known facts about Volterra integral
equations which are frequently used throughout this work.

\begin{lemma}
\label{lem_volterra}
Let $a \in \mathbb{R}$, $g \in L^\infty(a,\infty)$ and suppose the integral kernel $K$ satisfies
$$ \mu:=\int_a^\infty \sup_{x \in (a,y)}|K(x,y)|dy < \infty. $$
Then the Volterra equation
$$ f(x)=g(x)+\int_x^\infty K(x,y)f(y)dy $$
has a unique solution $f$ satisfying 
$$ \|f\|_{L^\infty(a,\infty)} \leq e^\mu \|g\|_{L^\infty(a,\infty)}. $$
\end{lemma}

\begin{proof}
See e.g. \cite{deift} or \cite{schlag1}.
\end{proof}

The next lemma states differentiability properties of solutions of Volterra
integral equations.

\begin{lemma}
\label{lem_volterradiffx}
If, in addition to the assumptions of Lemma \ref{lem_volterra}, $g
\in C^\infty(a,\infty)$ and
the kernel $K$ is smooth in both variables on $(a,\infty)$ and satisfies
$$ \int_a^\infty |\partial^k_x K(x,y)|dy < \infty $$
for any $x \geq a$ and all $k \in \mathbb{N}$ then the solution 
$f$ is smooth on $(a,\infty)$.
Furthermore, the derivatives can be calculated 
by formal differentiation, i.e., 
$$ f^{(k)}(x)=g^{(k)}(x)-\sum_{j=0}^{k-1} (\kappa_j f)^{(k-1-j)}(x)
+\int_x^\infty \partial_x^k K(x,y)f(y)dy $$
where $\kappa_j(x):=\partial_x^j K(x,y)|_{y=x}$.
\end{lemma}

\begin{proof}
The claim follows from a straightforward application of Lebesgue's theorem 
on dominated convergence and an induction.
\end{proof}

The next lemma shows how the dependence of the kernel $K$ on a parameter 
$\lambda$ carries over to the solution of the corresponding Volterra equation.

\begin{lemma}
\label{lem_volterradifflambda}
Let $I \subset \mathbb{R}$ be open and suppose
$$ \int_a^\infty \sup_{x \in (a,y)}|\partial_\lambda^m 
K(x,y,\lambda)|dy < \infty $$
as well as $\partial_\lambda^m g(\cdot,\lambda) \in L^\infty(a,\infty)$
for all $m \in \mathbb{N}_0$ and $\lambda \in I$.
Then the Volterra equation
$$ f(x,\lambda)=g(x,\lambda)+\int_x^\infty K(x,y,\lambda)f(y,\lambda)dy $$
has a unique solution $f(x,\lambda)$ for all $x\geq a$ and $\lambda \in I$
which is smooth in
$\lambda$.
Furthermore, we have $\partial_\lambda^m f(\cdot,\lambda) \in
L^\infty(a,\infty)$ for all $m \in \mathbb{N}_0$ and the derivatives are 
given by
$$ \partial_\lambda^m f(x,\lambda)=\partial_\lambda^m g(x,\lambda)+
\sum_{j=0}^m \left ( \begin{array}{c}m \\
j \end{array} \right )\int_x^\infty \partial_\lambda^j
K(x,y,\lambda)\partial_\lambda^{m-j}f(y,\lambda)dy. $$
\end{lemma}

\begin{proof}
According to Lemma \ref{lem_volterra}, the solution $f$ exists, is unique and 
satisfies $\|f(\cdot,\lambda)\|_{L^\infty(a,\infty)}< \infty$.
Now consider the integral equation
\begin{equation}
\label{eq_proofparamvolt}
h(x,\lambda)=\partial_\lambda g(x,\lambda)+
\tilde{h}(x, \lambda)+\int_x^\infty K(x,y,\lambda)h(y,\lambda)dy 
\end{equation}
where
$$ \tilde{h}(x,\lambda):=\int_x^\infty \partial_\lambda
K(x,y,\lambda)f(y,\lambda)dy. $$
We have $\|\tilde{h}(\cdot, \lambda)\|_{L^\infty(a,\infty)} < \infty$ and thus, 
by Lemma \ref{lem_volterra}, Eq.~(\ref{eq_proofparamvolt}) has a unique solution
$h(\cdot,\lambda) \in L^\infty(a,\infty)$ for all $\lambda \in I$.
However, by dominated convergence we conclude 
$$ \lim_{\nu \to 0}\left |
\frac{f(x,\lambda+\nu)-f(x,\lambda)}{\nu}-h(x,\lambda) \right |=0 $$
and hence, $\partial_\lambda f$ exists and equals $h$.
Existence of the higher derivatives follows by the Leibniz rule and an 
induction.
\end{proof}

We finally remark that all of the above Lemmas have counterparts for Volterra
equations of the form
$$ f(x,\lambda)=g(x,\lambda)+\int_a^x K(x,y,\lambda)f(y,\lambda)dy $$
with almost identical proofs.

\end{appendix}

\bibliography{wavess}
\bibliographystyle{plain}

\end{document}